\documentclass[conference]{IEEEtran}
\IEEEoverridecommandlockouts
\usepackage{multirow}
\usepackage{colortbl}

\usepackage{tikz}
\usetikzlibrary{arrows.meta,bending}

\usepackage{fancyhdr}
\usepackage{enumitem}
\usepackage{cite}
\usepackage{amsmath,amssymb,amsfonts}
\usepackage{graphicx}
\usepackage{verbatim}
\usepackage{float}
\usepackage{caption}
\usepackage{amsmath,amsthm,amsfonts,amssymb,amscd}
\usepackage{balance}
\usepackage{subfigure}
\usepackage{enumerate}
\usepackage{fancyhdr}
\usepackage{mathrsfs}
\usepackage{xcolor}
\usepackage{graphicx}
\usepackage{listings}
\usepackage{float}
\usepackage{textcomp}
\usepackage{xcolor}
\usepackage{amsthm,amsmath,amssymb}
\usepackage{mathrsfs}
\usepackage[hidelinks]{hyperref}
\usepackage{ragged2e}

\newcommand{\eg}{\textit{e.g.}}
\newcommand{\ie}{\textit{i.e.}}
\newcommand{\wrt}{\textit{w.r.t.}}

\usepackage{enumitem}
\usepackage[utf8]{inputenc}
\usepackage{graphicx}
\usepackage{subfigure}
\usepackage{lipsum}
\usepackage{amsmath,bm}
\usepackage{amssymb}
\usepackage{indentfirst}
\usepackage{amsbsy}
\usepackage{hyperref}
\usepackage{cleveref}
\usepackage{mathrsfs}
\usepackage{amsthm}
\usepackage{tabularx}
\usepackage{enumitem}
\usepackage{wrapfig}
\usepackage{multirow}
\makeatletter 
\usepackage[ruled, vlined, linesnumbered]{algorithm2e}
\usepackage{algpseudocode}

\newtheorem{theorem}{Theorem}
\newtheorem{lemma}[theorem]{Lemma}
\newtheorem{definition}{Definition}

\newtheorem{corollary}[theorem]{Corollary}
\newtheorem*{remark}{Remark}

\newcommand{\ourname}{MagnifierSketch}
\newcommand{\stageone}{Tower Sketch}
\newcommand{\stagetwo}{Value Sketch}

\definecolor{keywordcolor}{rgb}{0.0, 0.0, 0.6}
\definecolor{commentcolor}{rgb}{0.0, 0.5, 0.0}
\definecolor{stringcolor}{rgb}{0.58, 0.0, 0.82}

\lstdefinelanguage{SQL}{
  keywords={SELECT, FROM, WHERE, AND, OR, NOT, INSERT, INTO, VALUES, UPDATE, SET, DELETE, CREATE, TABLE, PRIMARY, KEY, FOREIGN, DROP, ALTER, ADD, AS, DISTINCT, GROUP, BY, ORDER, HAVING, JOIN, ON, IN, IS, NULL},
  sensitive=false,
  comment=[l]--,
  morestring=[b]',
  morestring=[b]",
}

\newcommand{\st}{\textit{s.t.}}

\begin{document}

\title{\ourname{}: Quantile Estimation Centered at One Point}

\author{
\IEEEauthorblockN{
Jiarui Guo\IEEEauthorrefmark{1}, 
Qiushi Lyu\IEEEauthorrefmark{1}, 
Yuhan Wu\IEEEauthorrefmark{1},
Haoyu Li\IEEEauthorrefmark{2}, 
Zhaoqian Yao\IEEEauthorrefmark{1}, \\
Yuqi Dong\IEEEauthorrefmark{1},
Xiaolin Wang\IEEEauthorrefmark{1}, 
Bin Cui\IEEEauthorrefmark{1}, 
and Tong Yang\IEEEauthorrefmark{1}
}

\IEEEauthorblockA{
\IEEEauthorrefmark{1}\textit{Peking University} 
\IEEEauthorrefmark{2}\textit{University of Texas
at Austin} 
}
}

\setlength{\subfigcapskip}{-0.15cm}
\setlength{\subfigbottomskip}{-0.05cm}

\maketitle

\pagestyle{plain}

\begin{abstract}
In this paper, we take into consideration quantile estimation in data stream models, where every item in the data stream is a key-value pair. 
Researchers sometimes aim to estimate per-key quantiles (\ie{} quantile estimation for every distinct key), and some popular use cases, such as tail latency measurement, recline on a predefined single quantile (\eg{} 0.95- or 0.99-quantile) rather than demanding arbitrary quantile estimation. 
However, existing algorithms are not specially designed for per-key estimation centered at one point. 
They cannot achieve high accuracy in our problem setting, and their throughput are not satisfactory to handle high-speed items in data streams. 
To solve this problem, we propose \ourname{} for point-quantile estimation. 
\ourname{} supports both single-key and per-key quantile estimation, and its key techniques are named \textbf{Value Focus}, \textbf{Distribution Calibration} and \textbf{Double Filtration}. 
We provide strict mathematical derivations to prove the unbiasedness of \ourname{} and show its space and time complexity. 
Our experimental results show that the Average Error (AE) of \ourname{} is significantly lower than the state-of-the-art in both single-key and per-key situations. 
We also implement \ourname{} on RocksDB database to reduce quantile query latency in real databases. 
All related codes of \ourname{} are open-sourced and available at GitHub.

\end{abstract}

\section{Introduction}
\label{sec:intro}

\subsection{Background and Motivation}
Quantile estimation is an important topic in many fields. In the field of statistics, researchers can have a better knowledge of the unknown parameters in the probability distribution function (pdf) by quantile estimation, which helps predict further sampling in the future \cite{parrish1990comparison}. In theoretical computer science, researchers make efforts to design elegant algorithms to estimate the quantile of a multiset with size $N$ within $o(N)$ space complexity \cite{munro1980selection, manku1998approximate, greenwald2001space, karnin2016optimal}. Quantile estimation also plays an important role in application, such as databases \cite{chen2000incremental}, meteorology \cite{timofeev2010using}, and finance \cite{ferrari2009maximum}. 

When it comes to data stream models, it is important to estimate the quantile of distribution. Yet, few work targets at \textit{per-key} quantile estimation. 
In a data stream, every item is a key-value pair $(k, v)$, and we have to independently estimating quantile for every distinct key rather than viewing different keys as a whole. With the above preliminaries, single-key and per-key quantile estimation can be respectively summarized as follows: 

\begin{lstlisting}
    -- Single-key quantile estimation
    SELECT Quantile(value, w)
    FROM DataStream
    -- Per-key quantile estimation
    SELECT Quantile(item.value, w)
    FROM DataStream
    WHERE item.key = key
\end{lstlisting}
where the function \texttt{Quantile(s, w)} returns the $w$-quantile of the multiset $s$. 
In fact, per-key quantile is more useful in many situations, and below we show some use cases.  

\textbf{Case 1: Latency Measurement.} Latency is an important topic in network scenarios \cite{padmanabhan1995improving, rumble2011s}. It is defined as the delay in network communication, or the time from request to response. Users usually expect low latency when they are visiting a website, but the overall low latency cannot guarantee low latency for every user. Therefore, to attract as many visitors as possible, website administrators shall make efforts to achieve low latency quantile for every visitor \cite{flach2013reducing, arapakis2014impact}. 
Also, the sudden increase of latency can imply potential cyber attacks or offending activities \cite{chen2020system, shahzad2016accurate}. However, the aggregated latency quantile can remain unchanged, as millions of packets with low latency can conceal this phenomenon. As a result, we can only detect the attack by per-key latency quantile estimation.

\textbf{Case 2: Data Preprocessing.} Data preprocessing is necessary in machine learning models \cite{brownlee2020data, chu2016data}. An effective method for data cleaning is to estimate per-key quantile. Since dirty data with extreme values usually have extremely high (low) percentage quantiles, we can detect and delete these values quickly \cite{hellerstein2008quantitative, dasu2003exploratory}. However, to obtain high-quality datasets for machine learning, traditional single-key algorithms either have to read the dataset multiple times, or have to build a data structure for every distinct key, resulting in a waste of time or space. Designing a per-key quantile estimation algorithm can solve this tough problem efficiently.

Also, while most prior work mainly focuses on all-quantile estimation (\ie{} quantile estimation for arbitrary point), researchers sometimes are more interested in the quantile function at one point. 
For example, network administrators usually pay more attention to tail latency (\eg{} 0.95 or 0.99-quantile of latency) in network scenarios, as many applications must achieve low tail latency to meet business objectives \cite{prekas2017zygos, zhao2023panakos}. 
In addition, the median (\eg{} 0.5-quantile) can be used for estimation of income levels within population \cite{harrison2009median} and analysis of large amounts of transaction data in finance \cite{harris_1987}. 
Finally, constructing $\left\lceil \frac{1}{\varepsilon}\right\rceil $ data structures for point-quantile estimation can solve the problem of all-quantile estimation within at most $\varepsilon$ error, so point-quantile estimation also lays solid foundation for all-quantile estimation. 

However, per-key point-quantile estimation can be challenging and costly in data stream models. We have to keep the information of every item, but the high-speed items in data streams require us to deal with every item in $O(1)$ time. Also, to ensure that the data structure is small enough to be placed on caches, we have to use as little memory as possible. As a result, the goal of this paper is to design a compact sketch algorithm which can accurately estimate per-key point-quantile with small time and space complexity. 

\subsection{Prior Art and Their Limitations}
We divide existing quantile estimation algorithms into two categories: \textit{single-key} algorithms and \textit{per-key} algorithms. Single-key algorithms can be used to estimate \textit{aggregated} quantile, which means they view different keys as a whole and estimate the quantile of all values. 
Many algorithms belong to this category \cite{manku1998approximate, greenwald2001space, karnin2016optimal, ddsketch, cormode2021relative, cormode2021theory, tdigest}. 
However, although attractive in theory, they cannot be directly applied to per-key situation. In data stream models, a large number of items are mixed together, and these algorithms require several copies to realize per-key estimation, which is unacceptable in practice.

Per-key algorithms can be directly used for per-key quantile estimation. SQUAD \cite{shahout2022squad}, SketchPolymer \cite{guo2023sketchpolymer} and M4 \cite{dongm4} are three of these per-key algorithms. These per-key algorithms are mainly designed for all-quantile estimation, so they can be applied for per-key point-quantile estimation as well. Nevertheless, in order to realize all-quantile estimation, these algorithms record too many value samples far from the target quantile, which is a waste of space in point-quantile situation. In addition, these algorithms fail to filter infrequent items efficiently, and infrequent items can have detrimental influence on the accuracy of frequent items. 

\subsection{Our Proposed Solution}
Toward the design goal, we propose \textbf{\ourname{}} for per-key point-quantile estimation. \ourname{} is space-efficient: it is small enough to be placed in CPU caches. \ourname{} is fast: it deals with every item in constant time. \ourname{} is accurate: it achieves much smaller error compared to the state-of-the-art.

\ourname{} includes two stages: Stage 1 is a \stageone{} \cite{towersketch} and Stage 2 is a hash table named
\stagetwo{}. \stageone{} is used to filter infrequent items in advance, and \stagetwo{} keeps samples of value for every frequent item. The techniques used in \stageone{} and \stagetwo{} are named \textbf{Value Focus}, \textbf{Distribution Calibration} and \textbf{Double Filtration}. We now introduce these three techniques below: 

\textbf{Key Technique 1: Value Focus (\S \ref{sec:single})}. In point-quantile query, values close to the query quantile $w$ shall be prioritized. Therefore, we design a special processing mechanism to keep as many values close to the target quantile as possible. \ourname{} maintains two data structures named the Candidate and the Representative in single-key situation. The value of every incoming item will be inserted into the Candidate, and only two median values can be inserted into the Representative after the Candidate is full when $w=0.5$. The Representative further applies replacement strategy to evict unqualified values. In this way, \ourname{} can estimate quantile accurately by reporting the median value in the Representative.

\textbf{Key Technique 2: Distribution Calibration (\S \ref{sec:cali})}. 
We study data streams where the item values remain a stable distribution $F$ in a certain interval \cite{zhou2007distributed, nguyen2015survey}. 
We find that this model is applicable in many scenarios, \eg{} database \cite{gibbons1997fast}, network monitoring \cite{paxson1997end} and commerce \cite{correa2017posted}. 
In addition, the performance of \ourname{} peaks when we set the query quantile $w=0.5$, and merely querying 0.5-quantile is simpler than querying an arbitrary quantile. Consequently, we apply probability method to ``calibrate'' the original distribution $F$ and construct a new distribution $F'$, so that the $w$-quantile of the original distribution $F$ is just the 0.5-quantile of $F'$ (See in Figure \ref{fig::cali}). We prove that the expected quantile function at the query result \wrt{} $F'$ is just 0.5, and \ourname{} still achieves high throughput after Distribution Calibration. 

\begin{figure}[t]
	\centering
	\subfigure[before]{
	\begin{minipage}[t]{0.22\textwidth}{
			\includegraphics[width=1\textwidth]{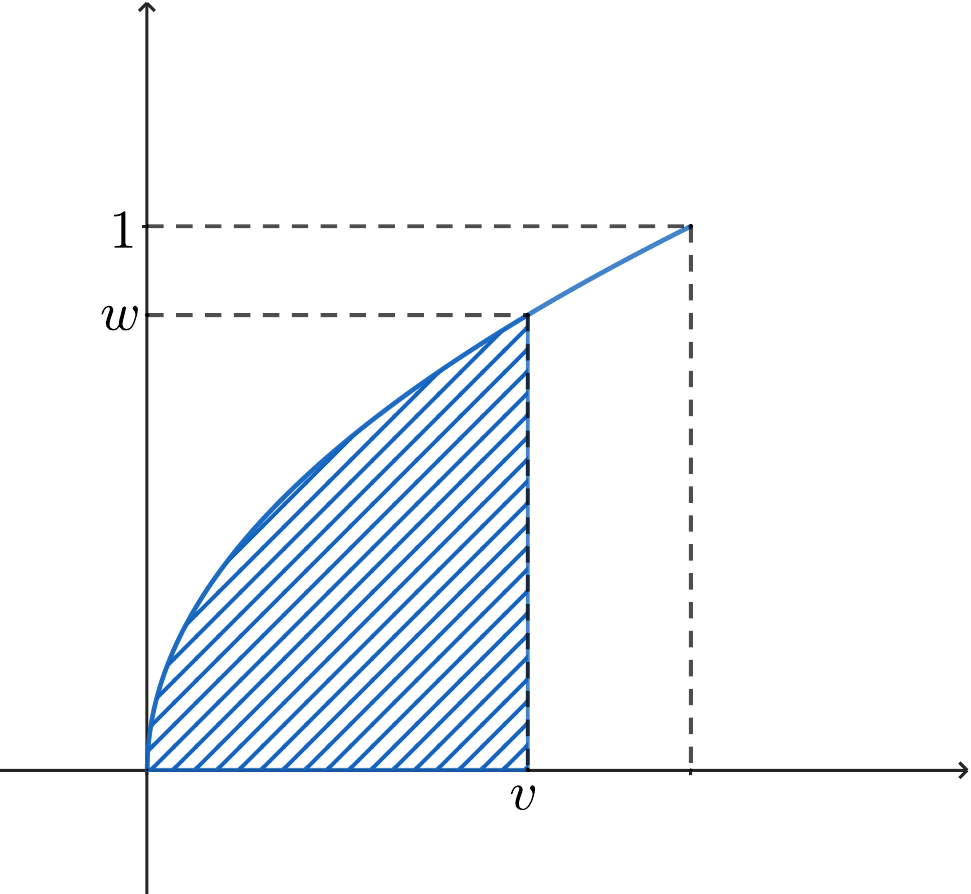}}
	\end{minipage}}
	\subfigure[after]{
	\begin{minipage}[t]{0.22\textwidth}{
		\includegraphics[width=1\textwidth]{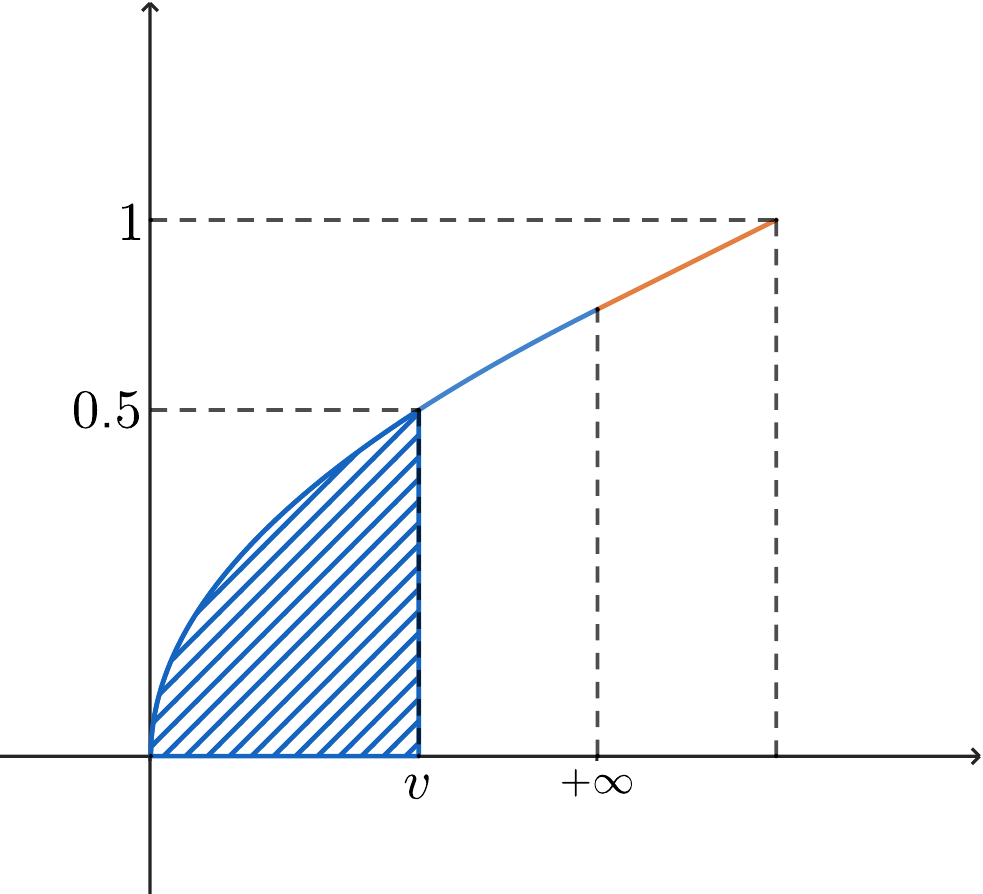}}
	\end{minipage}} \\
	\caption{Idea of Distribution Calibration}
	\label{fig::cali}
	\vspace{-0.2in}
\end{figure}

\textbf{Key Technique 3: Double Filtration (\S \ref{sec:per}-\ref{sec:tower})}. The majority of items in data streams are infrequent items \cite{faloutsos1996modeling, guo2025hourglasssketch}. However, their quantile can hardly be estimated precisely due to their low frequency. As a result, it is necessary to filter infrequent items to save memory for frequent items. In \ourname{}, we filter infrequent items in both \stageone{} (Stage 1) and \stagetwo{} (Stage 2). In \stageone{}, \ourname{} leverages the idea of \textbf{Early Screening}, which originates from Cold Filter \cite{coldfilter}. In \stagetwo{}, \ourname{} borrows the idea of \textbf{Ostracism} 
from Elastic Sketch \cite{yang2018elastic} to evict infrequent items. 
In this way, we drop infrequent items and keep the information of frequent items as accurately as possible, which can be used for quantile estimation.

\subsection{Key Contributions}
This paper makes the following contributions: 
\begin{itemize}[leftmargin=1em]
    \item We propose a two-stage data structure, namely the Candidate and the Representative, to estimate point-quantile in single-key situation. 
    \item We extend the algorithm to per-key situation and propose a novel data structure, namely \ourname{}, to estimate per-key point-quantile in data streams. 
    \item We show the error bound and time complexity of \ourname{} by strict derivation. Mathematical analysis shows the superiority of \ourname{} in both single-key and per-key situations. 
    \item We conduct extensive experiments on different datasets. Experimental results show that \ourname{} outperforms other algorithms in terms of error and still maintains high throughput in both single-key and per-key situations. 
\end{itemize}

\section{Problem Statement And Related Work}
\label{sec:prorel}

\subsection{Problem Statement}

The symbols frequently used in this paper are shown in Table \ref{table:symbol}.

\begin{definition}
\textbf{Data Stream.} A data stream $S$ is a series of items $\{e_1, e_2, \cdots, e_n, \cdots\}$ appearing in sequence. In this paper, every item $e$ is a key-value pair $(k, v)$. Items with different keys are called distinct items. 
\end{definition}

\begin{definition}
\textbf{Quantile.} Given a multiset of numbers $\mathcal{S} = \{a_1, a_2, \cdots, a_n\}$ and a percentage $w$, where $a_1\leq a_2\leq \cdots \leq a_n$ and $0\leq w\leq 1$, the $w$-quantile of multiset $\mathcal{S}$ is defined as $a_{\lfloor w(n-1) \rfloor + 1}$. The quantile function at $a_k$ is defined as $\frac{k-1}{n-1}$. 
\end{definition}

\begin{table}[H]
    \begin{center}
    \renewcommand\arraystretch{0.95}
    \caption{Symbols frequently used in this paper.}
    \label{table:symbol}
    \resizebox{.99\columnwidth}{!}{
    \begin{tabular}{|m{0.135\columnwidth}<{\centering}|m{0.75\columnwidth}|}
      \hline
      \textbf{Notation}&\textbf{Meaning}\\
      \hline
      $e$ & A distinct item in data streams \\
      $k$ & Key of a certain item \\
      $v$ & Value of a certain item \\
      $w$ & An arbitrary quantile \\
      $r$ & Maximum value samples in the Candidate \\
      $s$ & Maximum value samples in the Representative \\
      $T$ & Threshold for \stageone{} \\
      $\lambda$ & Threshold for ratio in \stagetwo{} \\
      $B[i][j]$ & $j^{th}$ cell in $i^{th}$ bucket in \stagetwo{} \\
      $h(.)$ & The hash function for \stagetwo{} \\
      $d$ & Number of cells in each bucket \\
      $u$ & Number of buckets in \stagetwo{} \\
      $vote^+$ & The positive vote field \\
      $vote^-$ & The negative vote field \\
      \hline
    \end{tabular}}
    \end{center}
\end{table}

We now officially state the definition of point-quantile estimation and all-quantile estimation respectively.

\begin{definition}
    \textbf{Point-quantile Estimation.} Given a quantile $w$ in advance, for $n$ numbers $t_1, \cdots, t_n$ in data streams, construct a data structure which estimates the $w$-quantile of all these numbers. 
\end{definition}

\begin{definition}
    \textbf{All-quantile Estimation.} For $n$ numbers $t_1, \cdots, t_n$ in data streams, construct a data structure which, for an arbitrary quantile $w$, estimates the $w$-quantile of all these numbers. 
\end{definition}

The difference between point-quantile and all-quantile estimation is that in point-quantile situation, we know the estimated quantile $w$ in advance, so we can specially design an algorithm which solves $w$-quantile estimation. Finally, we state the design goal of this paper. 

\begin{definition}
    \textbf{Per-key Point-quantile Estimation.} Given a fixed quantile $w$, for an arbitrary key $k$, the design goal of \ourname{} is to estimate the $w$-quantile of value for all items with key $k$.
\end{definition}

\subsection{Related Work}

\subsubsection{Quantile Estimation}~

Given a multiset $S$ with $N$ items, a quantile $0<w<1$ and an error parameter $0<\varepsilon<\frac{1}{2}$, the design goal of (single-key) quantile estimation is to return an item $x\in S$, \st{} the quantile function at $x$ (suppose it is $\hat{w}$) satisfies $|\hat{w}-w|<\varepsilon$. The most famous algorithm for this problem is GK algorithm \cite{greenwald2001space}. It solves this problem within $O(\frac{1}{\varepsilon}\log(\varepsilon N))$ space complexity. 
Other algorithms for quantile estimation uses randomization to guarantee $|\hat{w}-w|< \varepsilon$ with probability at least $1-\delta$. KLL algorithm \cite{karnin2016optimal} applies this technique and achieves $O(\frac{1}{\varepsilon}\log\log \frac{1}{\delta})$ space complexity for the first time. ReqSketch \cite{cormode2021relative, cormode2021theory} later guarantees $|\hat{w}-w|<\varepsilon w$ using $O(\frac{1}{\varepsilon}\log^{1.5}(\varepsilon N)\sqrt{\log\frac{1}{\delta}})$ space. 

Other sketching algorithms are also designed for single-key quantile estimation problem. 
DDSketch \cite{ddsketch} uses logarithm to divide all positive numbers into several discrete intervals, and records the frequency of each interval in the buckets. Moreover, DDSketch will merge adjacent buckets to save memory if it has too many buckets. 
Some algorithms utilize the fact that the data stream is often random-ordered (\ie{} i.i.d.) to save memory. 
Guha and Mcgregor \cite{guha2009stream} first provide a method that achieves $O\left(\log \frac{1}{\varepsilon} + \log\log \frac{1}{\delta}\right)$ space complexity for random-ordered streams. 
$t$-digest \cite{tdigest} is another work for random-ordered data streams. It divides the data into small clusters and aggregates and compresses the data within each cluster. The size of each cluster is adaptively determined based on the density of the data distribution. Moreover, $t$-digest employs a technique called ``centroids'' to group similar data points together and achieves accurate quantile estimation for $w$ close to 0 and 1.

Some algorithms are specifically designed for per-key quantile estimation. 
SQUAD \cite{shahout2022squad} combines both sketching and sampling methods to estimate per-key quantile for heavy-hitters. It uses Reservoir Sampling to sample several items from the data stream, and employs Space Saving to construct GK arrays for these items, which can be further used for quantile estimation. 
SketchPolymer \cite{guo2023sketchpolymer} is designed for estimating per-key tail quantile in data streams. It filters infrequent items in advance and applies Value Splitting and Sharing to count the frequency of every key in each interval for tail quantile estimation. 
M4 \cite{dongm4} is another framework in this field. For every single-key quantile estimation algorithm META, M4 constructs several layers of buckets with a META in each layer. Every distinct item will be mapped into several buckets by hash functions, and M4 aggregates information within different buckets to achieve per-key quantile estimation. 

However, none of these algorithms can be directly used for per-key point-quantile estimation in our problem setting. As to single-key algorithms, they achieve low space complexity in theory, but they have to make several copies of their data structures to support per-key point-quantile estimation, which is unacceptable in terms of memory. In addition, even if per-key algorithms (SQUAD, SketchPolymer, M4) can be directly utilized for our problem, SQUAD does not filter infrequent items in advance, and these infrequent items have a detrimental influence on overall accuracy; SketchPolymer is designed for estimating tail quantile, so it cannot be used for arbitrary point-quantile estimation. Finally, our experiments show that M4 cannot run within tight memory constraint, so it still has room for improvement in per-key quantile estimation.

\subsubsection{Frequency Estimation}~

Sketches are compact data structures which support approximate query with limited error, and they are widely used for frequency estimation in data streams. CM Sketch \cite{cmsketch} is the simplest sketch for frequency estimation. It is composed of $d$ counter arrays and each array is associated with a hash function. For every incoming item, CM Sketch uses $d$ hash functions to map it into $d$ counters and increments all these counters by 1. To report the frequency of an item, CM Sketch returns the minimum value of these $d$ counters. 

Tower Sketch \cite{towersketch} is a novel counting sketch originating from CM Sketch but uses different-sized counters for different arrays. Since Tower Sketch still allocates the same memory for different arrays, it has more small counters for infrequent items, which can keep the frequency of these items more accurately. Also, Tower Sketch has fewer counters for frequent items, and their frequencies are more likely to overflow in small counters. Therefore, frequent items suffer from overestimation error in practice, and Tower Sketch can separate frequent items from infrequent items more effectively in data streams compared to other algorithms \cite{fan2023finding}. 

Elastic Sketch \cite{yang2018elastic} consists of a hash table with several buckets. Each bucket records the key, positive votes $vote^+$ (equal to its frequency) and negative votes $vote^-$ (equal to the number of items colliding with it). To insert an item $e$, Elastic Sketch uses a hash function to map it to a bucket. If the item $e$ matches the key in the bucket, the positive votes $vote^+$ will be incremented by 1; otherwise, the negative votes $vote^-$ will be incremented by 1. Moreover, Elastic Sketch applies a technique called Ostracism to evict infrequent items: If $\frac{vote^-}{vote^+}$ exceeds a predefined threshold, the item will be viewed as an infrequent item and will be evicted from the bucket. 

\section{\ourname{} Algorithm}

\label{sec:alg}

In this section, we propose the data structure of \ourname{}. 
We start from the simplest situation when all items in the data stream share the same key and $w=0.5$. Then we introduce the idea of Distribution Calibration and support an arbitrary $w$ in single-key situation. We adapt the \ourname{} to per-key situation and introduce the \stagetwo{}. Finally we present an optimized version of \ourname{} to filter infrequent items in advance for higher accuracy. Due to space constraint, we put relevant pseudo-code of \ourname{} in Appendix \ref{appendix:code}.

\subsection{Single-key Situation when $w=0.5$: Value Focus}
\label{sec:single}

\noindent
\textbf{Value Focus Rationale}: The main idea of \textbf{Value Focus} is to maintain both the Candidate and the Representative in our data structure. Items in the data stream will be inserted into the Candidate first. When the Candidate is full, we select two value samples from the Candidate and insert them into the Representative. The Representative further evicts value samples far from the 0.5-quantile when it is full.

\noindent
\textbf{The Candidate Operation}: The design goal of the Candidate is to keep several value samples from the data stream, and only values which have the potential to be the 0.5-quantile can be selected into the Representative. Specifically, the Candidate can keep $r$ value samples in the data structure (we assume that $r$ is an even number). For every item $e=(k, v)$ in the data stream, we first insert its value into the Candidate. Then we check whether the Candidate is full after insertion. If the number of value samples in the Candidate is smaller than $r$, we just finish insertion procedure and return; Otherwise, we select two median value samples in the Candidate and send them to the Representative. Finally, we clear all value samples in the Candidate to make room for further items in the data stream. 

\begin{figure*}[htbp]
\centering
\includegraphics[width=0.9\linewidth]{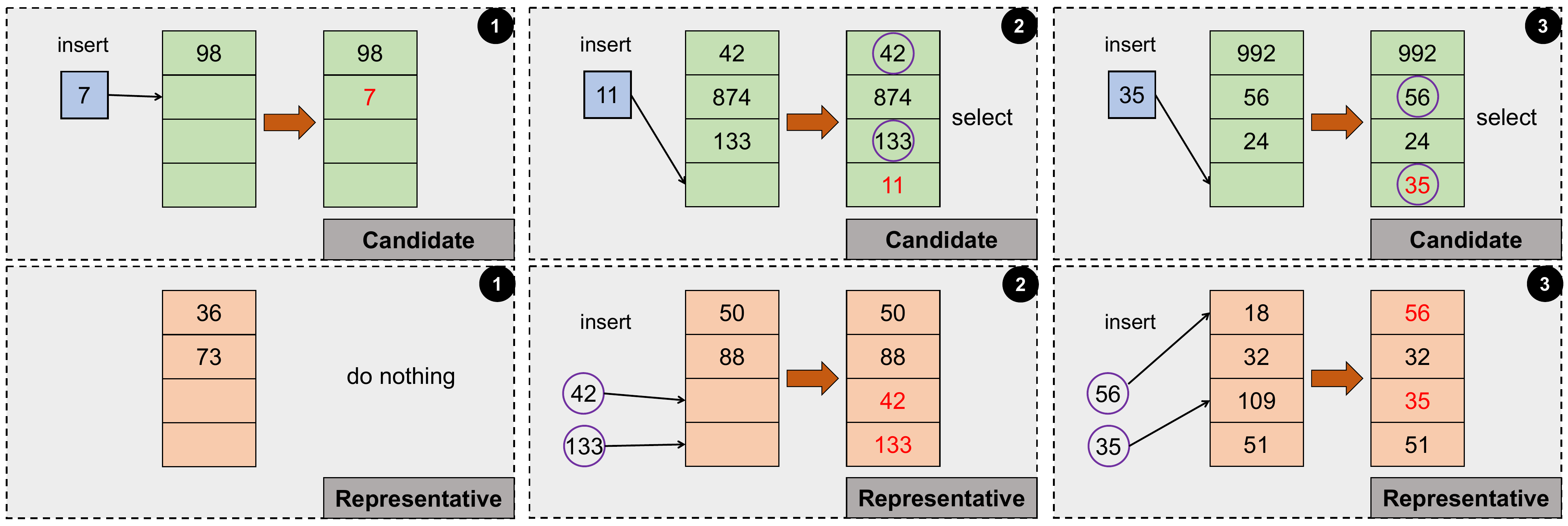}
\vspace{-0.05in}
\caption{Examples in Single-key Situation}
\label{fig::single-exp}
\vspace{-0.2in}
\end{figure*}

\noindent
\textbf{The Representative Operation}: The design goal of the Representative is to record value samples close to the target quantile. The Representative can keep $s$ value samples (we assume that $s$ is also an even number). To insert two value samples $v'$ and $v''$ (which are selected from the Candidate) into the Representative, we first check whether the Representative is full. If the Representative still has space for value samples, we just insert $v'$ and $v''$ into the Representative; Otherwise, we traverse the Representative and find the smallest and largest value samples from $s$ value samples in the Representative and the value samples to be inserted (suppose they are $v_m$ and $v_M$ respectively). The two smallest and largest value samples $v_M$ and $v_m$ will be evicted from the bucket and $v'$ and $v''$ will be inserted into the Representative. Finally, to query for the 0.5-quantile, we regard value samples in the Representative and return the median in the Representative as the 0.5-quantile of all items.

\noindent
\textbf{Running Examples:} For simplicity, we choose $r=s=4$. Figure \ref{fig::single-exp} shows three running examples in single-key situation. In the first example, to insert the value $v=7$, we insert it into the Candidate. Since the Candidate is not full, no value will be inserted into the Representative and we finish the insertion procedure. In the second example, to insert the value $v=11$, we insert 11 into the Candidate, and the Candidate will be full after insertion. Consequently, we select two median value samples from the Candidate (42 and 133 here). As the Representative still has space for these two value samples, 42 and 133 will be inserted into the Representative. In the last example, to insert the value $v=35$, we insert 35 into the Candidate. The Candidate is full after insertion, so we select two median value samples (56 and 35 here) and insert them into the Representative. Since the Representative is full, we select the maximum and the minimum values from these six numbers (18 and 109 here). 18 and 109 will be evicted from the Representative, and 56 and 35 will take their place in the Representative.

\subsection{Idea of Distribution Calibration}
\label{sec:cali}
Suppose the value of all items in the data stream follows an arbitrary distribution $F$. To support arbitrary quantile estimation in single-key situation, we apply \textbf{Distribution Calibration} technique. Its key idea is to construct a new distribution $F'$, \st{} the $w$-quantile of the original distribution $F$ is just the 0.5-quantile of the new distribution $F'$. Specifically, to ``calibrate'' the distribution, we apply probability method to generate several positive (negative) infinities and insert them into the data structure. When $w> 0.5$, we set a random variable $Z$ following geometric distribution with parameter $\frac{1}{2w}$. We insert $Z-1$ positive infinities and $v$ into \ourname{}. In this way, every value sample taken from the distribution $F'$ follows the distribution $F$ with probability $\frac{1}{2w}$, and the value sample will be positive infinity with probability $\frac{2w-1}{2w}$. The situation when $w< 0.5$ is similar, except that $Z$ follows geometric distribution with probability $\frac{1}{2-2w}$, and we insert $v$ together with $Z-1$ negative infinities instead of positive infinities into \ourname{}. We will prove that the expectation of quantile function at the query result is 0.5 \wrt{} $F'$ in Section \ref{sec:math}, hence \ourname{} gives an unbiased estimation of quantiles \wrt{} $F'$. 

\subsection{Per-key Situation: Value Sketch}
\label{sec:per}
\noindent
\textbf{\stagetwo{} Data Structure:} We find that the technique called Ostracism, which is first proposed by Elastic Sketch \cite{yang2018elastic} for frequency estimation, fits well for our problem setting in per-key quantile estimation. To cater for per-key situation in data streams, \ourname{} maintains a hash table with $u$ buckets in the \stagetwo{}. Each bucket has $d$ cells, and each cell records three fields: the key $k$, the positive votes $vote^+$, which is equal to its frequency numerically, and the value field, which includes the Representative and Candidate similar to single-key situation. Moreover, every bucket records the negative votes $vote^-$, which is equal to the number of items that collide in the bucket. 

\begin{figure*}[htbp]
\centering
\includegraphics[width=0.9\linewidth]{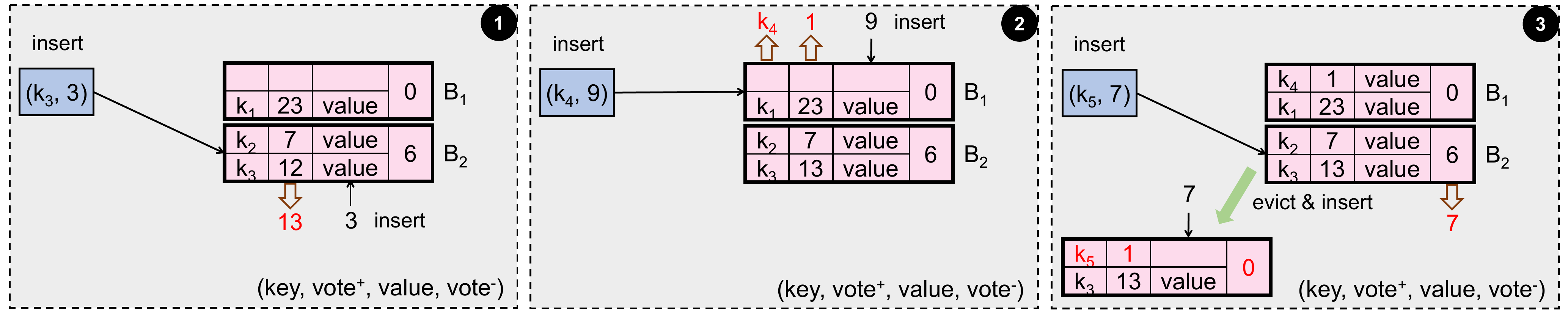}
\vspace{-0.05in}
\caption{Examples in Per-key Situation}
\label{fig::per-exp}
\vspace{-0.2in}
\end{figure*}

\noindent
\textbf{\stagetwo{} Insertion Operation:} To insert an item $e=(k, v)$, \stagetwo{} first uses the hash function $h(.)$ to map $e$ into the bucket $B[h(k)]$ and check the cells in the bucket. Specifically, there are three sub-cases:

\textit{Case 1:} The key $k$ of the item $e$ matches the key in some cell $B[h(k)][j]$. In this case, we just increment its positive vote $vote^+$ by 1 and insert the value $v$ of the item $e$ into the value field in $B[h(k)][j]$. Specifically, we generate positive (negative) infinities by probability method, and insert these infinities and the value into the Candidate. If the Candidate is full, we select two medians from the Candidate, insert them into the Representative and clear the Candidate. Moreover, if the Representative is full, we evict the largest and smallest value samples to make room for new value samples, which is just similar to single-key situation.

\textit{Case 2:} The key $k$ of the item $e$ does not match any key in bucket $B[h(k)]$, but the bucket still has empty cells. In this case, we insert the item $e=(k,v)$ into one empty cell (suppose it is $B[h(k)][j]$). The key in $B[h(k)][j]$ is set to $k$, and the positive vote $vote^+$ is set to 1. Finally, we insert the value $v$ of the item $e$ into the value field in $B[h(k)][j]$. 

\textit{Case 3:} The key $k$ of the item $e$ does not match any key in bucket $B[h(k)]$, and the bucket does not have empty cells. In this case, we increment the negative votes $vote^-$ of bucket $B[h(k)]$ by 1, which shows that hash collision happens. Then, we find the cell in bucket $B[h(k)]$ with minimum positive vote $vote^+$ (suppose it is $B[h(k)][j]$). We now check whether $\frac{B[h(k)].vote^-}{B[h(k)][j].vote^+}\geq \lambda$ holds: If the inequality holds after we increment $B[h(k)].vote^-$, we treat this item as an infrequent item: we evict the item from \stagetwo{} and insert $e=(k,v)$ into the cell by clearing $B[h(k)][j]$, initializing $B[h(k)][j].k$ to $k$, setting $B[h(k)][j].vote^+$ to 1, inserting the value $v$ into the value field $B[h(k)][j].value$ and resetting $B[h(k)].vote^-$ to 0; If the inequality does not hold, then nothing happens: no item will be evicted from \stagetwo{}, and $e$ will not be inserted into \stagetwo{} either.

\noindent
\textbf{\stagetwo{} Query Operation:} The query operation is simple. To query the $w$-quantile of value for items with key $k$, we again use the hash function $h(.)$ to locate bucket $B[h(k)]$. \stagetwo{} then traverses the bucket to find the cell which records $k$. \stagetwo{} will return the query result of the value field as the $w$-quantile of value for items with key $k$.

\noindent
\textbf{Running Examples:} For simplicity, we choose $d=u=2$ and $\lambda=1$. Figure \ref{fig::per-exp} shows three running examples of \stagetwo{} in per-key situation. In the first example, to insert item $e=(k_3, 3)$, we use a hash function $h(.)$ to map it into the bucket $B[2]$. $k_3$ matches the key stored in the second cell, so we increment the $vote^+$ by 1 and insert the value 3 into the value field. In the second example, to insert item $f=(k_4, 9)$, we map it into the bucket $B[1]$. The key $k_4$ does not match any key in the bucket, but the first cell $B[1][1]$ is empty. Hence, we insert $(k_4, 9)$ into $B[1][1]$: the key field is set to $k_4$, the $vote^+$ is set to 1 and the value 9 is inserted into the value field. In the last example, to insert item $g=(k_5, 7)$, we map it into the bucket $B[2]$. The key $k_5$ does not match any key in the bucket, and the bucket does not have empty cell either. So we increment $vote^-$ by 1. Next we find the minimum $vote^+$ in the bucket $B[2]$, and check whether the inequality $\frac{vote^-}{vote^+}\geq \lambda$ holds. The cell $B[2][1]$ has the minimum $vote^+$, and $\frac{vote^-}{vote^+}=1$. Consequently, we clear the cell $B[2][1]$ and insert $g=(k_5, 7)$ into $B[2][1]$: the key field is initialized as $k_5$, $vote^+$ is set to 1, the value 7 is inserted into the value field and $vote^-$ is reset to 0.

\subsection{Optimization: Early Screening}
\label{sec:tower}
\noindent
\textbf{Idea of Early Screening:} To support per-key point-quantile estimation, \ourname{} records the value of both frequent and infrequent items. However, the quantile of infrequent items cannot be accurately estimated due to their low frequency. In addition, it is widely known that the distribution of real datasets is usually skewed, and items with low frequency make up the majority of the data stream. Inspired by the Cold Filter \cite{coldfilter}, \ourname{} proposes \textbf{Early Screening} to effectively filter infrequent items in advance: \ourname{} uses a \stageone{} which records the frequency of every item, and only items with frequency exceeding the predefined threshold are allowed to enter \stagetwo{}. In this way, \stageone{} automatically separates frequent items from infrequent items, and allows \ourname{} to keep the information of frequent items in \stagetwo{} as accurately as possible.

\noindent
\textbf{\ourname{} Insertion Operation:} To insert an item $e=(k,v)$, \ourname{} first uses its key $k$ to query \stageone{} to get its frequency. If its frequency exceeds the threshold $T$, we insert $e=(k,v)$ into \stagetwo{}; Otherwise, we simply insert $e$ into \stageone{} and return.

\noindent
\textbf{\ourname{} Query Operation:} The query operation is simple: To query the $w$-quantile of value for items with key $k$, \ourname{} just query \stagetwo{} to get the estimation: namely, use the hash function $h(.)$ to locate the bucket $B[h(k)]$ in \stagetwo{}, get the cell which records the information of key $k$, and query the value field to give an estimation of the quantile.

\section{Mathematical Analysis}
\label{sec:math}

In this section, we conduct mathematical analysis for \ourname{}. We first prove that the query result of \ourname{} is unbiased on the new distribution in Section \ref{math::unbias}. Then we derive the error bound for both the Candidate and the Representative, and give the space complexity in single-key situation in Section \ref{math::single}. We give an error bound for \stagetwo{} in per-key situation on Section \ref{math::per}. Finally we analyze the overall time complexity of \ourname{} in Section \ref{math::time}. 

\subsection{Proof of Unbiasedness}
In this part, we assume that the value of all items in the data stream follows an arbitrary continuous distribution $F$, and the calibrated distribution is $F'$. All quantiles are \wrt{} the calibrated distribution $F'$. We prove that the expectation of quantile function at the query result \wrt{} $F'$ is just 0.5 in single-key situation. 
\label{math::unbias}

\begin{theorem}
\label{thm:uniform}
    For every value or infinity $v$ inserted into \stagetwo{}, the distribution of the quantile function at $v$ (\wrt{} $F'$) is uniform distribution on the interval $[0,1]$. 
\end{theorem}

\begin{proof}
    Without loss of generality, we assume $w> 0.5$. Suppose the quantile functions at $v$ \wrt{} $F$ and $F'$ are $\hat{w}$ and $\tilde{w}$ respectively, then for $0\leq x\leq \frac{1}{2w}$, 
    \[
    \mathbb{P}(\tilde{w}\leq x) = \mathbb{P}(\text{$v$ is not infinity}) \cdot \mathbb{P}(\hat{w}\leq 2wx) = \frac{1}{2w} \cdot 2wx = x. 
    \]
    Similarly, we can prove that $\mathbb{P}(\tilde{w}\leq x)=x$ for $\frac{1}{2w}\leq x\leq 1$, hence $\tilde{w}$ obeys a uniform distribution on $[0,1]$.  
\end{proof}

\begin{theorem}
\label{thm:unbiased}
    Suppose \ourname{} returns $v$ as the query result, and the quantile function at $v$ is $\tilde{w}$, then the expectation of $\tilde{w}$ is 0.5. 
\end{theorem}

\begin{proof}
    The proof is based on symmetry: for every two value samples $v_1<v_2$ sent to the Representative, suppose the quantile function at $v_1$ and $v_2$ are $\tilde{w}_1$ and $\tilde{w}_2$ respectively, then $\tilde{w}_1$ and $1-\tilde{w}_2$ have identical distribution. As a result, $\tilde{w}$ and $1-\tilde{w}$ have identical distribution, and $\mathbb{E}\tilde{w} = \mathbb{E}(1-\tilde{w})$. Hence $\mathbb{E} \tilde{w} = 0.5$.
\end{proof}

\subsection{Error Bound in Single-key Situation}
\label{math::single}
In this part, we analyze the property of the Candidate and the Representative in single-key situation respectively. We first derive the error bound of the Candidate, then we use the error bound of Candidate to show the error bound of the Representative. Finally we prove that the space complexity of the Candidate and the Representative in single-key situation is $O\left(\frac{1}{\varepsilon}\sqrt{\log \frac{1}{\delta}}\right)$. We assume that the value of all items in the data stream follows an arbitrary distribution $F$, and the calibrated distribution is $F'$. All quantiles are \wrt{} the calibrated distribution $F'$ in this part.

\begin{theorem}
\label{thm:candidate}
For every value $v$ which is sent from the Candidate to the Representative, suppose the real quantile function at $v$ is $\hat{w}$. Then for any small positive number $\varepsilon$, 
\begin{equation}
    \mathbb{P}(\hat{w}-0.5<-\varepsilon) = \mathbb{P}(\hat{w}-0.5>\varepsilon) \leq e^{-\varepsilon ^2r}.
\end{equation}
\end{theorem}

To prove Theorem \ref{thm:candidate}, we first cite the famous Chernoff Bound. 

\begin{lemma}[Chernoff Bound]
\label{thm:chernoff}
Suppose $X_1, \cdots, X_n$ are i.i.d. random variables taking values in $\{0, 1\}$. Let $X=X_1+\cdots +X_n$ denote their sum and $\mu = \mathbb{E}X$ denote its expectation. For any positive number $\eta >0$, we have 
\begin{equation}
 \mathbb{P}(X\geq (1+\eta)\mu) \leq e^{-\eta^2\mu / 2}.    
\end{equation}
\end{lemma}

Then we apply this bound to prove Theorem \ref{thm:candidate}. 

\begin{proof}
Assume $r$ value samples in the Candidate are $v_1, \cdots, v_r$, and $w_1, \cdots, w_r$ are the quantile function at $v_1, \cdots, v_r$ respectively. We define indicative variables 
\[
X_j = 1_{\{w_j < 0.5-\varepsilon\}}, Y_j = 1_{\{w_j > 0.5+\varepsilon\}}, j=1, \cdots, r. 
\]
Let $X=X_1 + \cdots +X_r$ and $Y=Y_1+\cdots +Y_r$, then $\mathbb{E}X=\mathbb{E}Y=(0.5-\varepsilon)r$. Let $\eta = \frac{\varepsilon}{0.5-\varepsilon}$. Applying Lemma \ref{thm:chernoff}, we get 
\[
\mathbb{P}(X \geq (1+\eta) \mathbb{E}X) = \mathbb{P}(Y \geq (1+\eta) \mathbb{E}Y) \leq e^{-\frac{\eta^2 \mathbb{E}X}{2}} \approx e^{-\varepsilon^2 r}.
\]
Note that $\hat{w}-0.5<-\varepsilon$ if and only if $X \geq 0.5r$, and $\hat{w}-0.5>\varepsilon$ if and only if $Y\geq 0.5r$, so we get the result. 
\end{proof}

\begin{theorem}
\label{thm::representative}
    Assume that the number of items in the data stream is sufficiently large. Suppose the Representative reports $v$ as the query result, and the real quantile function at $v$ is $\tilde{w}$. Then for any small positive number $\varepsilon$, 
    \begin{equation}
        \mathbb{P}(|\tilde{w}-0.5|>\varepsilon) \leq 2e^{-\varepsilon^2rs}. 
    \end{equation}
\end{theorem}

To prove Theorem \ref{thm::representative}, we first show a lemma in stochastic process and prove it. 

\begin{lemma}[Random Walk with Reflecting Barrier]
\label{thm:markov}
    Consider a random walk in $\{0, 1, \cdots, s\}$, with transition probability 
    \begin{align*}
        &\mathbb{P}(X_{l+1}=i+1|X_l=i) = x, & 0\leq i\leq s-1, \\
        &\mathbb{P}(X_{l+1}=i-1|X_l=i) = y, & 1\leq i\leq s,  \\
        &\mathbb{P}(X_{l+1}=i|X_l=i) = 1-x-y, &1\leq i\leq s-1, 
    \end{align*}
    and 
    \[
    \mathbb{P}(X_{l+1}=0|X_l=0) = 1-x, \mathbb{P}(X_{l+1}=s|X_l=s) = 1-y, 
    \]
    where $X_l$ denotes the position at time $l$. Its stationary distribution is 
    \begin{equation}
    \pi_i = \frac{\left(\frac{x}{y}\right)^i}{\sum _{j=0}^s \left(\frac{x}{y}\right)^j}.    
    \end{equation}
\end{lemma}

Then we use the lemma to prove Theorem \ref{thm::representative}. 

\begin{proof}
    Define $X_k$ as the number of value samples in the Representative which are smaller than the $(0.5-\varepsilon)$-quantile after $k$ value samples are sent to the Representative, and $x=p^2, y=(1-p)^2$, where $p$ is equal to $\mathbb{P}(\hat{w}-0.5<-\varepsilon)$ and $\mathbb{P}(\hat{w}-0.5>\varepsilon)$ in Theorem \ref{thm:candidate}. It is easy to verify that $X_k$ obeys the random walk with reflecting barrier in Lemma \ref{thm:markov}. Since the number of items in the data stream is sufficiently large, we assume that the probability distribution of $X_k$ converges to the stationary distribution. Notice that $\tilde{w}-0.5< \varepsilon$ if and only if $X_k \geq 0.5s$. Hence 
    \[
    \begin{aligned}
    & \mathbb{P}(\tilde{w}-0.5< -\varepsilon) = \mathbb{P}(X_k \geq 0.5s) \approx \sum _{i=0.5s}^s \pi_i \\ 
    & = \frac{u^{0.5s}(1-u^{s-0.5s})}{1-u^s} \approx u^{0.5s} = \left(\frac{p}{1-p}\right)^{s} \approx p^{s}, 
    \end{aligned}
    \]
    where $u=\frac{x}{y}$. Similarly, we can prove 
    \[
    \mathbb{P}(\tilde{w}-0.5>\varepsilon) \leq p^{s}.
    \]
    Applying union bound, we get 
    \[
    \mathbb{P}(|\tilde{w}-0.5|>\varepsilon) \leq 2p^{s} \leq 2e^{-\varepsilon^2rs}. \qedhere
    \]
\end{proof}

\begin{remark}
According to the Dobrushin Theorem \cite{dobrushin1956central}, the convergence rate of an irreducible aperiodic Markov chain with finite states is exponential, so the distribution of the Representative converges quickly. 
\end{remark}

\begin{corollary}
\label{thm:space}
    Given two small positive number $\varepsilon$ and $\delta$, to guarantee that $|\tilde{w}-0.5|\leq \varepsilon$ with probability at least $1-\delta$, the Candidate and the Representative needs $O\left(\frac{1}{\varepsilon}\sqrt{\log\frac{1}{\delta}}\right)$ memory. 
\end{corollary}

\begin{proof}
    Let $r=s$ and $2e^{-\varepsilon^2rs} \leq \delta$, we get 
    \[
    s \geq \frac{1}{\varepsilon}\sqrt{\log\frac{2}{\delta}}, 
    \]
    which shows that $r+s = 2s = O\left(\frac{1}{\varepsilon}\sqrt{\log\frac{1}{\delta}}\right)$. 
\end{proof}

\subsection{Error Bound in Per-key Situation}
\label{math::per}

In this part, we first derive the probability of hash collision in \stagetwo{}, then we use this conclusion to give an error bound for \stagetwo{} in per-key situation.

\begin{theorem}
\label{thm:per}
    For any bucket in \stagetwo{}, the probability of hash collision is 
    \begin{equation}
        P_{hc} = 1-e^{-\frac{H}{u}} \sum _{i=0} ^d \frac{1}{i!}\left(\frac{H}{u}\right)^i,
    \end{equation}
    where $H$ is the number of distinct items entering \stagetwo{}. 
\end{theorem}

\begin{proof}
    The proof is first given in \cite{yang2018elastic} when $d=1$. We now extend the conclusion to $d\geq 2$. There are totally $H$ distinct items which enter \stagetwo{}, and every distinct item is randomly mapped into a bucket by the hash function $h(.)$. Given an arbitrary bucket and an item $e=(k, v)$, the probability that $e$ is mapped to the bucket is $\frac{1}{u}$. Therefore, for any bucket, the number of distinct items mapped to the bucket $W$ follows a Binomial distribution $B\left(H, \frac{1}{u}\right)$. When both $H$ and $u$ are large, the distribution of $W$ can be approximated by Poisson distribution $\pi\left(\frac{H}{u}\right)$, hence
    \[
    P(W=i) = \frac{1}{i!}\left(\frac{H}{u}\right)^i e^{-\frac{H}{u}}.
    \]
    Notice that hash collision happens if and only if $W\geq d+1$, hence
    \begin{align*}
    P_{hc} & = P(W\geq d+1) = 1-P(W\leq d) \\
    & = 1-e^{-\frac{H}{u}} \sum _{i=0} ^d \frac{1}{i!}\left(\frac{H}{u}\right)^i.       \qedhere
    \end{align*} 
\end{proof}

\begin{corollary}
    Let $r=s=\frac{1}{\varepsilon}\sqrt{\log\frac{2}{\delta}}$ in Corollary \ref{thm:space}. For an arbitrary key $k$, suppose \stagetwo{} reports $v$ as the query result, and the real quantile function at $v$ is $\tilde{w}$. Then \stagetwo{} guarantees $|\tilde{w}-0.5|\leq \varepsilon$ with probability at least $(1-\delta)(1-P_{hc})$.
\end{corollary}

\vspace{-0.1in}
\subsection{Time Complexity}
\label{math::time}
In this part, we analyze the amortized time complexity of \ourname{} in both single-key and per-key situations. Since \ourname{} has a small probability to insert a large number of infinities, we will only show the \textit{expectation} of amortized time complexity of \ourname{}. 
\begin{theorem}
\label{thm:time}
The expectation of amortized time complexity to insert an arbitrary item $e=(k, v)$ is at most $O\left(\frac{s}{r}\right)$ in single-key situation. 
\end{theorem}

\begin{proof}
    To insert item $e$, \ourname{} sets $Z$ to generate positive (negative) infinities. 
    The expected number of positive (negative) infinities inserted into the Representative before $e$ is inserted is $|2w-1|$, which is a constant. 
    Every value or infinity will first be inserted into the Candidate, and if the Candidate is full, then two value samples in the Candidate will be inserted into the Representative. 
    The insertion into the Candidate costs $O(1)$ time, and finding the median among $r$ value samples costs $O(r)$ time. 
    If the Representative is full, then \ourname{} will find the maximum and minimum value in the Representative and evict them, which costs $O(s)$ time. 
    In conclusion, if we insert $r$ values or infinities into the Candidate, then the time complexity is at most $O(r+s)$. 
    Hence the expectation of amortized time complexity to insert $e$ is at most $O\left(\frac{r+s}{r}\right) = O\left(\frac{s}{r}\right)$ in single-key situation. 
\end{proof}

\begin{remark}
Although \ourname{} occasionally inserts a large number of infinities, the expectation number of generated infinities to insert a value is $|2w-1|$, and it is smaller than 1, so the insertion throughput of \ourname{} will only be halved due to generating infinities in the worst case (when $w$ is close to 0 or 1). In addition, by Kolmogorov's Strong Law of Large Numbers (SLLN), the median in the Representative will rarely be infinity, so \ourname{} gives a valid estimation in most cases. 
\end{remark}

\begin{theorem}
\label{thm:per-time}
The expectation of amortized time complexity to insert an arbitrary item $e=(k,v)$ is at most $O\left(\frac{s}{r}+d\right)$ in per-key situation. 
\end{theorem}

\begin{proof}
    To insert item $e$, \ourname{} first query \stageone{} to get its frequency. 
    If its frequency does not exceed the threshold $T$, $e$ will only be inserted into \stageone{} and return, which can be done in $O(1)$ time. 
    Otherwise, \ourname{} uses a hash function $h(.)$ to locate a bucket and traverse the bucket in \stagetwo{}, which costs $O(d)$ time. 
    If $e$ is in the bucket or $e$ is not in the bucket but the bucket still has empty cells, then $e$ will be inserted into the cell, the expectation of amortized time complexity is $O\left(\frac{s}{r}\right)$. 
    If $e$ is not in the bucket and the bucket does not have empty cells, \ourname{} will just increment $vote^-$ and return. 
    In conclusion, the expectation of amortized time complexity to insert $e$ is at most $O\left(\frac{s}{r}+d\right)$ in per-key situation. 
\end{proof}

As a result, if $r$ and $s$ are close numbers, and $d$ is not very large, then the overall time complexity to insert an item $e$ can be viewed as $O(1)$, which shows that \ourname{} can catch up with the high-speed items in the data stream and achieve high throughput. 
\section{Experimental Results}
\label{sec:experiments}

In this section, we provide experimental results with \ourname{}. First, we describe the experimental setup in Section \ref{exp::setup}. Then, we show how parameter settings affect \ourname{} performance in Section \ref{exp::param}. We compare the performance of \ourname{} with state-of-the arts in both single-key and per-key situations in Section \ref{exp::single} and \ref{exp::per} respectively. Finally, we implement \ourname{} on RocksDB database to reduce quantile query latency in Section \ref{exp::rocksdb}. All related codes are released on GitHub \cite{source}. 

\subsection{Experimental Setup}
\label{exp::setup}

\begin{figure*}[!ht]
	\centering
	\begin{minipage}{0.216\textwidth}{
			\includegraphics[width=1\textwidth]{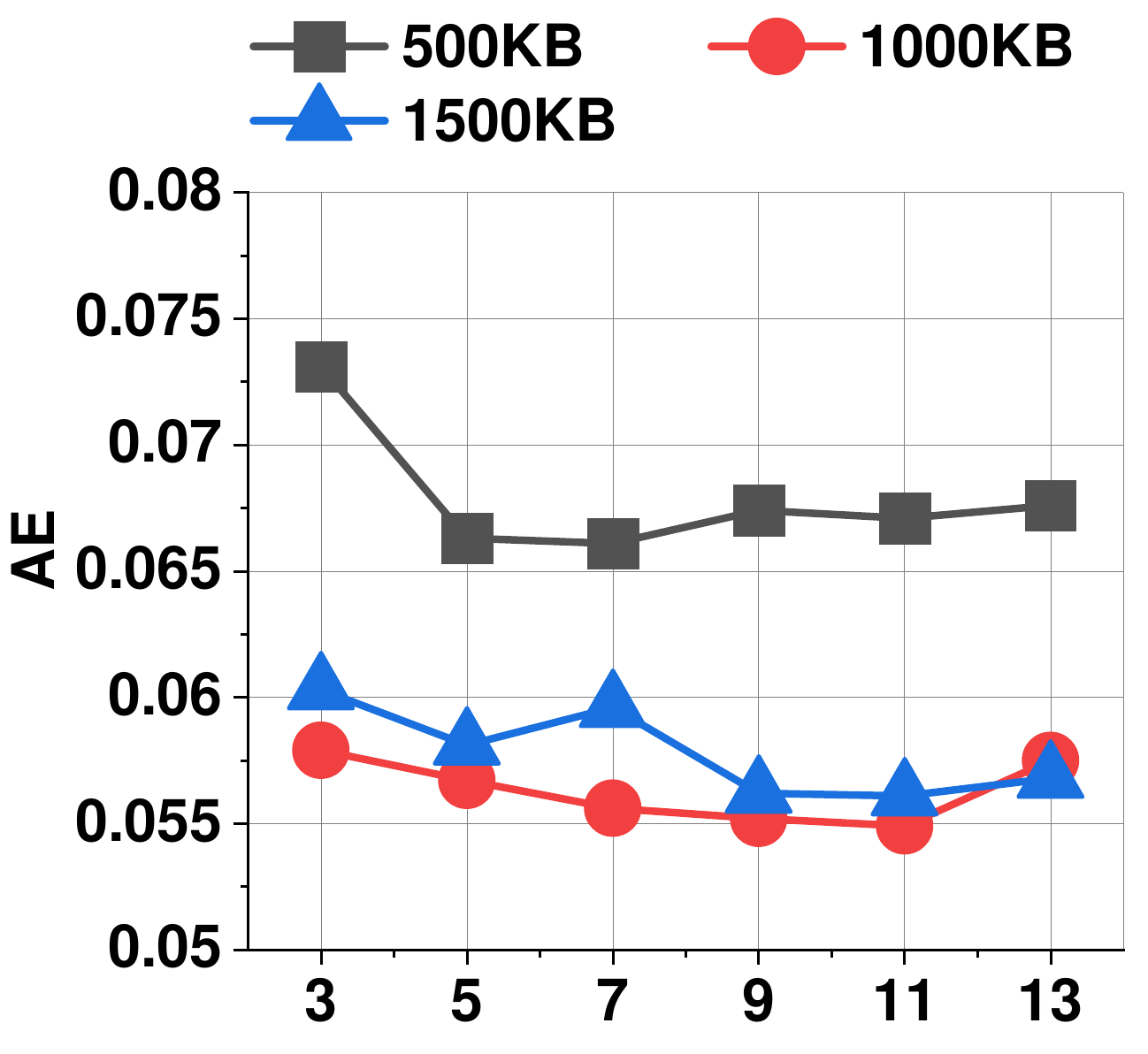}
            \vspace{-0.25in}
            \caption{Effects of $d$}
            \label{fig::d}
		}
	\end{minipage}
	\begin{minipage}{0.216\textwidth}{
			\includegraphics[width=1\textwidth]{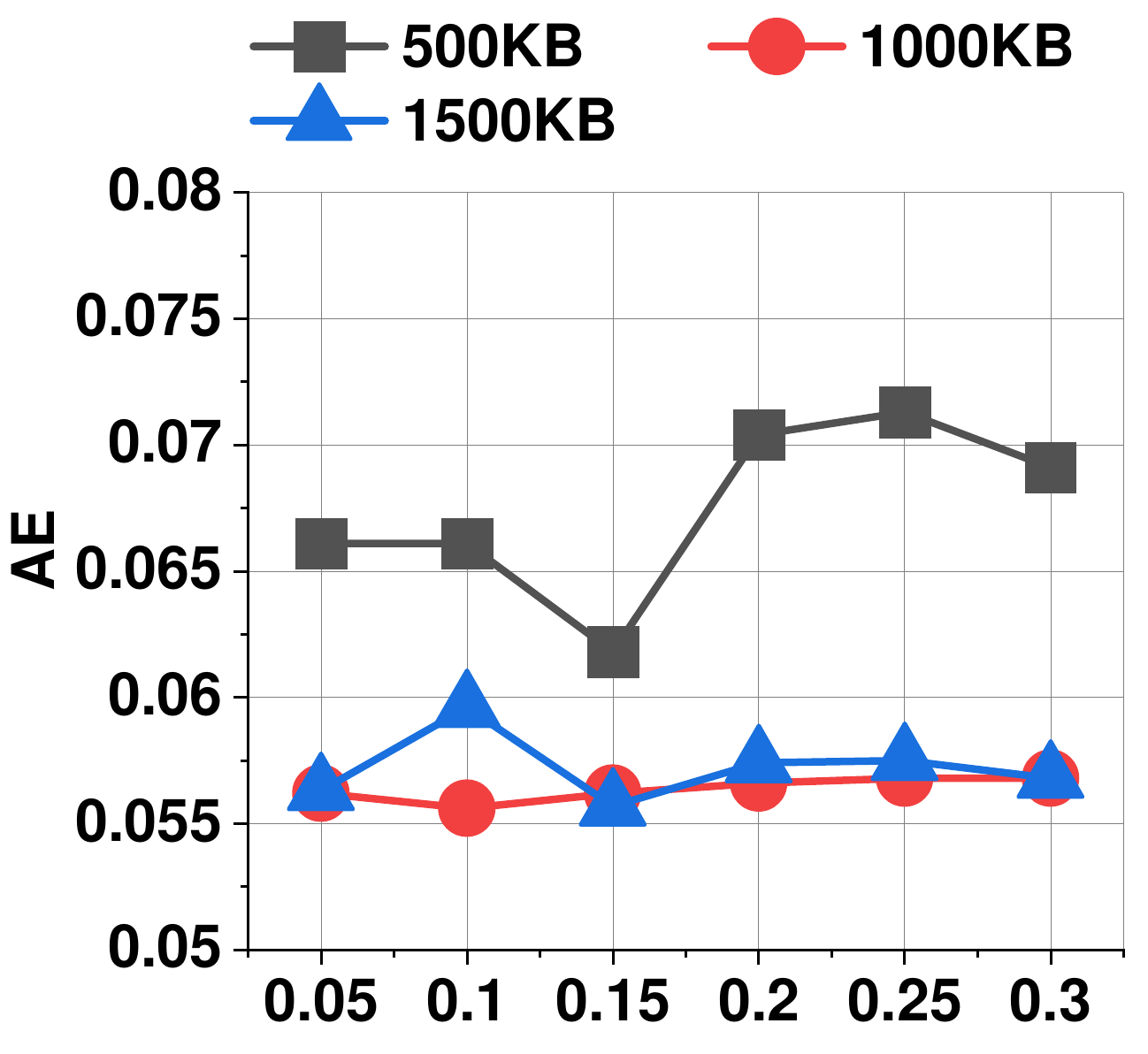}
            \vspace{-0.25in}
            \caption{Effects of $q$}
            \label{fig::q}
		}
	\end{minipage}
	\begin{minipage}{0.216\textwidth}{
			\includegraphics[width=1\textwidth]{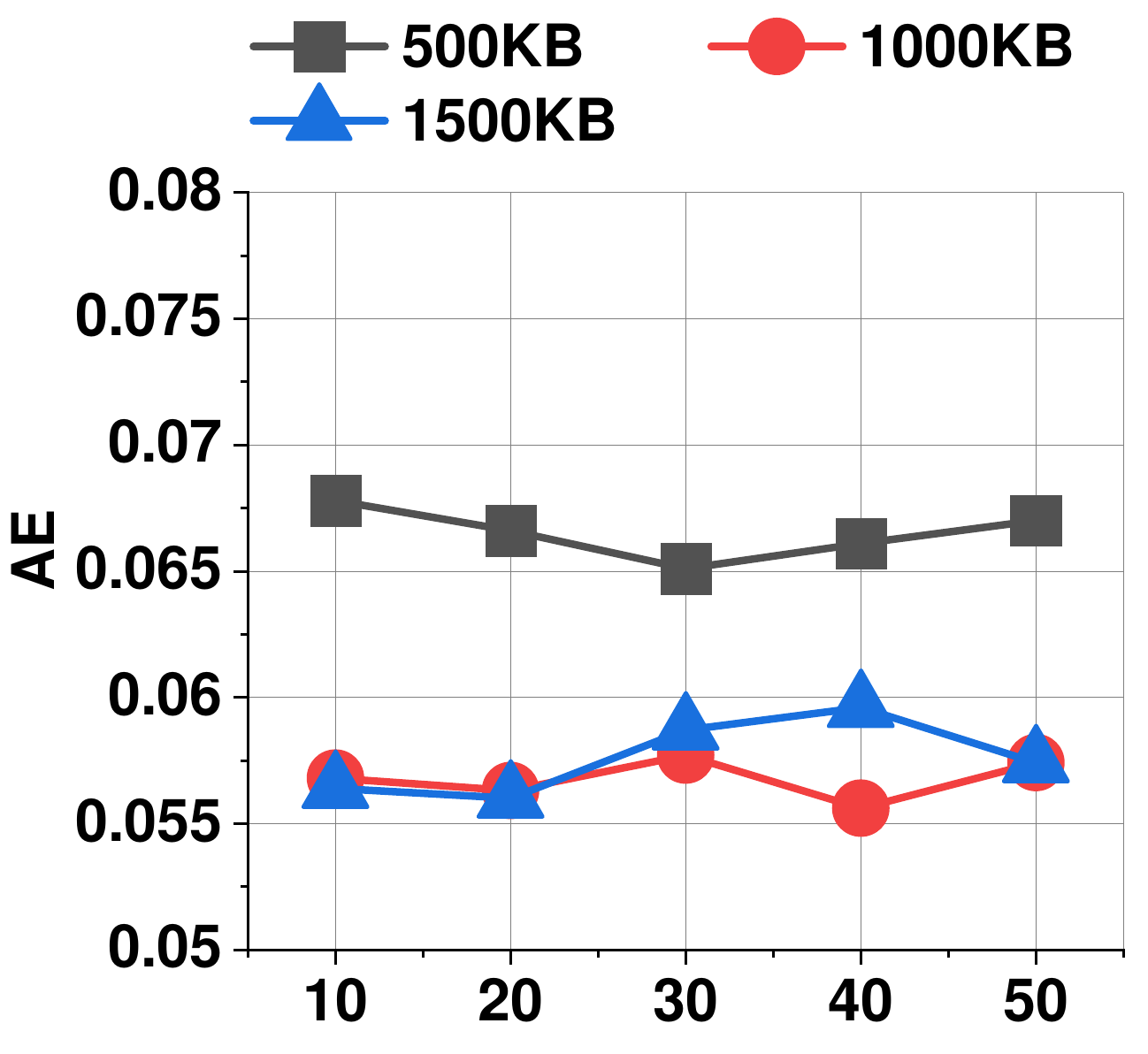}
            \vspace{-0.25in}
            \caption{Effects of $T$}
            \label{fig::t}
		}
	\end{minipage} \\
 \vspace{0.05in}
	\begin{minipage}{0.216\textwidth}{
			\includegraphics[width=1\textwidth]{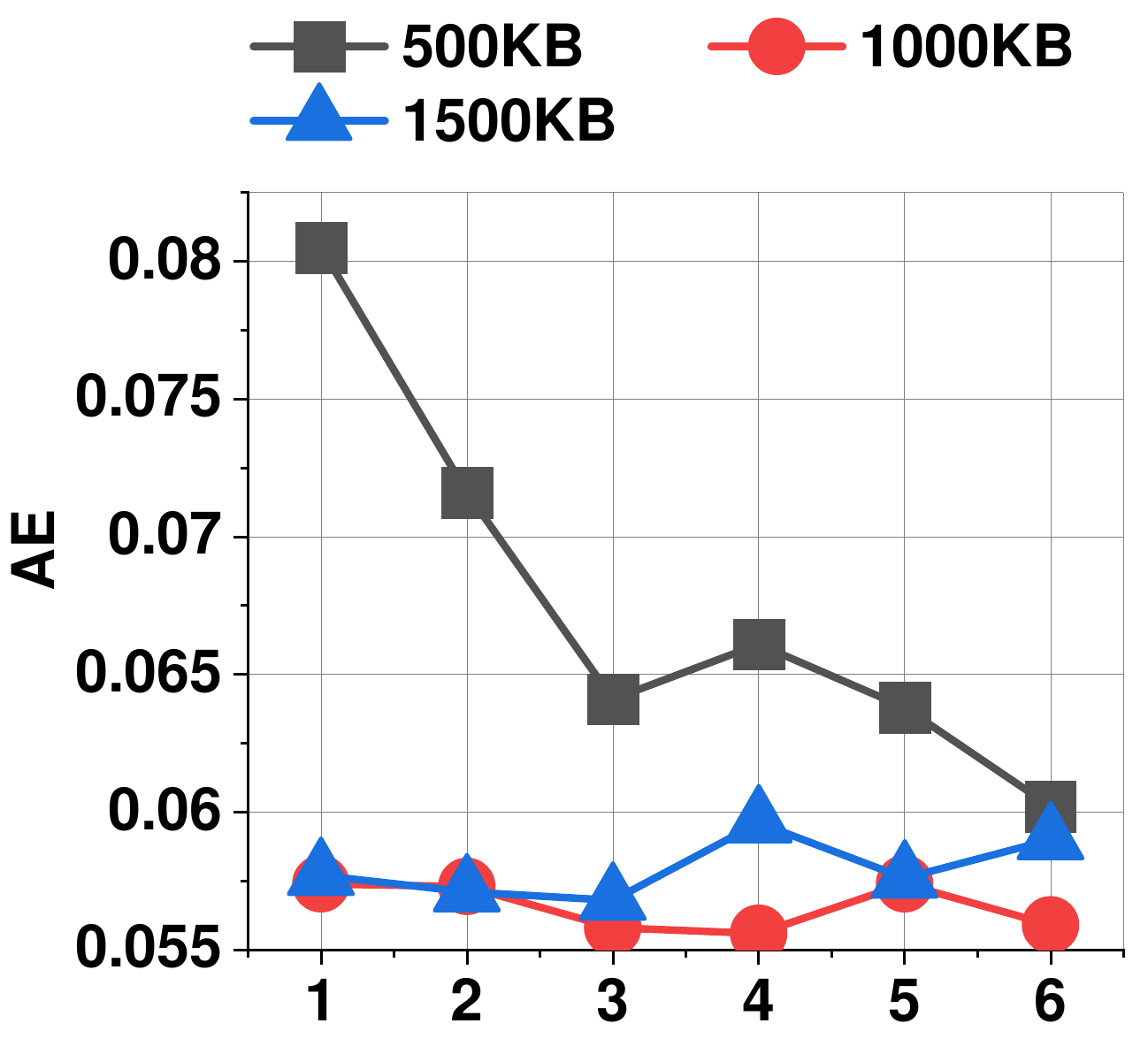}
			\vspace{-0.25in}
            \caption{Effects of $\lambda$}
            \label{fig::l}
		}
	\end{minipage}
	\begin{minipage}{0.216\textwidth}{
			\includegraphics[width=1\textwidth]{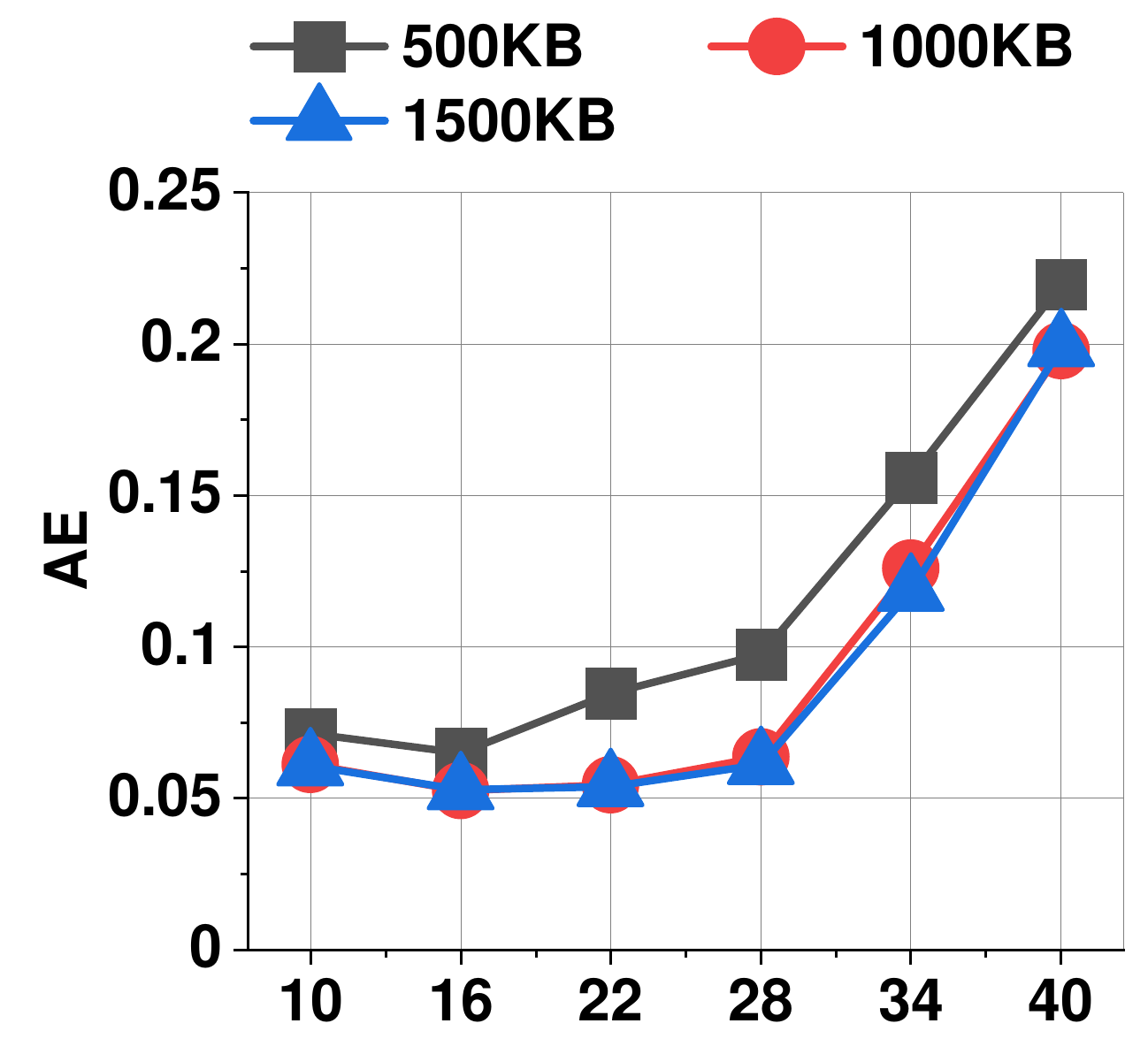}
			\vspace{-0.25in}
            \caption{Effects of $r$}
            \label{fig::r}
		}
	\end{minipage}
	\begin{minipage}{0.216\textwidth}{
			\includegraphics[width=1\textwidth]{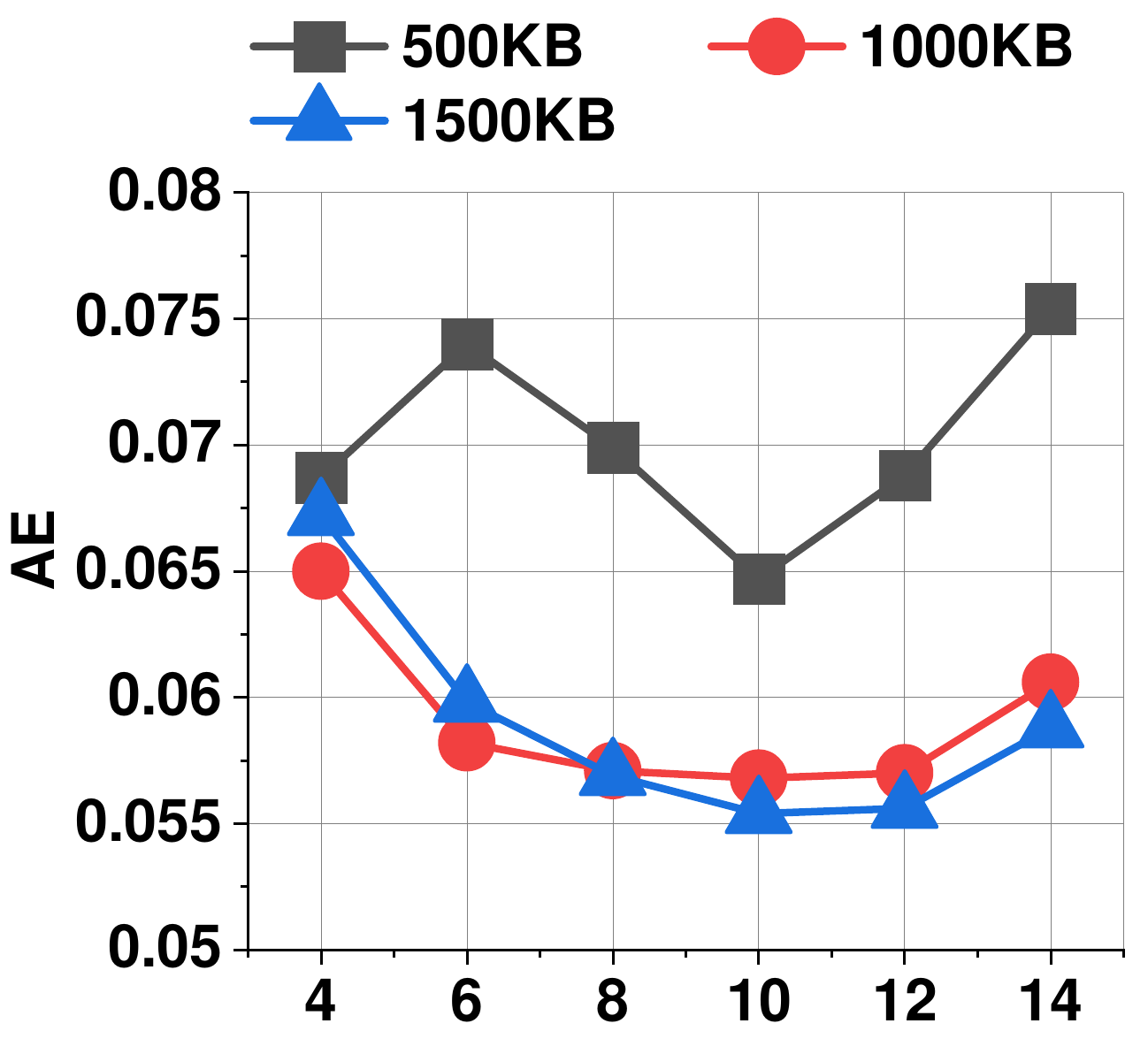}
			\vspace{-0.25in}
            \caption{Effects of $s$}
            \label{fig::s}
		}
	\end{minipage} \\
    \vspace{-0.2in}
 \end{figure*}

\noindent
\textbf{Implementation:} We implement \ourname{} and all other algorithms in C++. In all experiments, we use Bob Hash \cite{bobhash} with different hash seeds to implement the hash functions. 

\noindent
\textbf{Computation Platform:} We conducted all the experiments on a server with one 18-core processor (36 threads, Intel(R) Core(TM) i9-10980XE CPU @ 3.00GHz) and 128 GB DRAM memory. The processor has 64KB L1 cache, 1MB L2 cache for each core, and 24.75MB L3 cache shared by all cores.

\noindent
\textbf{Metrics:} 

\noindent
\textbf{1) Average Error (AE):} Suppose we query $w$-quantile of $k_1, \cdots, k_n$, and the quantile function at the query results are $\hat{w}_1, \cdots, \hat{w}_n$, the average error is defined as $\frac{1}{n}\sum _{i=1}^n |\hat{w}_i-w|$. 

\noindent
\textbf{2) Throughput:} We use million of operations (insertions and queries) per second (Mops) to measure the throughput. We repeat the experiment for 10 times and calculate the average results as our throughput.  

\noindent
\textbf{Datasets:}

\noindent
\textbf{1) CAIDA Dataset:} This Dataset is streams of anonymous IP traces collected from 2016 by CAIDA \cite{caida}. We regard the interval of two consecutive packets as its value. We use 20 million items. 

\noindent
\textbf{2) Campus Dataset:} The Campus Dataset is comprised of IP packets captured from the network of our campus. We regard the flow\footnote{A flow is generally defined as a part of the five-tuple: source IP address, destination IP address, source port, destination port, and protocol.} ID as the key, and latency as the value. We use 14 million items

\noindent
\textbf{3) Seattle Dataset:} The Seattle Dataset \cite{cappos2009seattle, zhu2016network} consists of round-trip times (RTT) between 99 nodes in the Seattle network systems. We treat RTTs from the same node as the same key, and RTTs as the value.

\noindent
\textbf{4) Synthetic Dataset:} We generate the Synthetic Dataset following the Pareto distribution with $\alpha=1$. We regard the interval between two consecutive items with the same key as its value. We use 20 million items.

\noindent
\textbf{Baseline Solution:} In single-key situation, we compare the performance of \ourname{} with GK, KLL, DDSketch, $t$-digest and ReqSketch. In per-key situation, we compare the performance of \ourname{} with GK, KLL, DDSketch, $t$-digest, ReqSketch, SQUAD and SketchPolymer. Since GK, KLL, DDSketch, $t$-digest and ReqSketch do not focus on estimating per-key quantiles, we make a small modification to apply these algorithms for our problem, which is similar to \cite{guo2023sketchpolymer}: For every algorithm, suppose it consumes $m$ memory in single-key situation, and we allocate $M$ memory on aggregate, then we construct $\frac{M}{m}$ buckets, and each bucket contains a copy of the data structure. For every coming item in the data stream, we use one hash function to map its key into a bucket, and record its value in the mapping bucket. Items mapped into the same bucket can be viewed to share the same key. Thus, to query the $w$-quantile of a key, we use the same hash function to map the key into a bucket, and calculate the quantile using these algorithms. 

 \begin{figure*}
 \begin{center}
     \begin{minipage}{.6\textwidth}
        \includegraphics[width=1\textwidth]{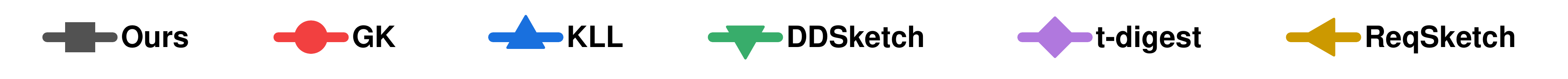}
    \end{minipage} 
 \end{center}
    \vspace{-0.1in}
 \begin{minipage}{.48\textwidth}
\flushright
    \subfigure[CAIDA]{
	\begin{minipage}{0.45\textwidth}{
			\includegraphics[width=1\textwidth]{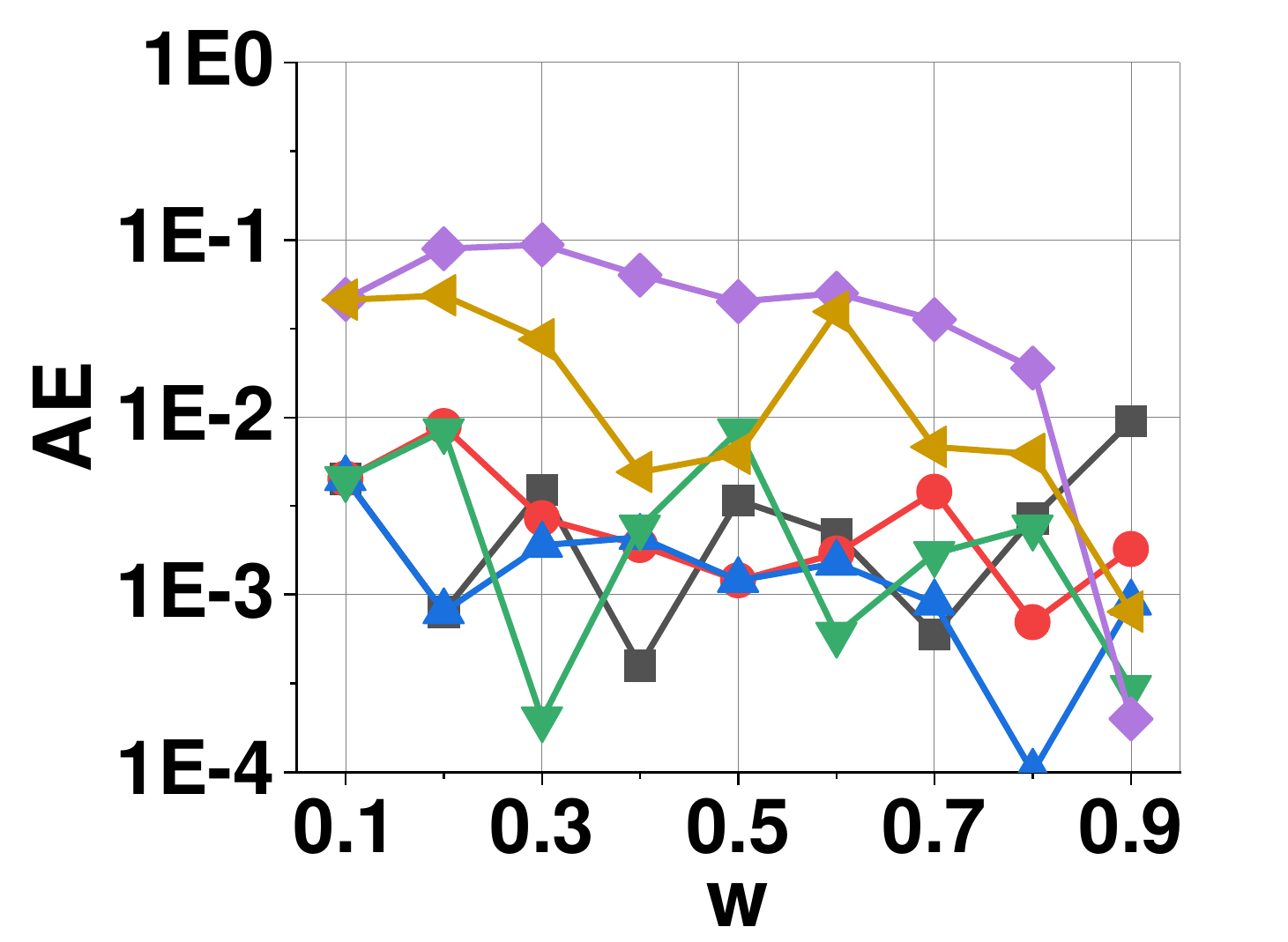}
            \vspace{-0.05in}
		}
	\end{minipage}}
    \subfigure[Synthetic]{
	\begin{minipage}{0.45\textwidth}{
			\includegraphics[width=1\textwidth]{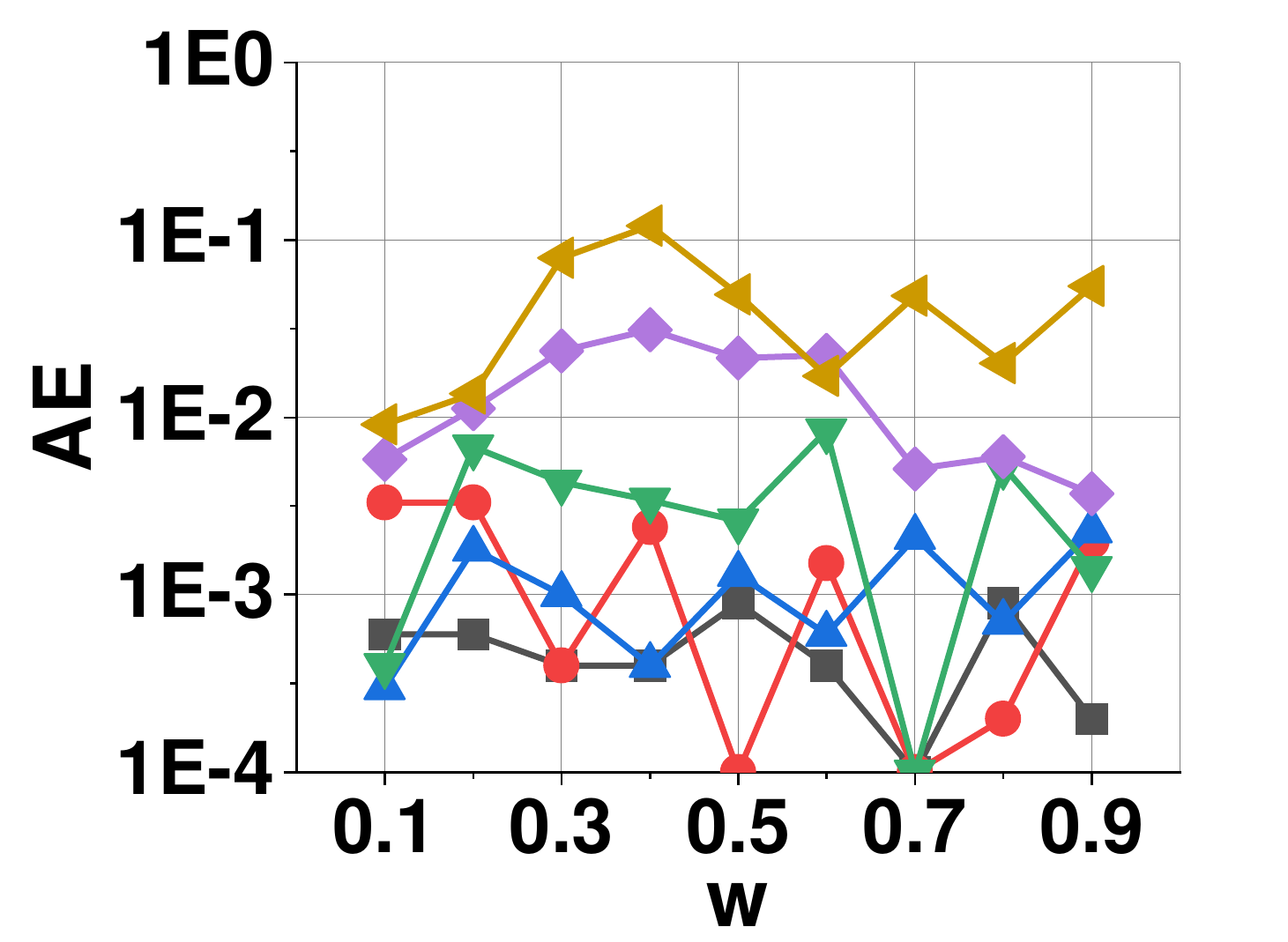}
            \vspace{-0.05in}
		}
	\end{minipage}}
    \caption{AE in Single-key Situation}
    \label{fig::single::ae}
    \end{minipage}
    \begin{minipage}{.48\textwidth}
    \subfigure[CAIDA]{
	\begin{minipage}{0.45\textwidth}{
			\includegraphics[width=1\textwidth]{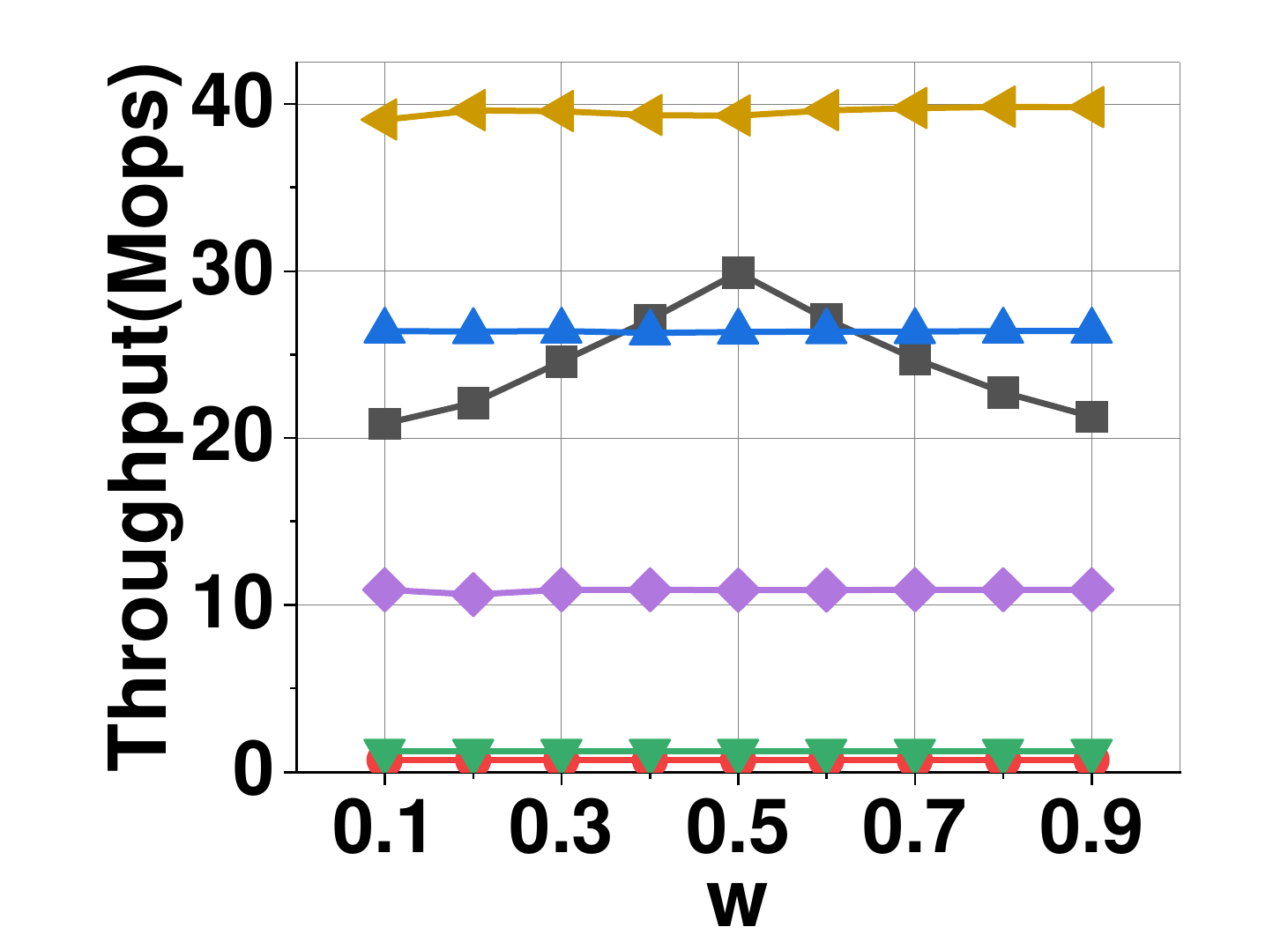}
            \vspace{-0.05in}
		}
	\end{minipage}}
    \subfigure[Synthetic]{
	\begin{minipage}{0.45\textwidth}{
			\includegraphics[width=1\textwidth]{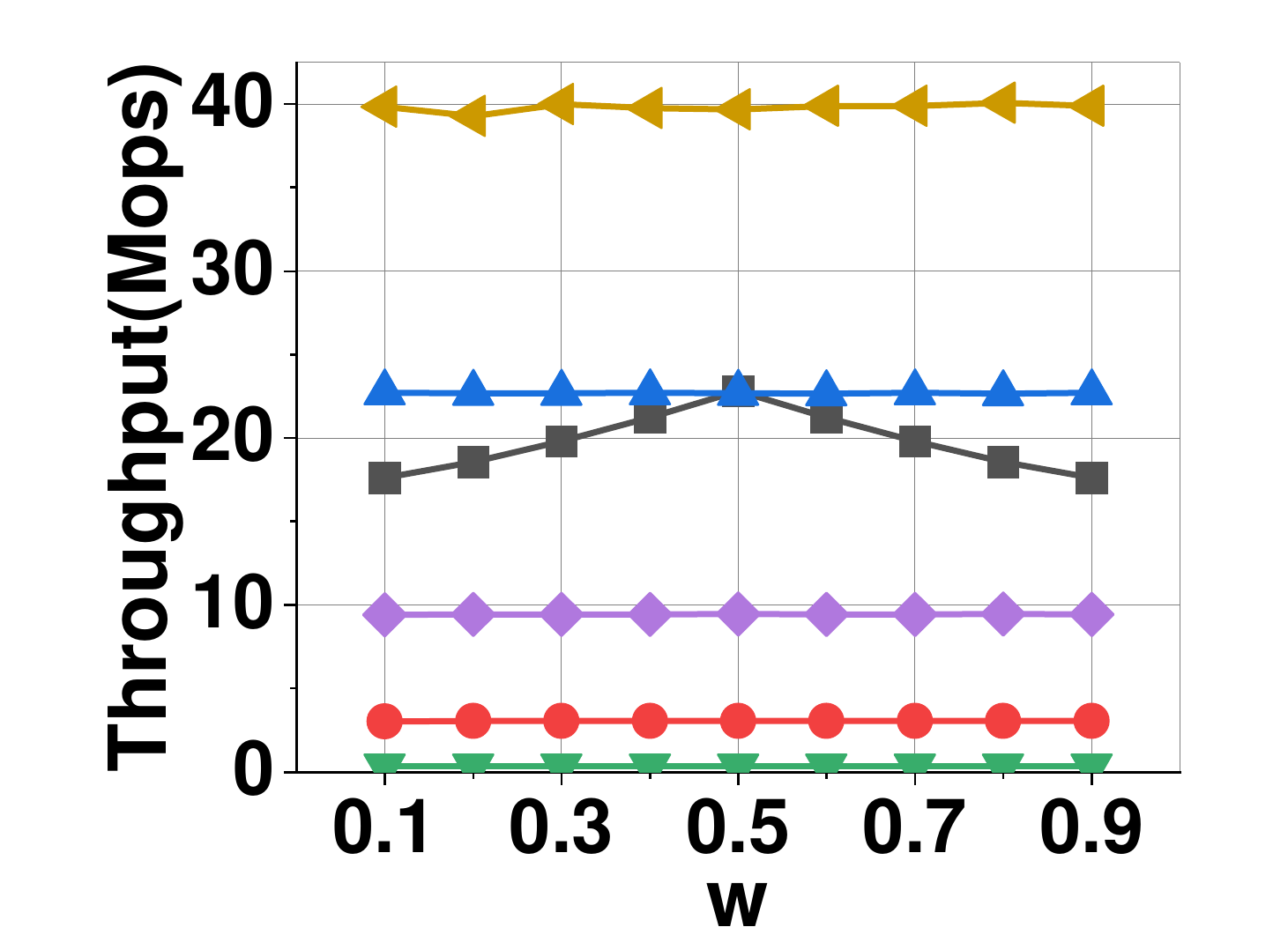}
            \vspace{-0.05in}
		}
	\end{minipage}}
    \caption{Insertion Throughput in Single-key Situation}
    \label{fig::single::thp}
    \end{minipage} \\
 \vspace{-0.2in}
\end{figure*}

\begin{figure*}[!ht]
	\centering
    \begin{minipage}{.8\textwidth}
        \includegraphics[width=1\textwidth]{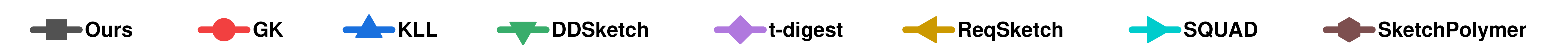}
    \end{minipage} 
    \vspace{-0.1in}
    \\
    \subfigure[900KB]{
	\begin{minipage}{0.216\textwidth}{
			\includegraphics[width=1\textwidth]{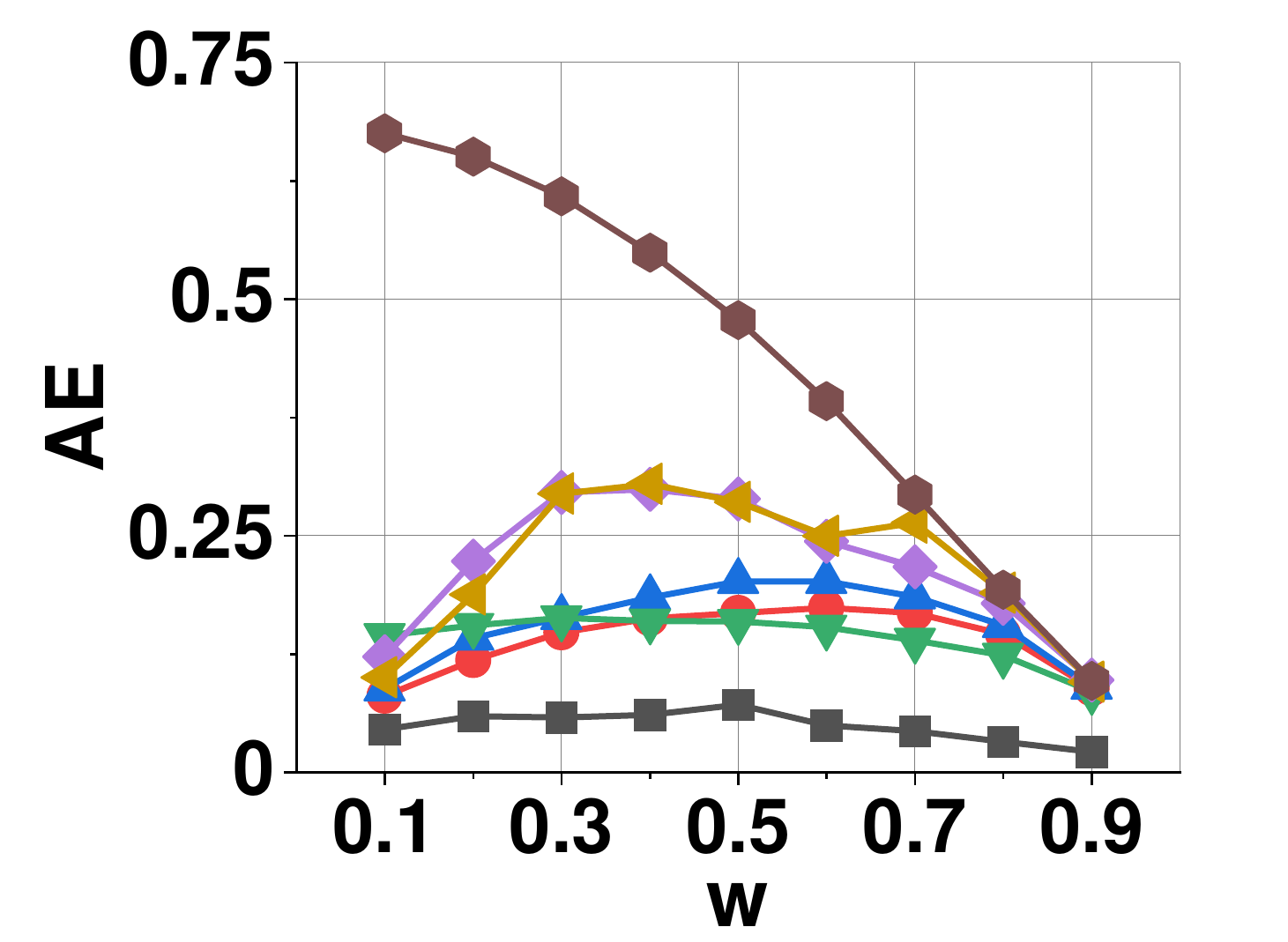}
            \vspace{-0.05in}
		}
	\end{minipage}}
    \subfigure[1100KB]{
	\begin{minipage}{0.216\textwidth}{
			\includegraphics[width=1\textwidth]{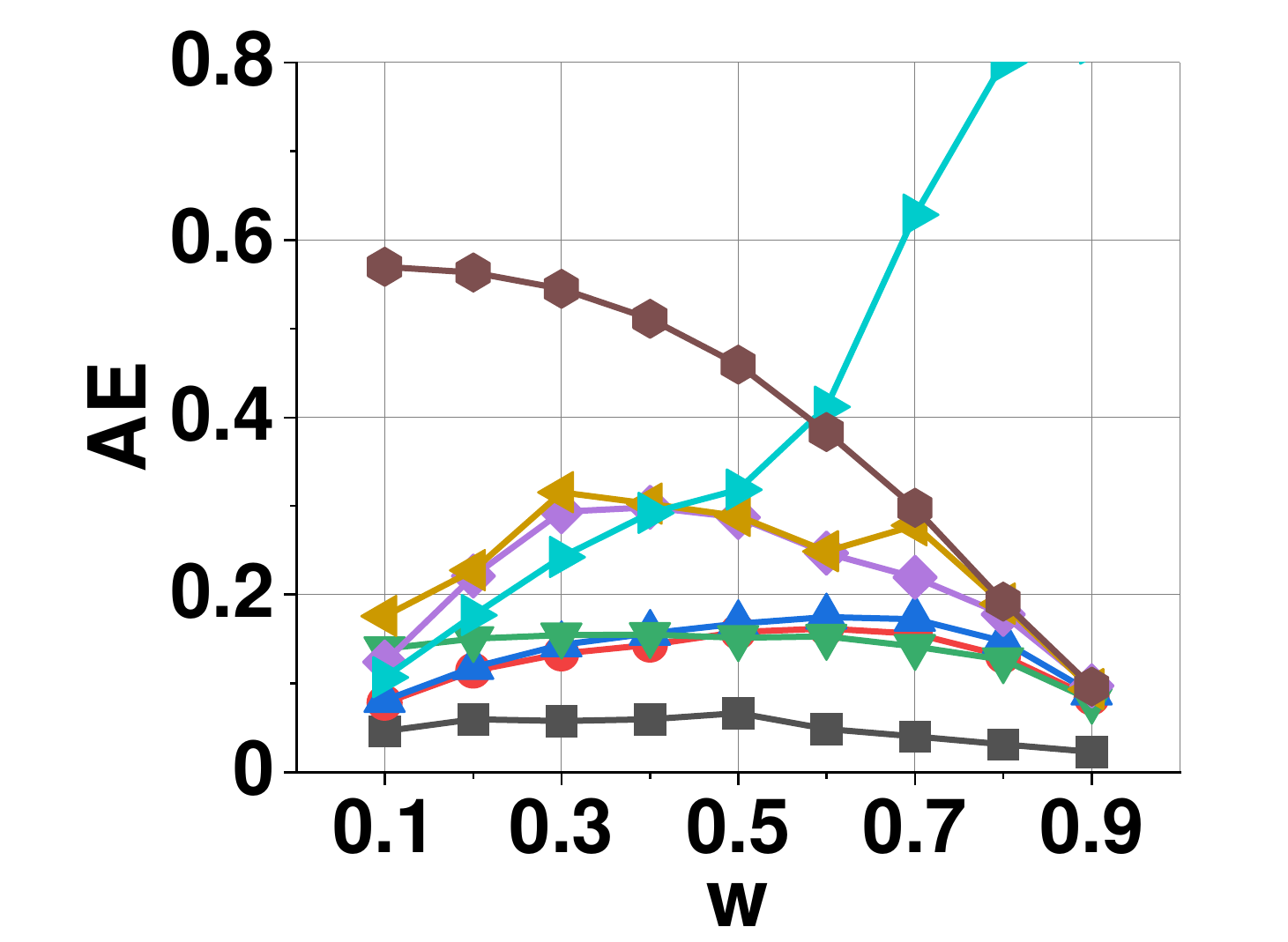}
            \vspace{-0.05in}
		}
	\end{minipage}}
    \subfigure[1300KB]{
	\begin{minipage}{0.216\textwidth}{
			\includegraphics[width=1\textwidth]{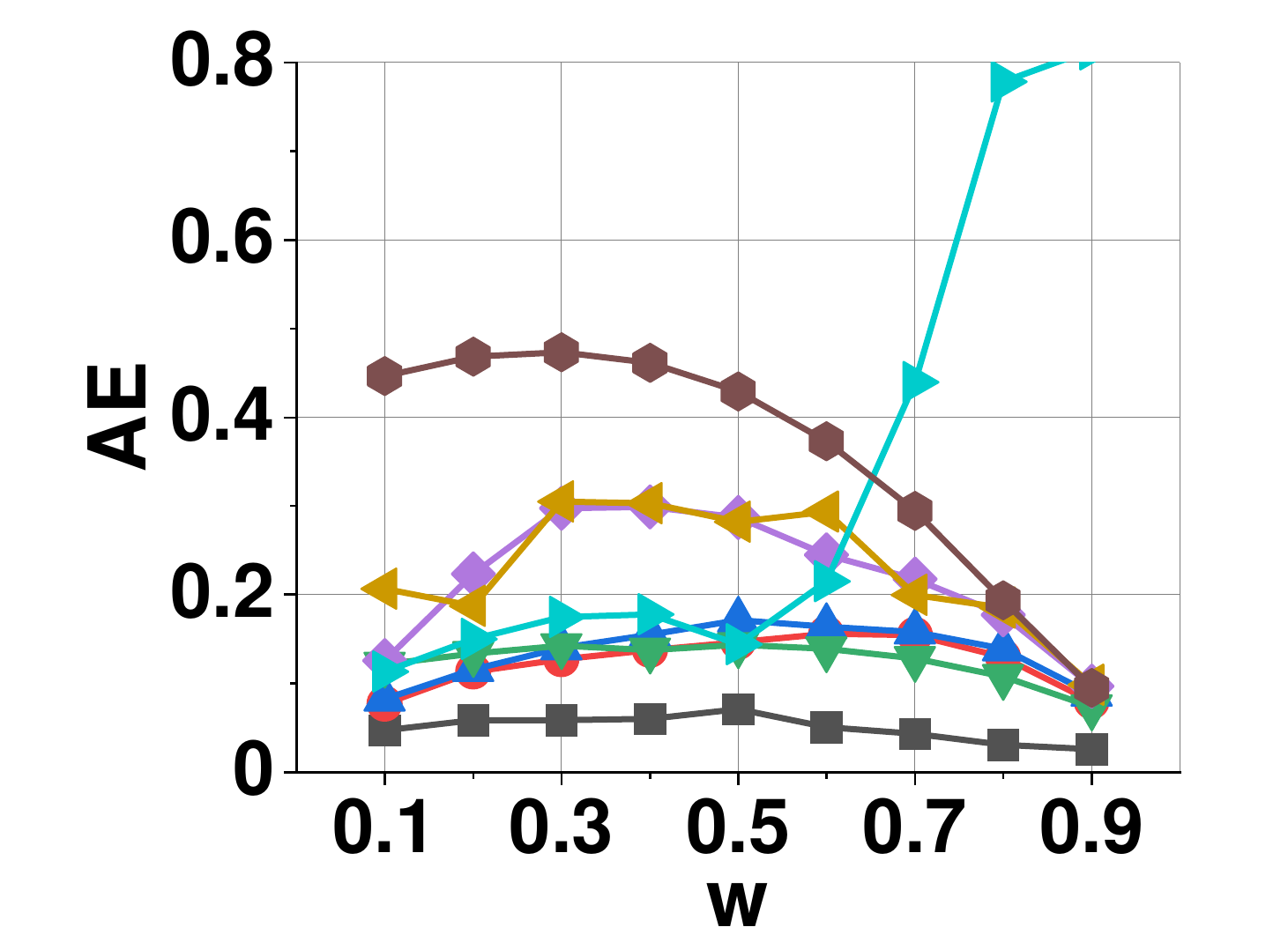}
            \vspace{-0.05in}
		}
	\end{minipage}}
    \subfigure[1500KB]{
	\begin{minipage}{0.216\textwidth}{
			\includegraphics[width=1\textwidth]{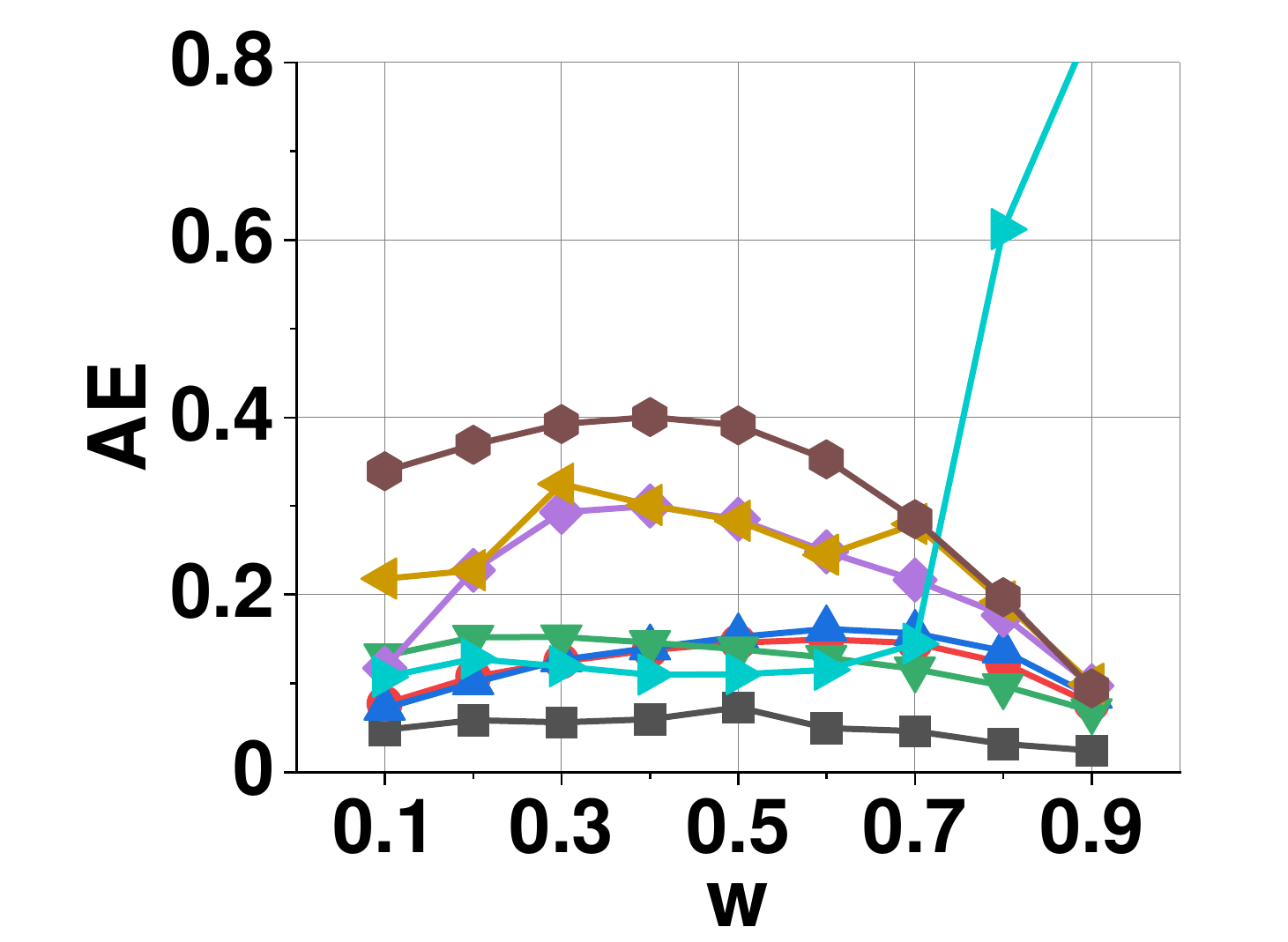}
            \vspace{-0.05in}
		}
	\end{minipage}}
    \caption{AE on CAIDA Dataset in Per-key Situation}
    \vspace{-0.15in}
    \label{fig::per::caida}
\end{figure*}

\begin{figure*}[!ht]
	\centering
    \begin{minipage}{.8\textwidth}
        \includegraphics[width=1\textwidth]{Exp_new/8legend.pdf}
    \end{minipage} 
    \vspace{-0.1in}
    \\
    \subfigure[900KB]{
	\begin{minipage}{0.216\textwidth}{
			\includegraphics[width=1\textwidth]{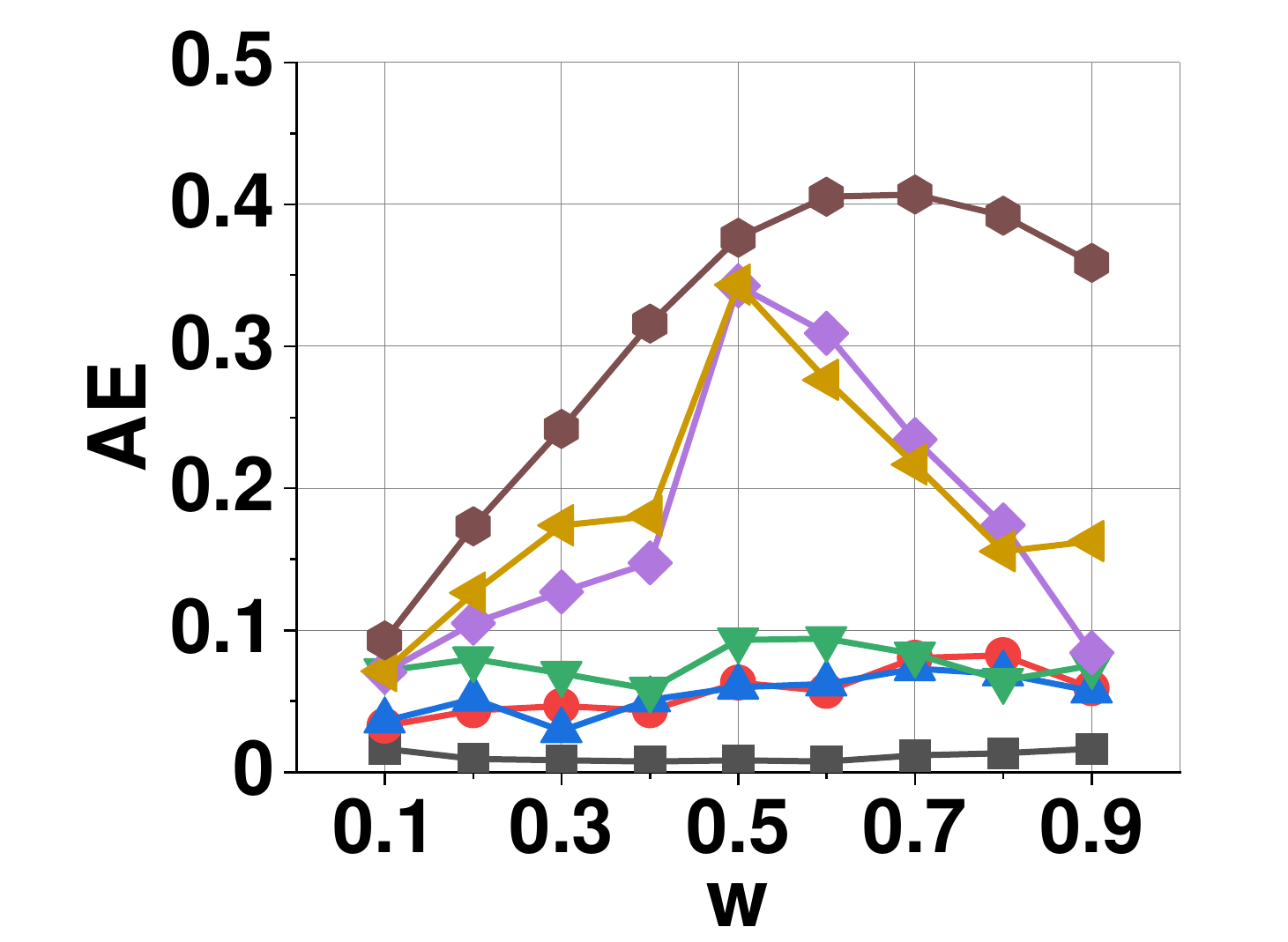}
            \vspace{-0.1in}
		}
	\end{minipage}}
    \subfigure[1100KB]{
	\begin{minipage}{0.216\textwidth}{
			\includegraphics[width=1\textwidth]{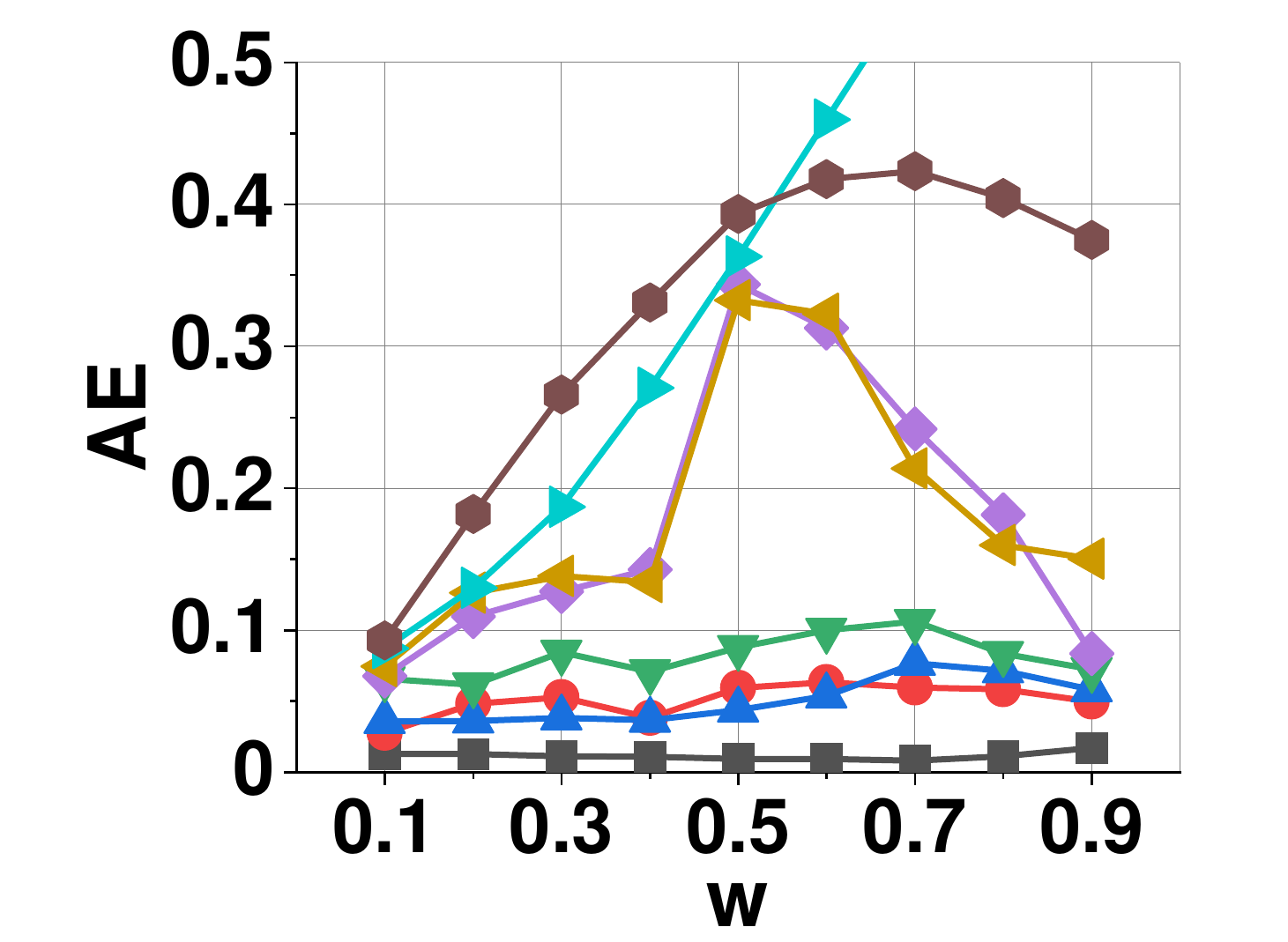}
            \vspace{-0.1in}
		}
	\end{minipage}}
    \subfigure[1300KB]{
	\begin{minipage}{0.216\textwidth}{
			\includegraphics[width=1\textwidth]{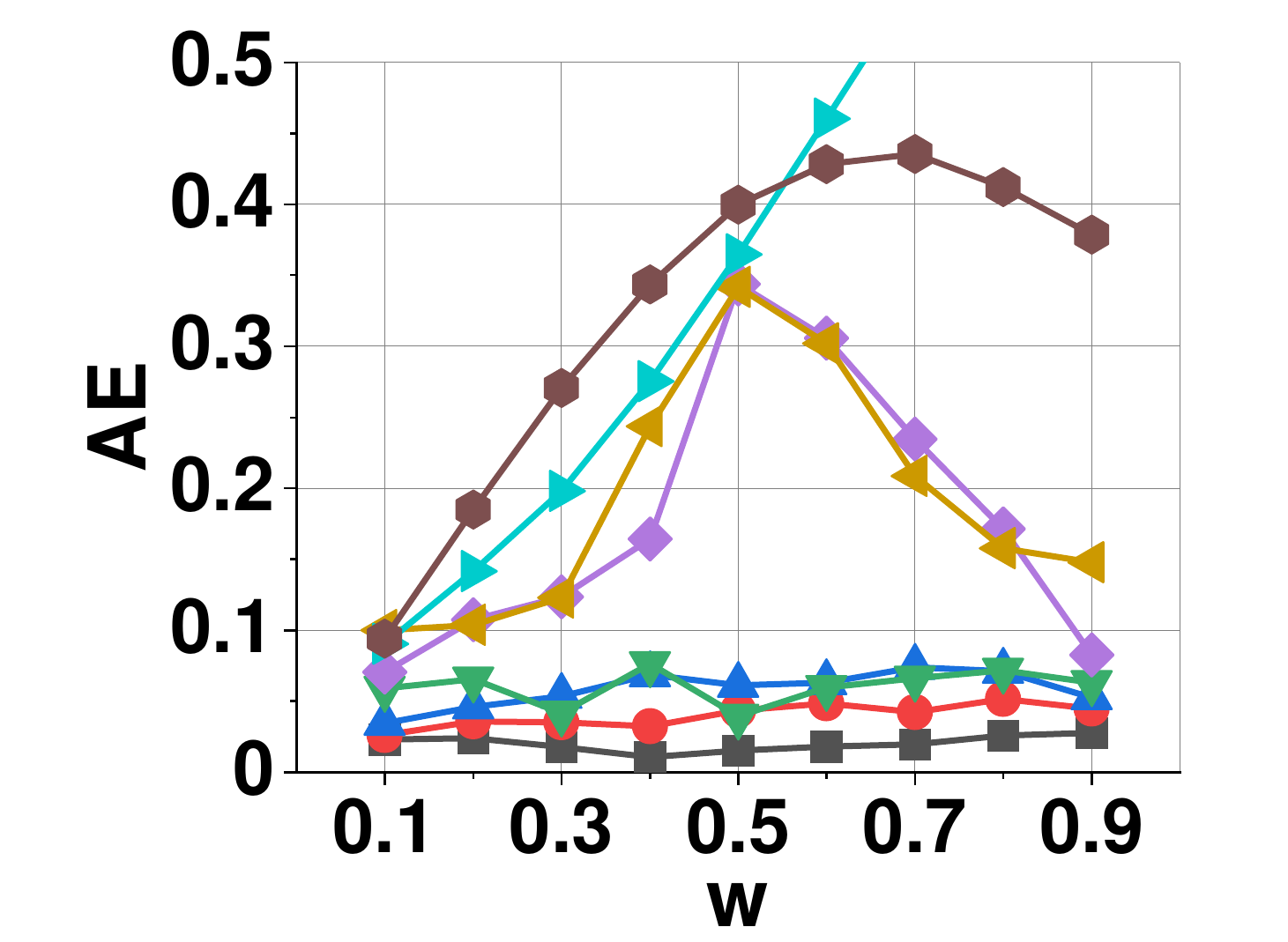}
            \vspace{-0.1in}
		}
	\end{minipage}}
    \subfigure[1500KB]{
	\begin{minipage}{0.216\textwidth}{
			\includegraphics[width=1\textwidth]{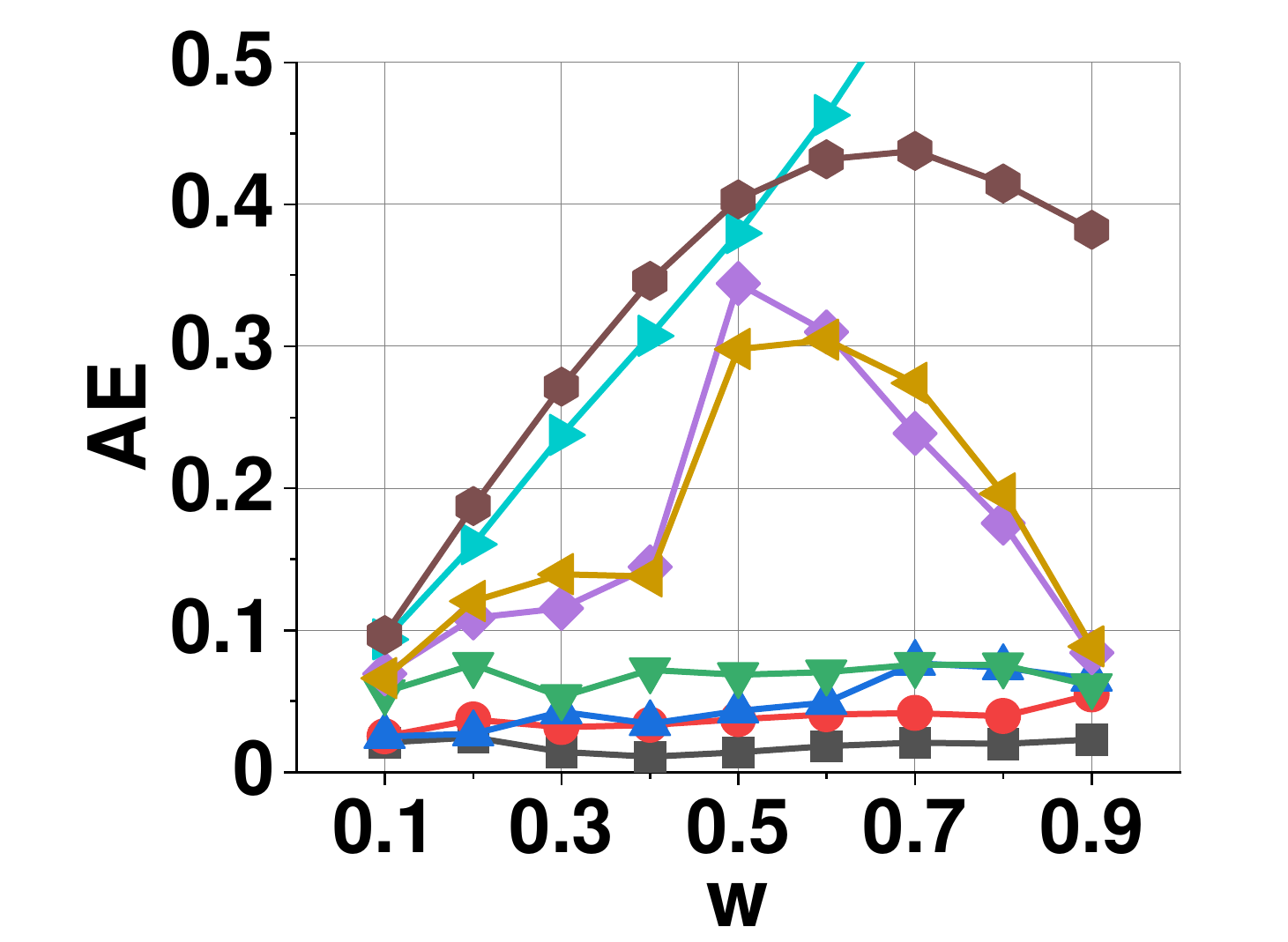}
            \vspace{-0.1in}
		}
	\end{minipage}}
    \caption{AE on Campus Dataset in Per-key Situation}
    \vspace{-0.15in}
    \label{fig::per::campus}
\end{figure*}

\subsection{Experiments on Parameter Settings}
\label{exp::param}
\vspace{-0.03in}

In this section, we measure the effects of some key parameters of
\ourname{}, namely, the number of cells in each bucket $d$,
the ratio of the memory size of \stageone{} to the memory size of the whole \ourname{} $q$, the threshold for \stageone{} $T$, the threshold for ratio in \stagetwo{} $\lambda$, the size of the Candidate $r$, and the size of the Representative $s$.
We do not measure the effect of the number of buckets in \stagetwo{} $u$, as $u$ can be calculated once the total memory and other parameters are fixed. 
We use CAIDA dataset in these experiments, and AE to measure these effects.

\textbf{Effects of the number of cells in each bucket $d$ (Figure \ref{fig::d}):} \textit{Experimental results show that the best option of $d$ is between 7 and 11.} We vary $d$ from $3$ to $13$ in each experiment, and results show that when the memory size is 500KB, \ourname{} achieves lowest AE when $d=7$. When the memory size is 1000/1500KB, \ourname{} performs best when $d=11$. In fact, users can tune the parameter $d$ to make a trade-off between accuracy and speed: a larger $d$ will make \ourname{} more accurate at the sacrifice of overall throughput, as \ourname{} has to traverse the bucket to find a cell when inserting an item in \stagetwo{}. Thus, we choose $d=7$ by default.

\textbf{Effects of the ratio of the memory size of \stageone{} to the memory size of the whole \ourname{} $q$ (Figure \ref{fig::q}):} \textit{Experimental results show that a small proportion of \stageone{} can take up the role of filtering infrequent items.} We vary $q$ from $0.05$ to $0.3$ in total memory from 500KB to 1500KB, and results show that the best choice of $q$ is between $0.1$ to $0.15$.  Taking into consideration that \stagetwo{} keeps the information of every frequent item, we allocate most memory to \stagetwo{}, so we set $q=0.1$ in other experiments. 

\textbf{Effects of the threshold for \stageone{} $T$ (Figure \ref{fig::t}):} \textit{Experimental results show that the best option of $T$ is among 20 to 40.} In each experiment, we try different values of $T$, and results show that the best option of $T$ is 30/40/20 within 500/1000/1500KB memory respectively. Since setting a larger $T$ can filter as many infrequent items as possible in \stageone{}, we choose $T=40$.

\textbf{Effects of the threshold for ratio in \stagetwo{} $\lambda$ (Figure \ref{fig::l}):} \textit{Experimental results show that the best option of $\lambda$ is between 3 to 6.} We set the memory from 500KB to 1500KB, and vary $\lambda$ from $1$ to $6$ in the experiment. The results show that the best option of $\lambda$ is 6/4/3 within 500/1000/1500KB memory respectively. If $\lambda$ is too small, items might be replaced too frequently, which results in the loss of information in the data stream. If $\lambda$ is too large, the mechanism of Ostracism wouldn't do much good. As a result, we choose $\lambda=4$.

\textbf{Effects of the size of the Candidate $r$ (Figure \ref{fig::r}):} \textit{Experimental results show that the best $r$ is 16.} In order to choose 2 median value samples to insert to the Representative when the Candidate is full, the size of the Candidate should be an even number. We vary $r$ from $10$ to $40$, and experimental results show that the best option of $r$ is always 16 when total memory is 500/1000/1500KB. In fact, if we keep the total memory of \stagetwo{} unchanged, then a larger $r$ will result in a smaller $u$, so hash collision will be inevitable. As a result, we set $r=16$ by default. 

\textbf{Effects of the size of the Representative $s$ (Figure \ref{fig::s}):} \textit{Experimental results show that the best choice of $s$ is 10.} Since we insert two numbers into the Representative every time, we should set $s$ to be an even number, and we vary it from 4 to 14. The experimental results show that the best option of $s$ is always 10 when the total memory is 500/1000/1500KB. In fact, the more value samples the Representative has, the fewer buckets \stagetwo{} contains. To ensure that there are enough buckets in \stagetwo{}, we set $s=10$ by default.

\begin{figure*}[!ht]
	\centering
    \vspace{-0.05in}
    \begin{minipage}{.8\textwidth}
        \includegraphics[width=1\textwidth]{Exp_new/8legend.pdf}
    \end{minipage} 
    \vspace{-0.1in}
    \\
    \subfigure[50KB]{
	\begin{minipage}{0.216\textwidth}{
			\includegraphics[width=1\textwidth]{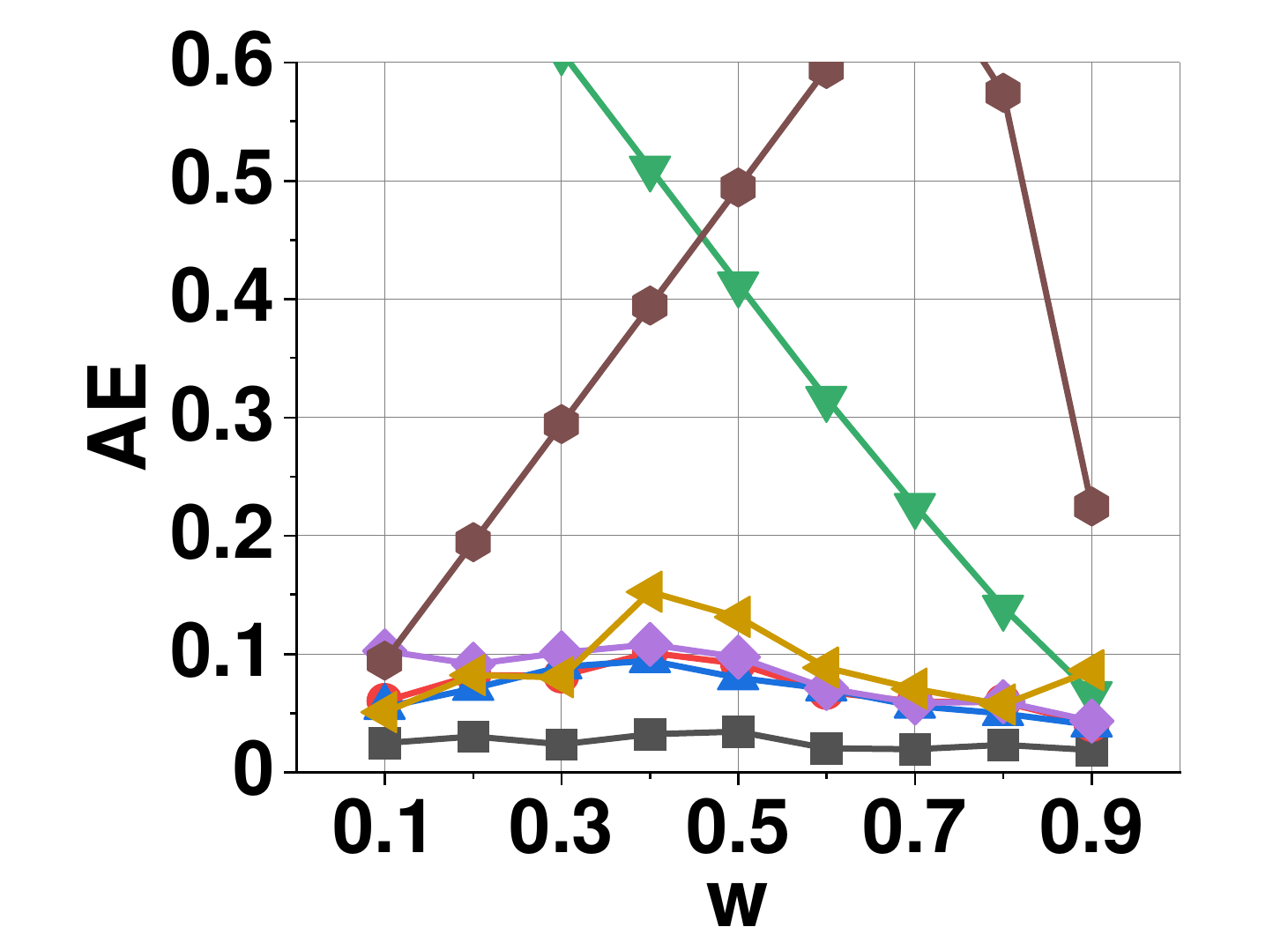}
            \vspace{-0.05in}
		}
	\end{minipage}}
    \subfigure[60KB]{
	\begin{minipage}{0.216\textwidth}{
			\includegraphics[width=1\textwidth]{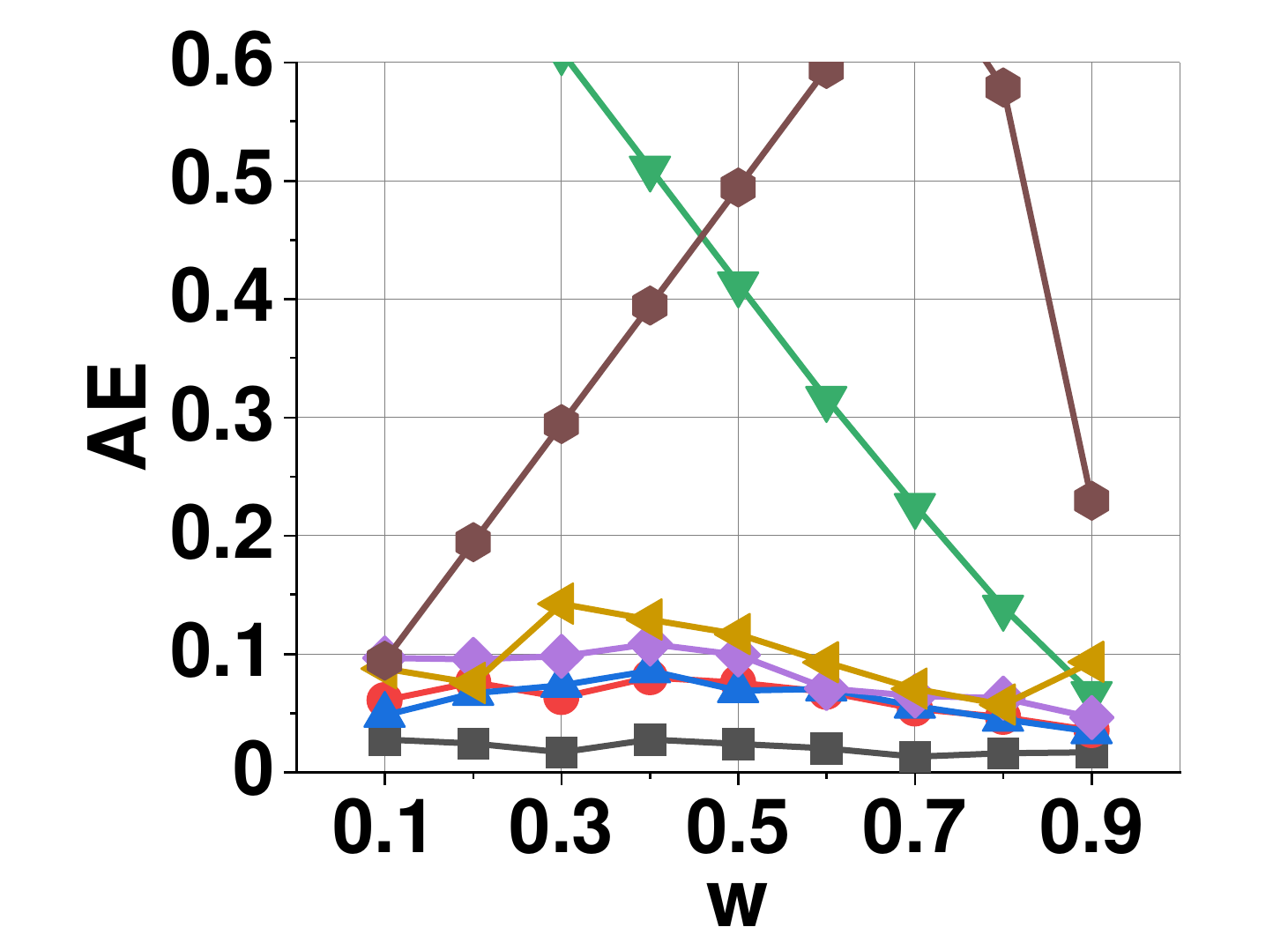}
            \vspace{-0.05in}
		}
	\end{minipage}}
    \subfigure[70KB]{
	\begin{minipage}{0.216\textwidth}{
			\includegraphics[width=1\textwidth]{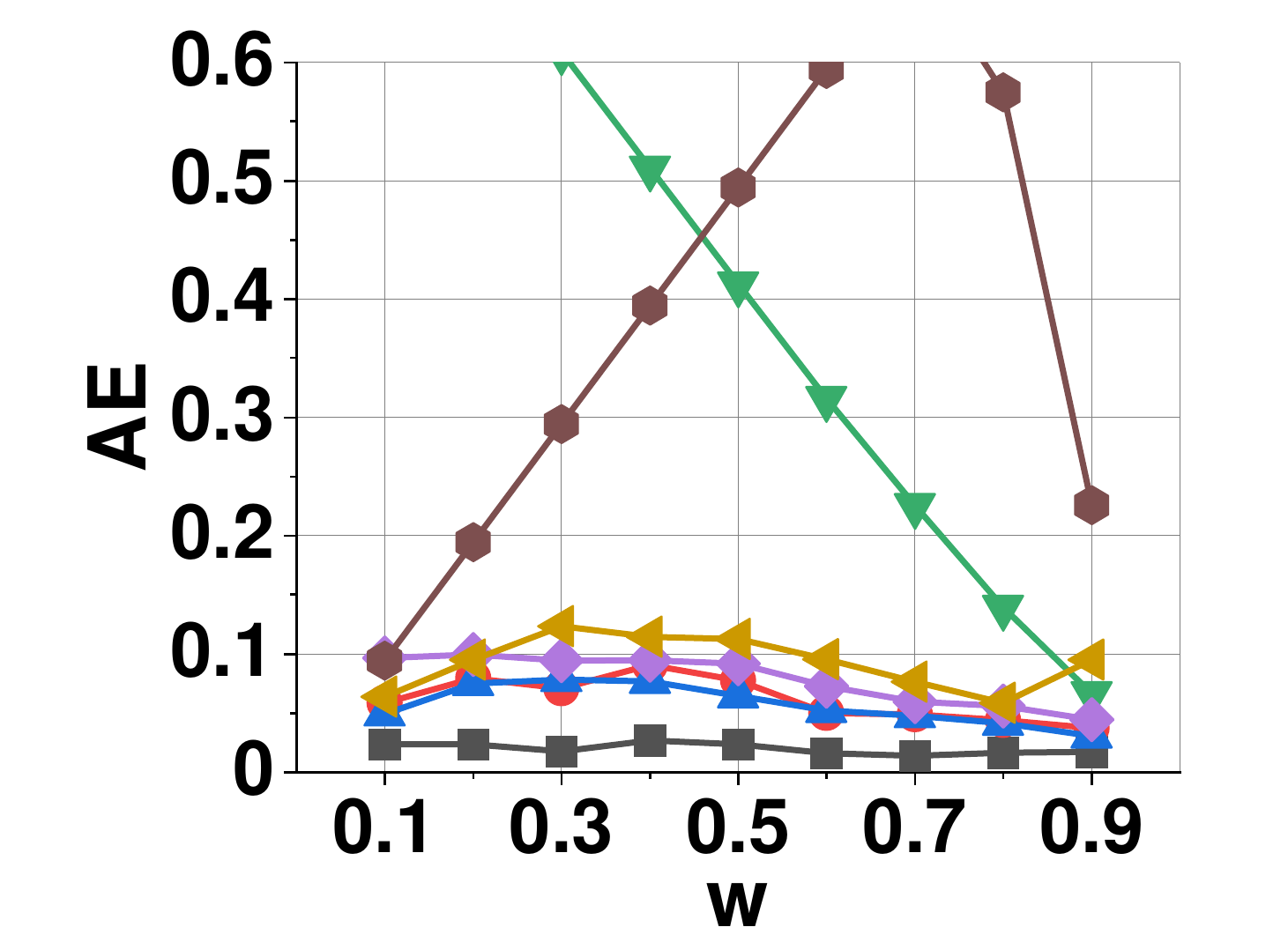}
            \vspace{-0.05in}
		}
	\end{minipage}} 
    \subfigure[80KB]{
	\begin{minipage}{0.216\textwidth}{
			\includegraphics[width=1\textwidth]{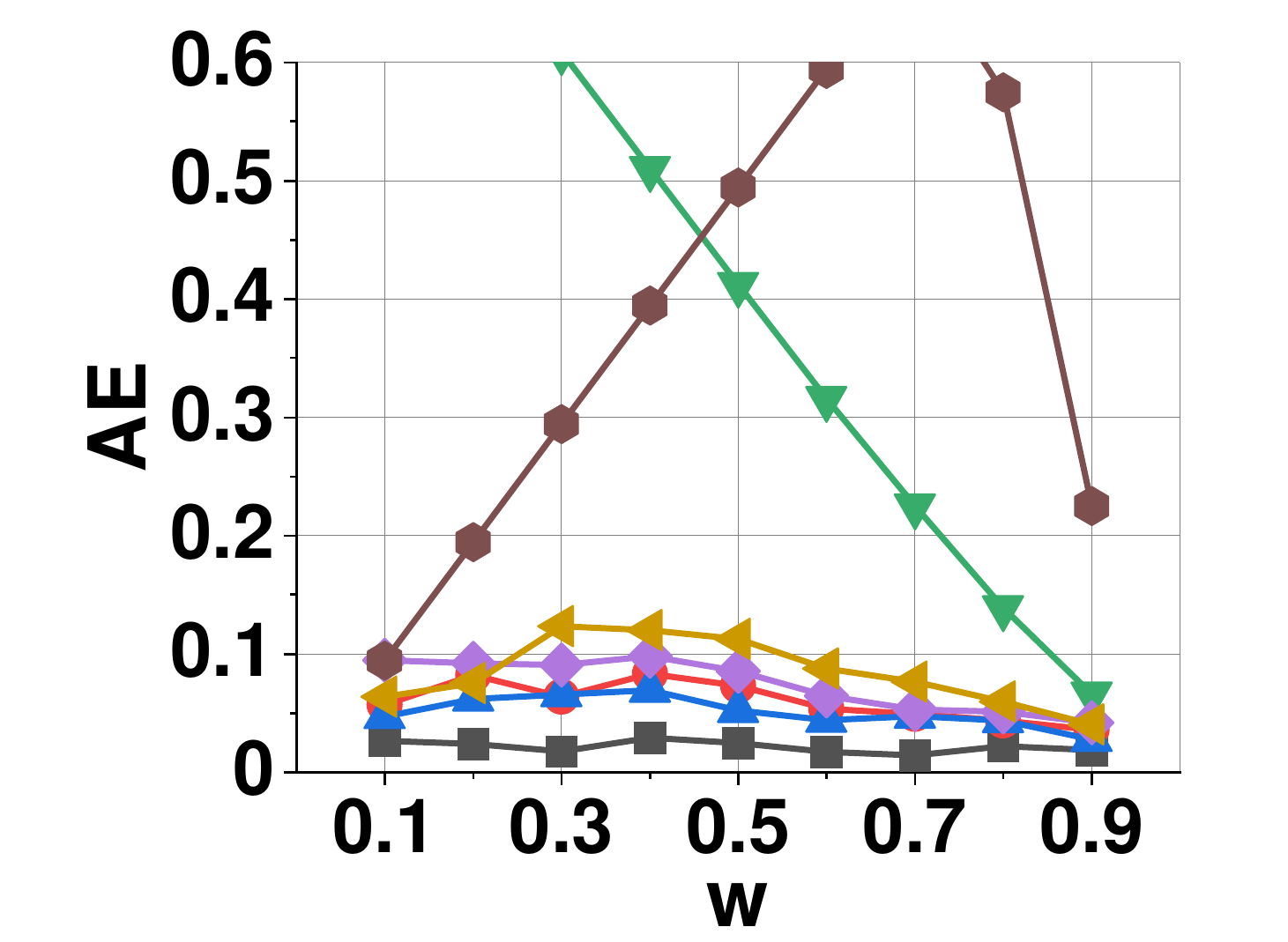}
            \vspace{-0.05in}
		}
	\end{minipage}}
    \caption{AE on Seattle Dataset in Per-key Situation}
    \vspace{-0.15in}
    \label{fig::per::seattle}
\end{figure*}

\begin{figure*}[!ht]
	\centering
  \vspace{-0.05in}
    \begin{minipage}{.8\textwidth}
        \includegraphics[width=1\textwidth]{Exp_new/8legend.pdf}
    \end{minipage} 
    \vspace{-0.1in}
    \\
    \subfigure[50KB]{
	\begin{minipage}{0.216\textwidth}{
			\includegraphics[width=1\textwidth]{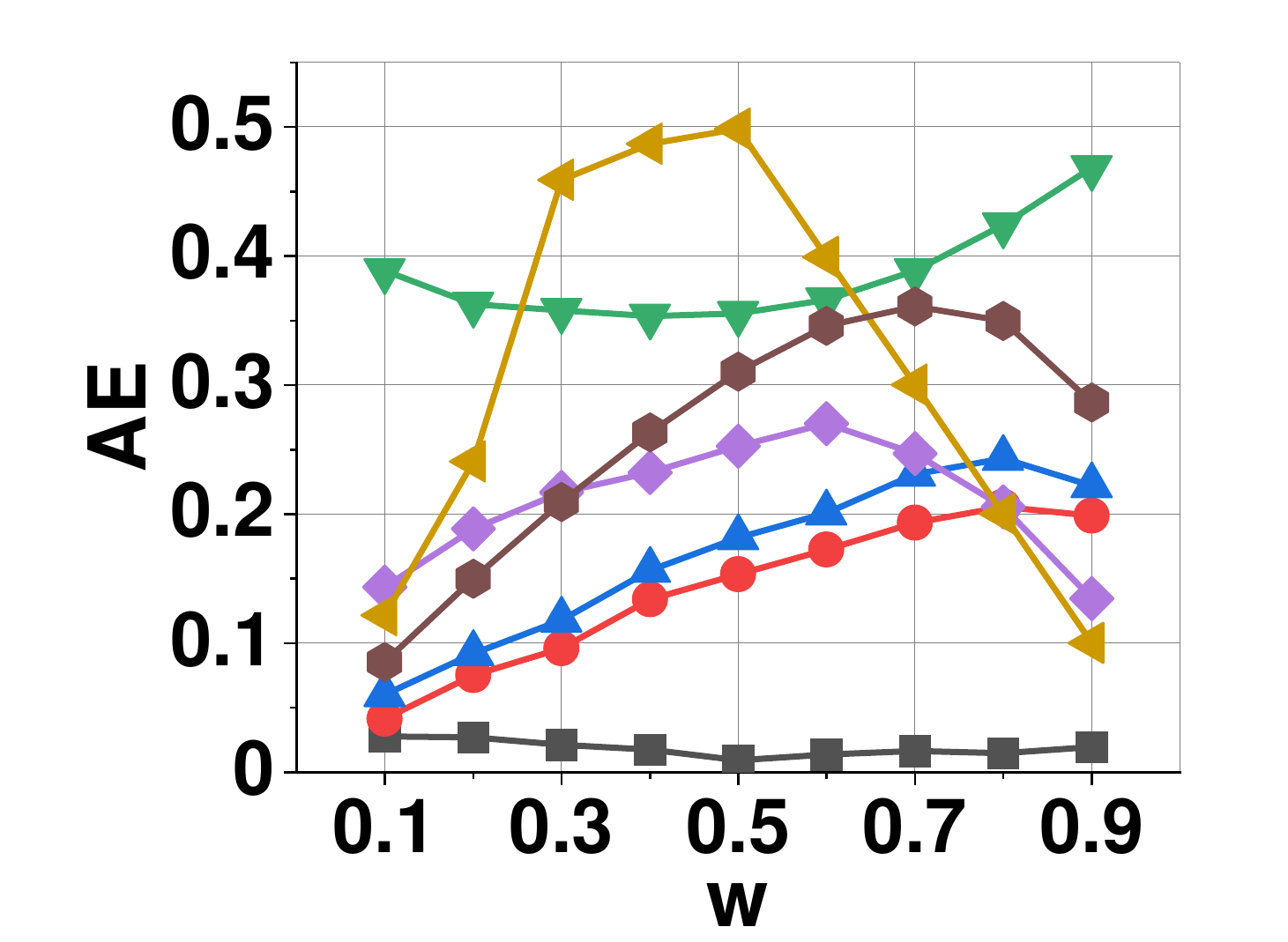}
            \vspace{-0.05in}
		}
	\end{minipage}}
    \subfigure[60KB]{
	\begin{minipage}{0.216\textwidth}{
			\includegraphics[width=1\textwidth]{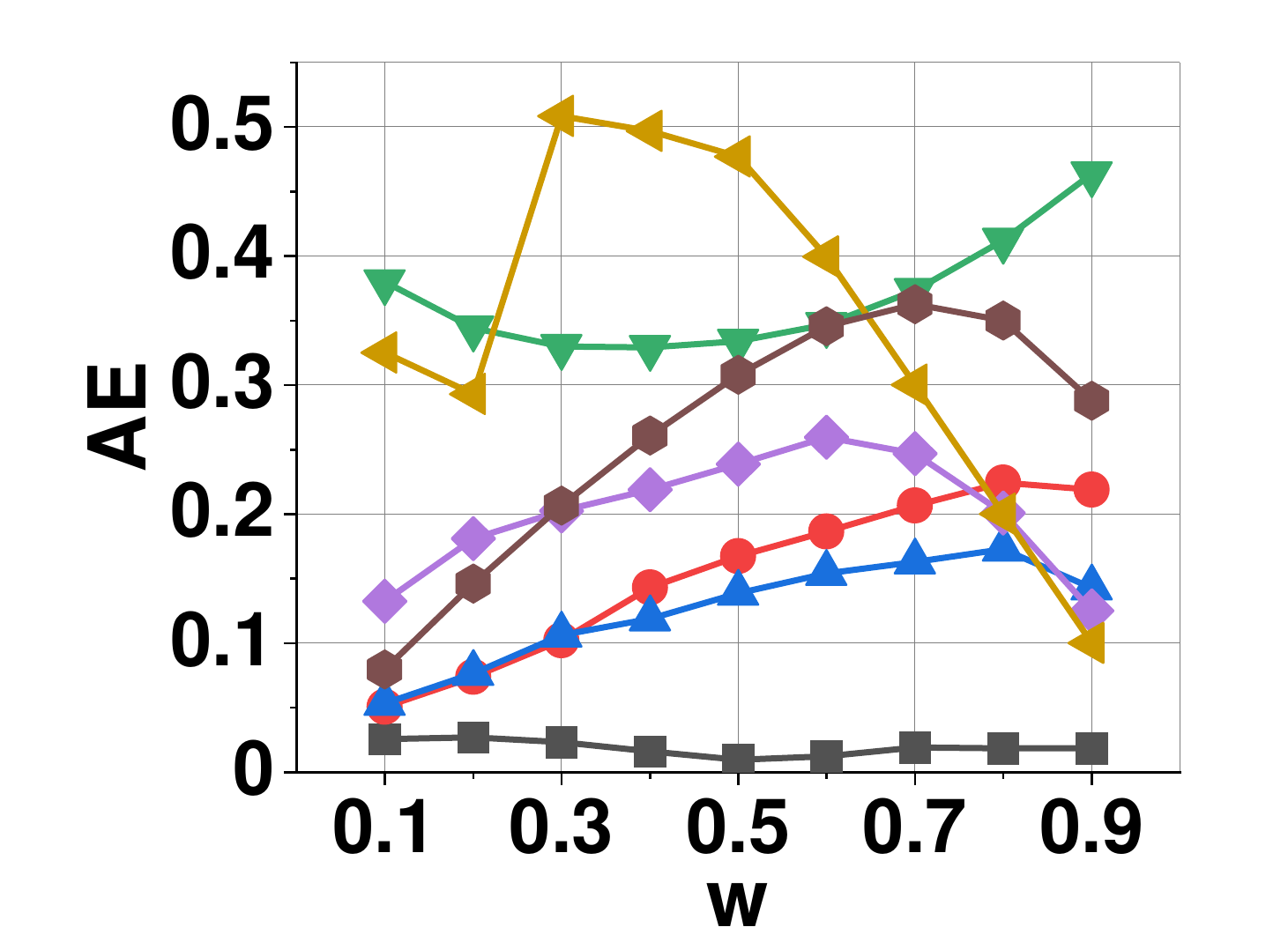}
            \vspace{-0.05in}
		}
	\end{minipage}}
    \subfigure[70KB]{
	\begin{minipage}{0.216\textwidth}{
			\includegraphics[width=1\textwidth]{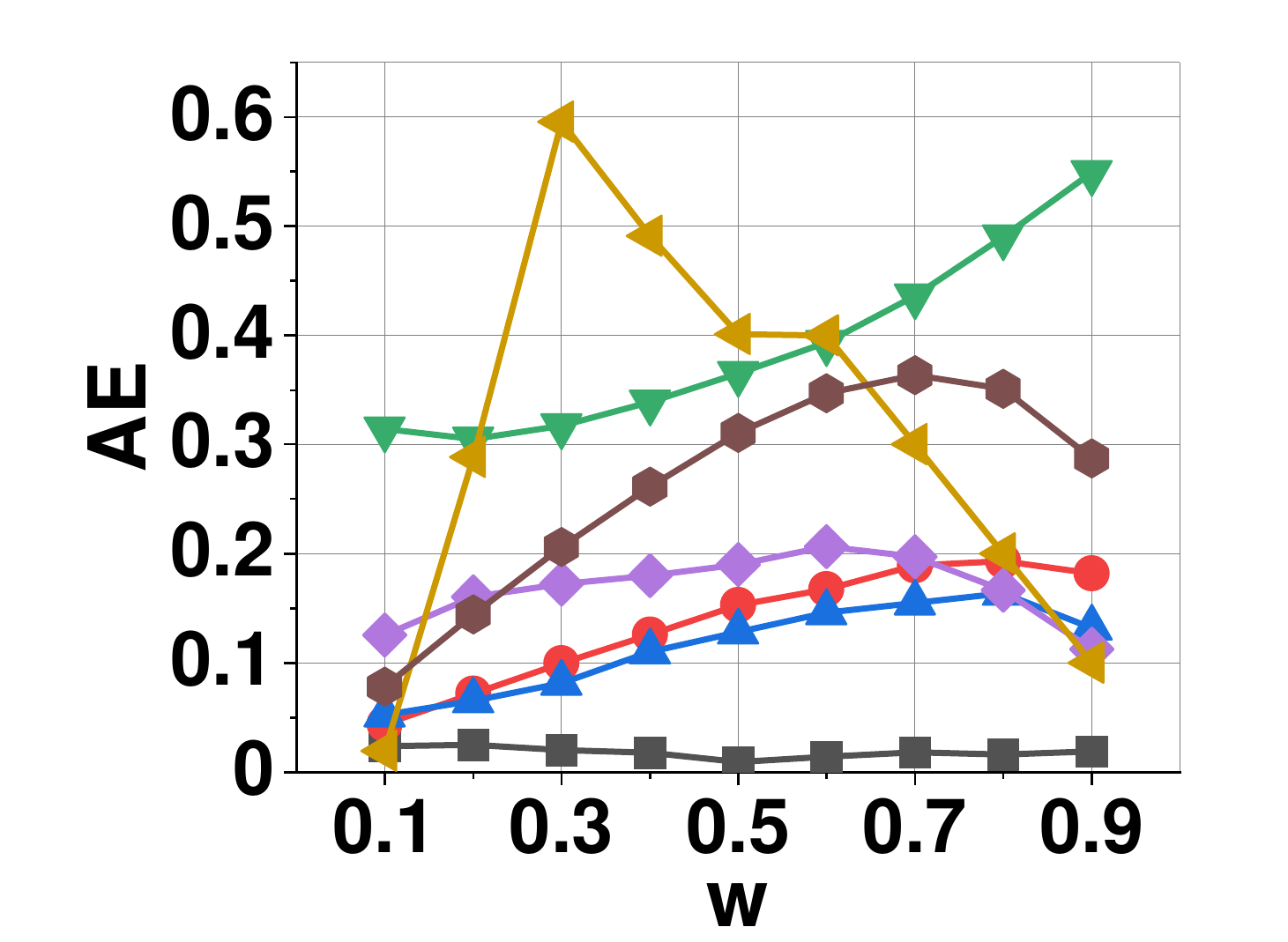}
            \vspace{-0.05in}
		}
	\end{minipage}}
    \subfigure[80KB]{
	\begin{minipage}{0.216\textwidth}{
			\includegraphics[width=1\textwidth]{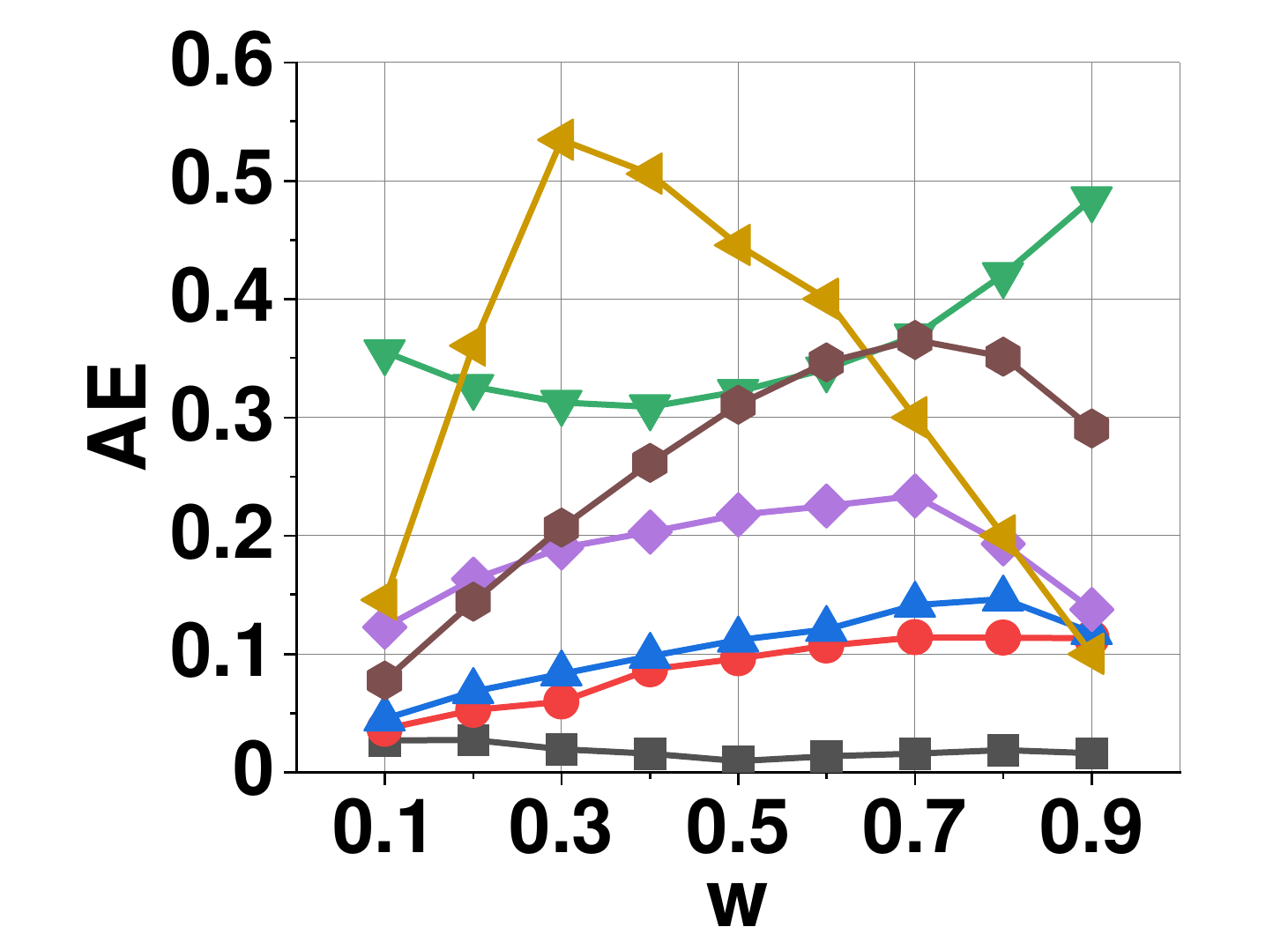}
            \vspace{-0.05in}
		}
	\end{minipage}}
    \caption{AE on Synthetic Dataset in Per-key Situation}
    \vspace{-0.15in}
    \label{fig::per::stn}
\end{figure*}

\begin{figure*}[!ht]
	\centering
   \vspace{-0.05in}
    \begin{minipage}{.8\textwidth}
        \includegraphics[width=1\textwidth]{Exp_new/8legend.pdf}
    \end{minipage} 
    \vspace{-0.1in}
    \\
    \subfigure[CAIDA (1500KB)]{
	\begin{minipage}{0.216\textwidth}{
			\includegraphics[width=1\textwidth]{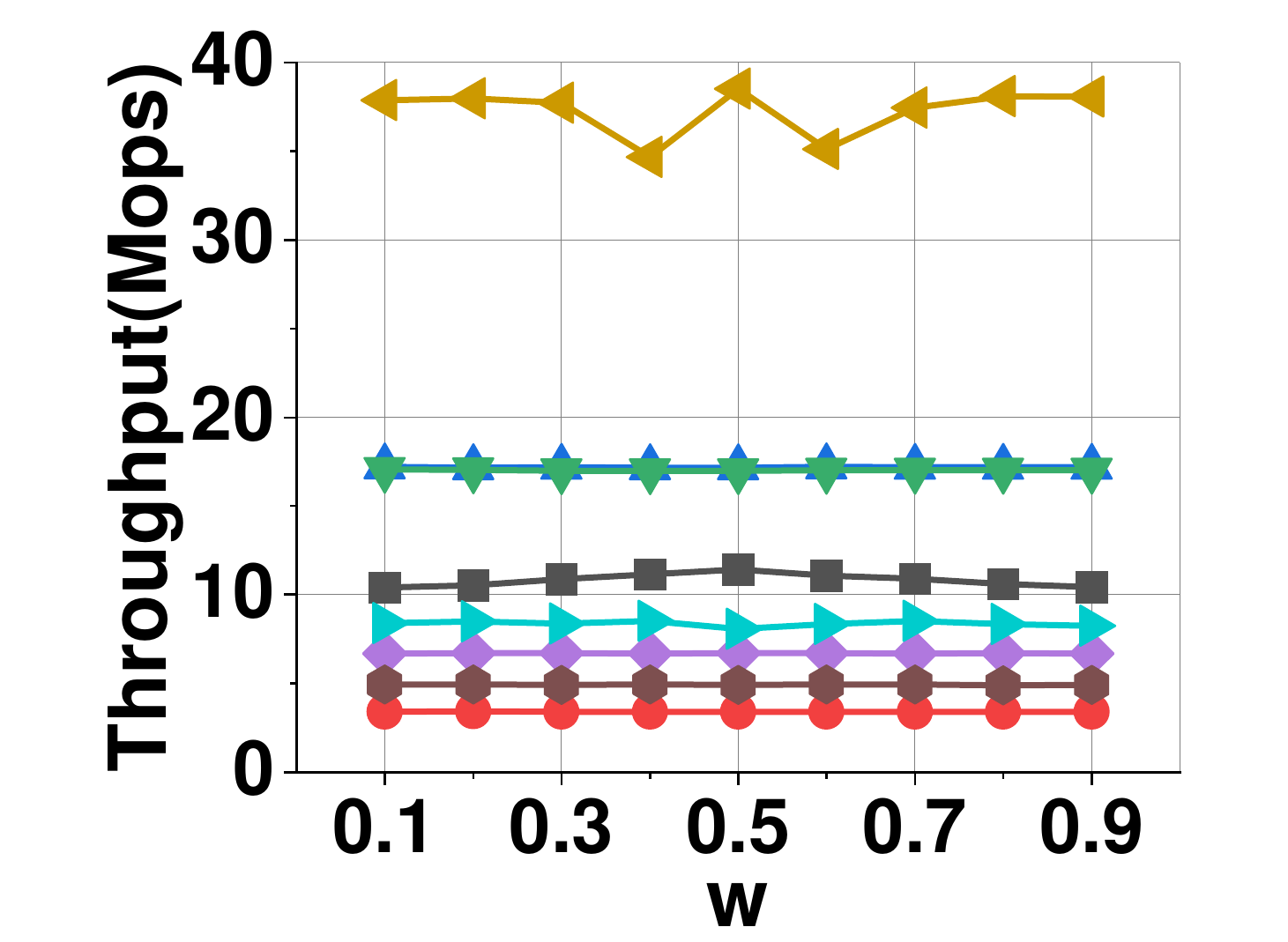}
            \vspace{-0.05in}
		}
	\end{minipage}}
    \subfigure[Campus (1500KB)]{
	\begin{minipage}{0.216\textwidth}{
			\includegraphics[width=1\textwidth]{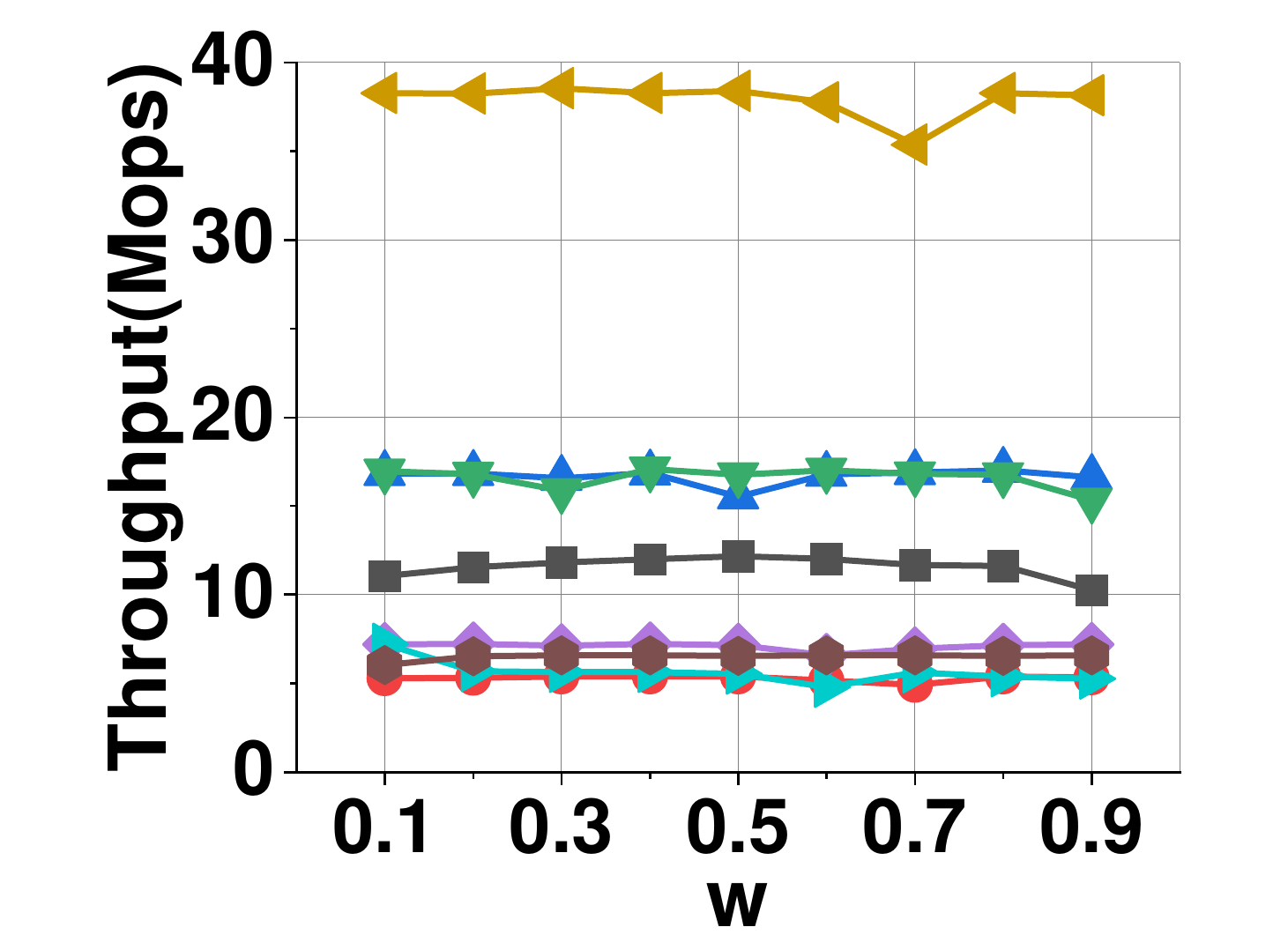}
            \vspace{-0.05in}
		}
	\end{minipage}}
    \subfigure[Seattle (80KB)]{
	\begin{minipage}{0.216\textwidth}{
			\includegraphics[width=1\textwidth]{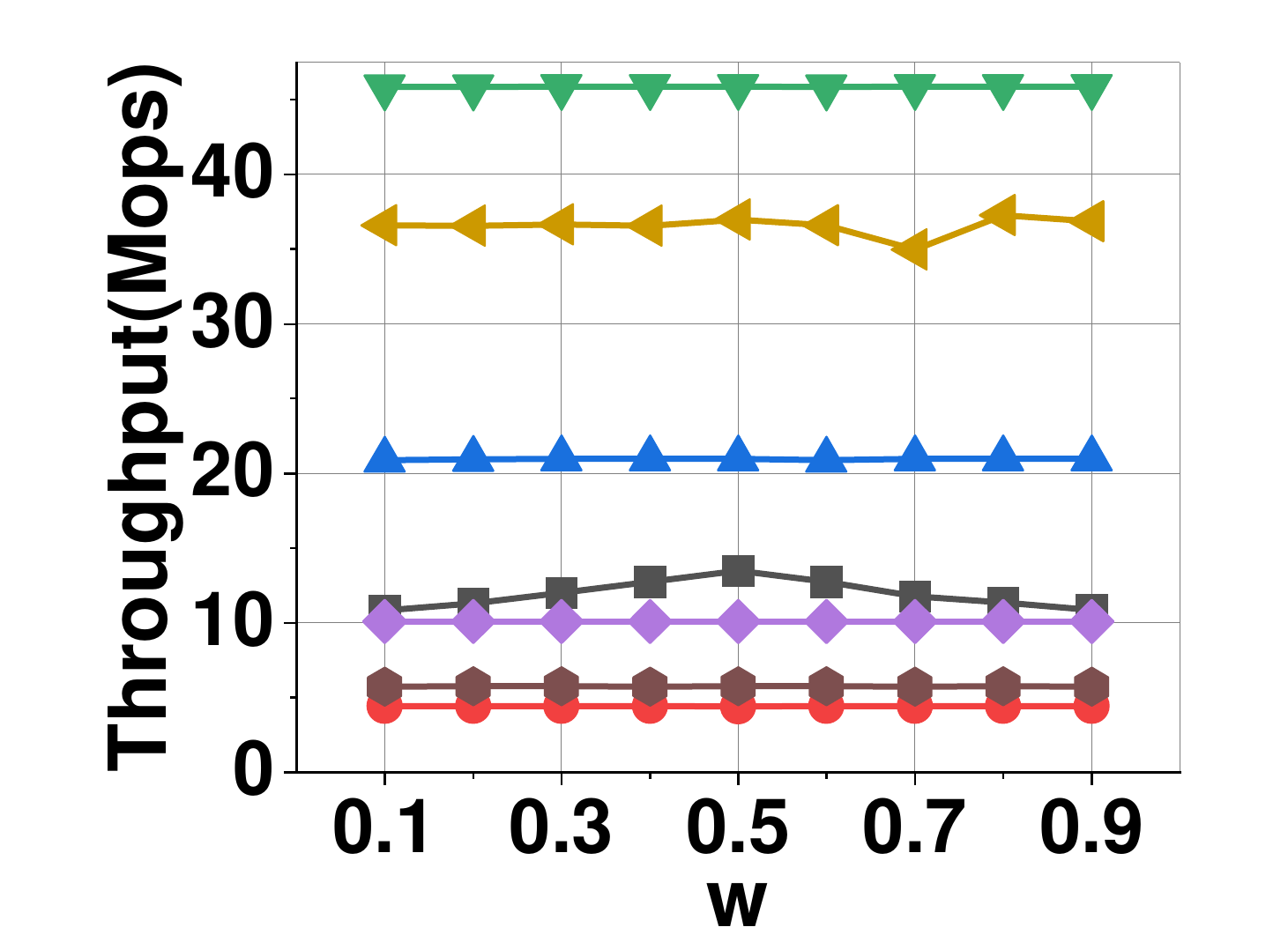}
            \vspace{-0.05in}
		}
	\end{minipage}}
    \subfigure[Synthetic (80KB)]{
	\begin{minipage}{0.216\textwidth}{
			\includegraphics[width=1\textwidth]{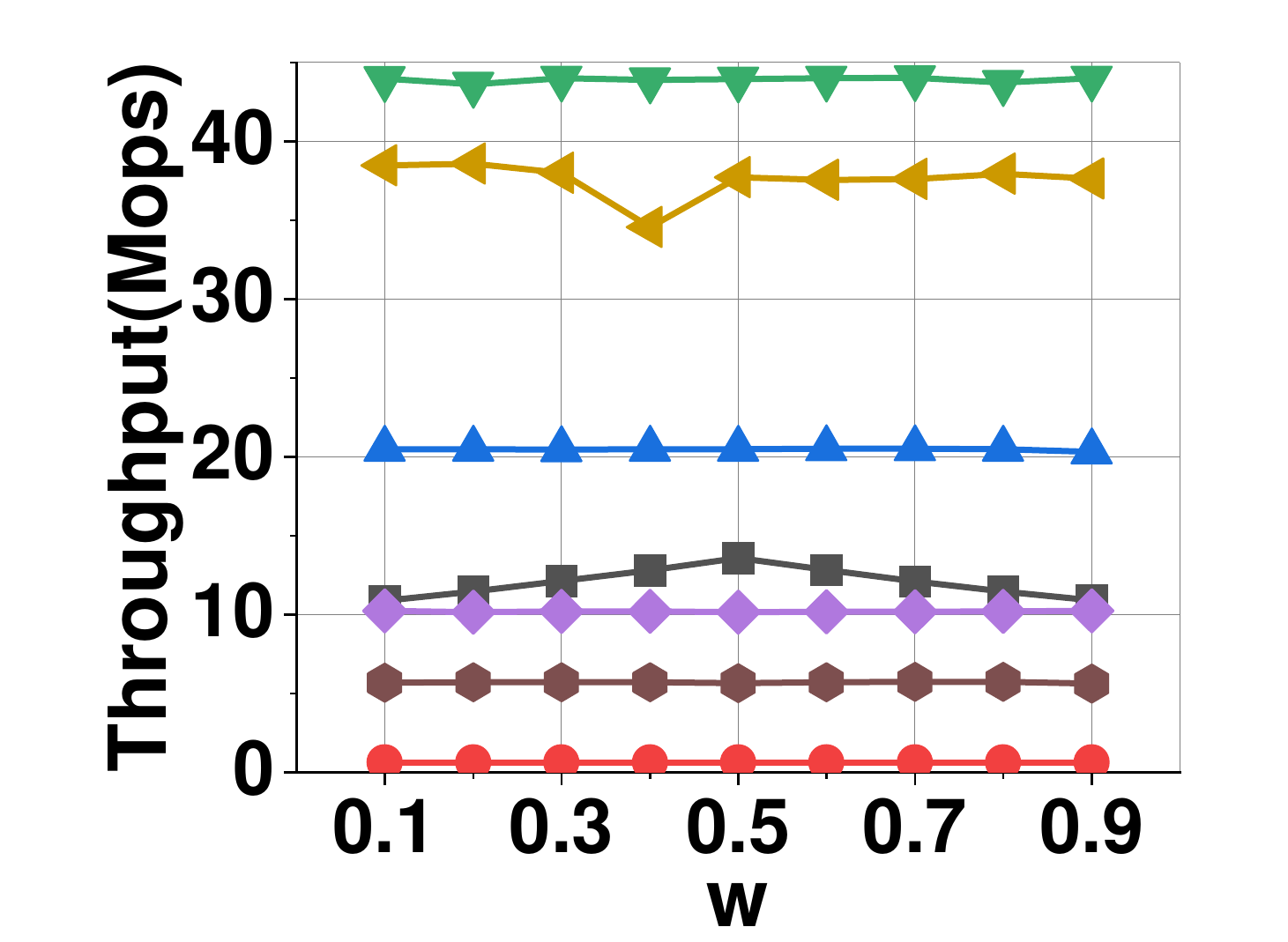}
            \vspace{-0.05in}
		}
	\end{minipage}}
    \caption{Insertion Throughput on Per-key Situation}
    \vspace{-0.05in}
    \label{fig::per::thp}
\end{figure*}
\begin{figure*}[!ht]
	\centering
   \vspace{-0.05in}
    \begin{minipage}{.8\textwidth}
        \includegraphics[width=1\textwidth]{Exp_new/8legend.pdf}
    \end{minipage} 
    \vspace{-0.1in}
    \\
    \subfigure[CAIDA (1500KB)]{
	\begin{minipage}{0.216\textwidth}{
			\includegraphics[width=1\textwidth]{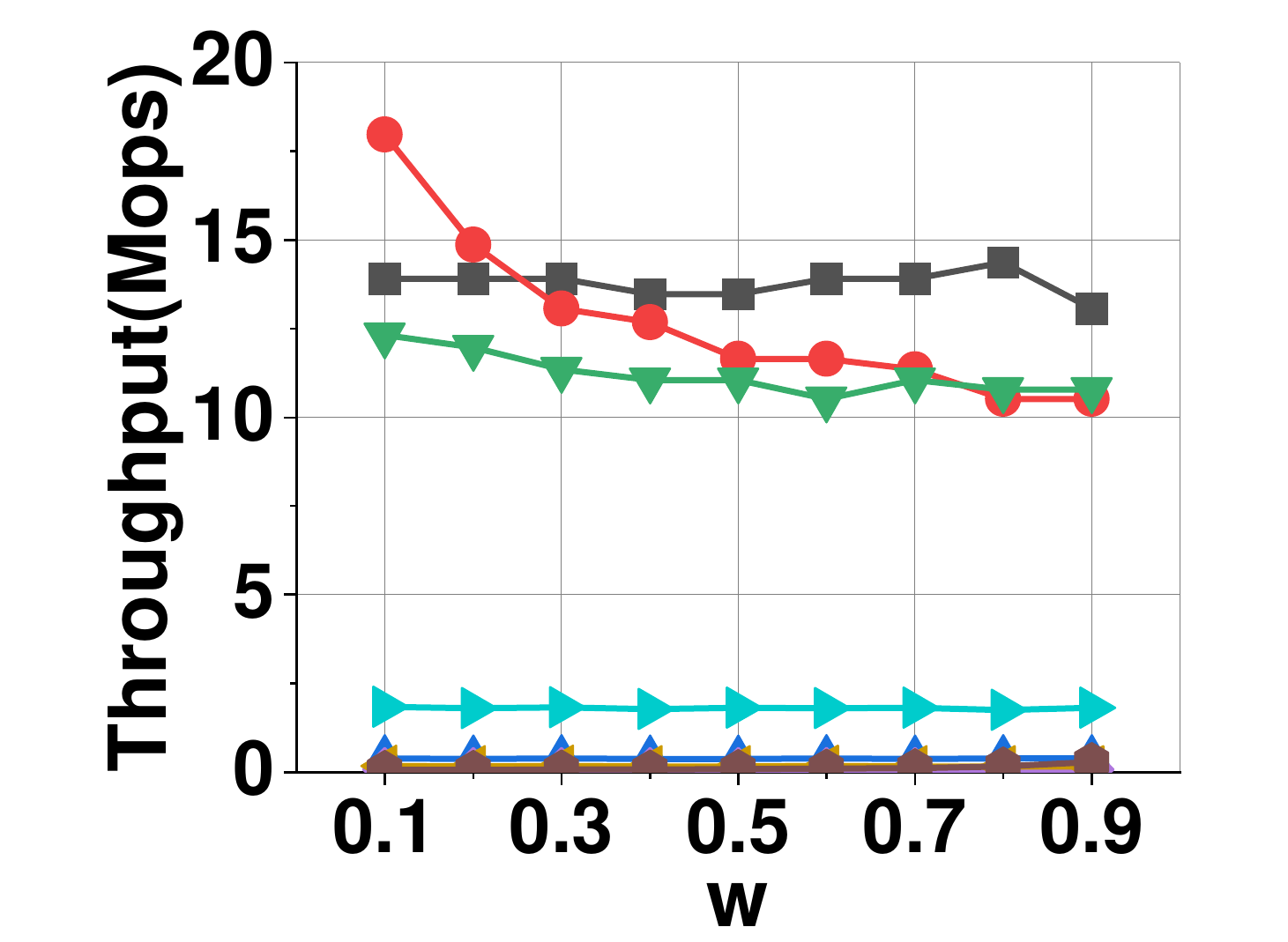}
            \vspace{-0.05in}
		}
	\end{minipage}}
    \subfigure[Campus (1500KB)]{
	\begin{minipage}{0.216\textwidth}{
			\includegraphics[width=1\textwidth]{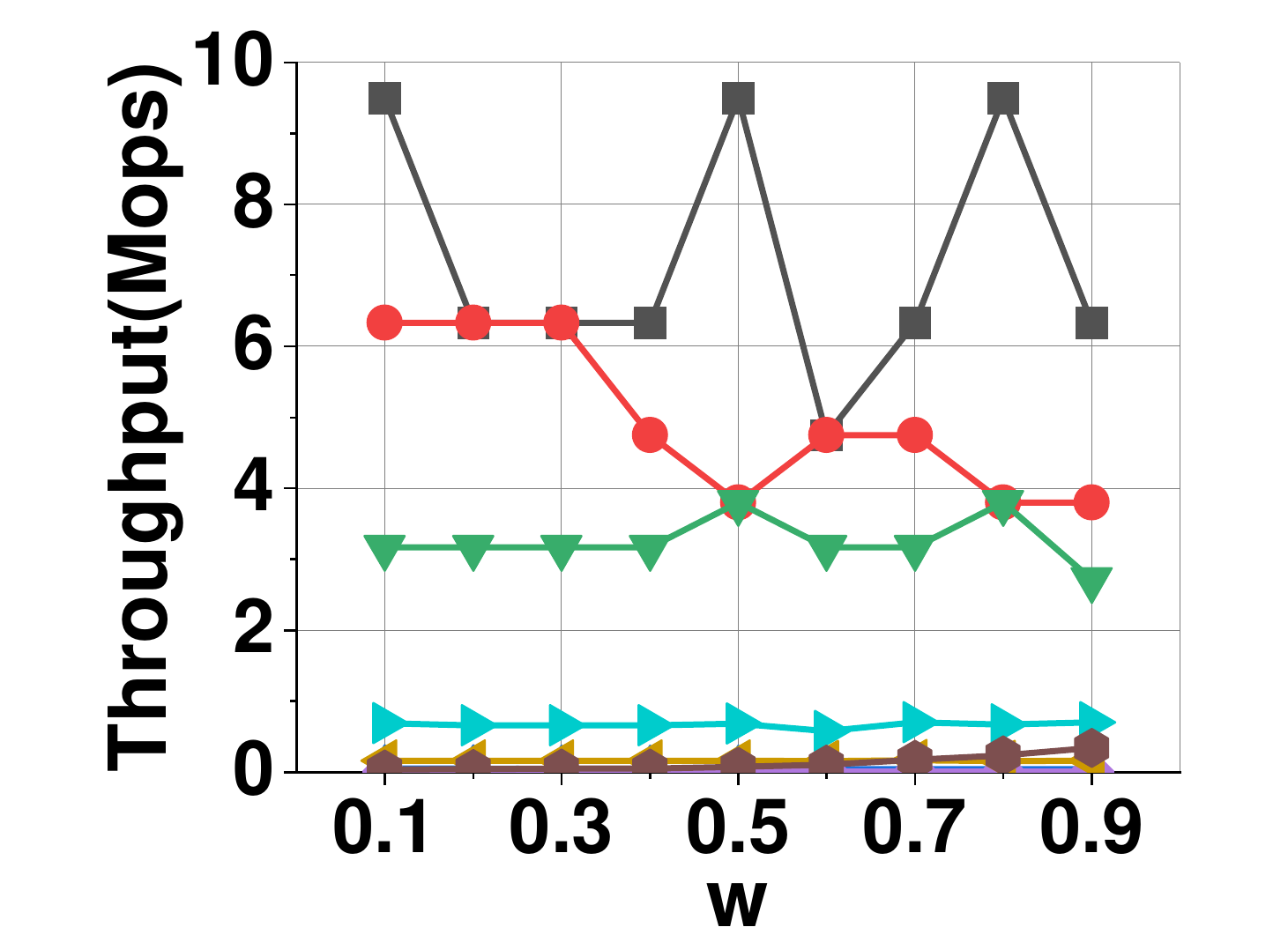}
            \vspace{-0.05in}
		}
	\end{minipage}}
    \subfigure[Seattle (80KB)]{
	\begin{minipage}{0.216\textwidth}{
			\includegraphics[width=1\textwidth]{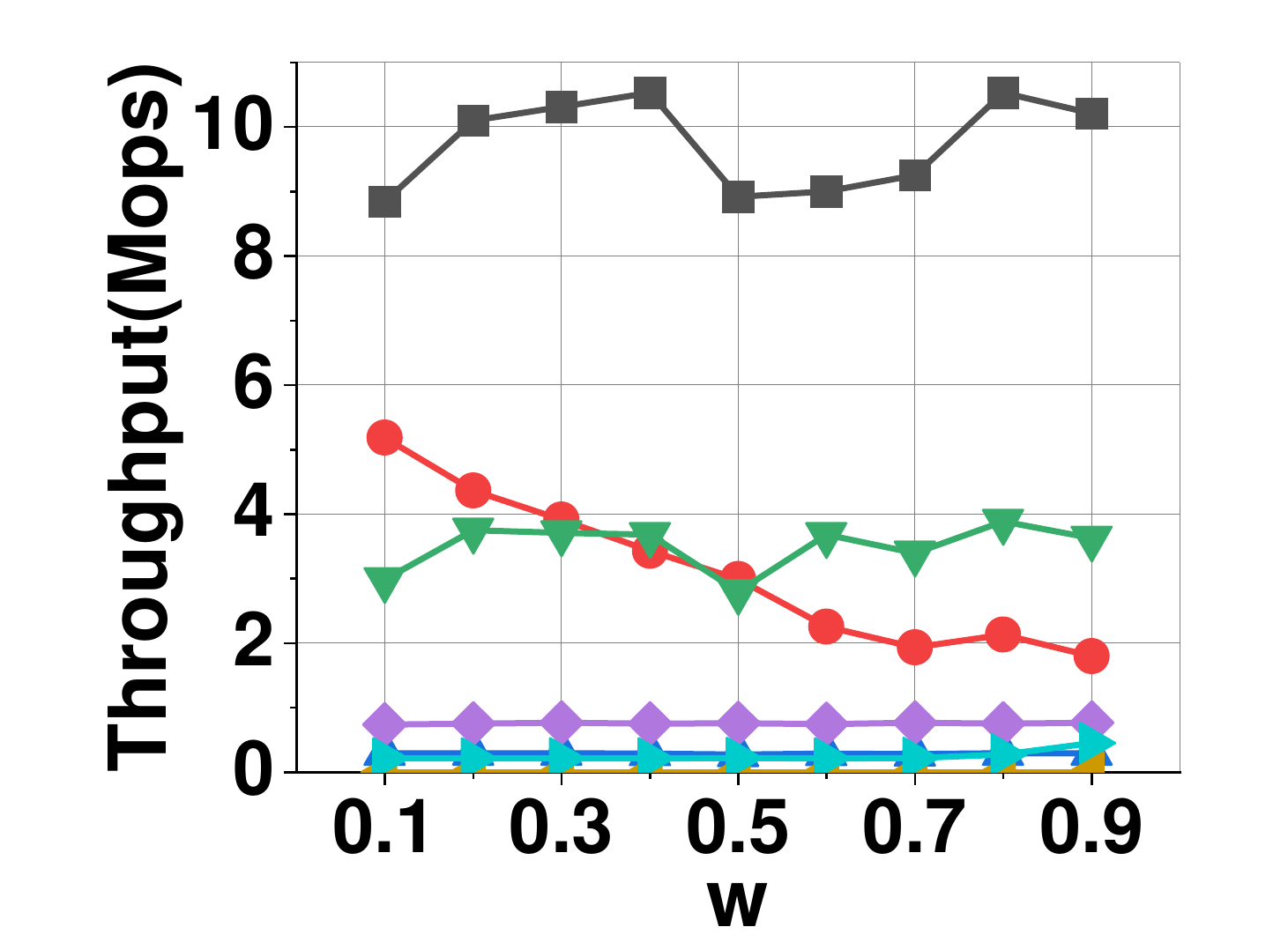}
            \vspace{-0.05in}
		}
	\end{minipage}}
    \subfigure[Synthetic (80KB)]{
	\begin{minipage}{0.216\textwidth}{
			\includegraphics[width=1\textwidth]{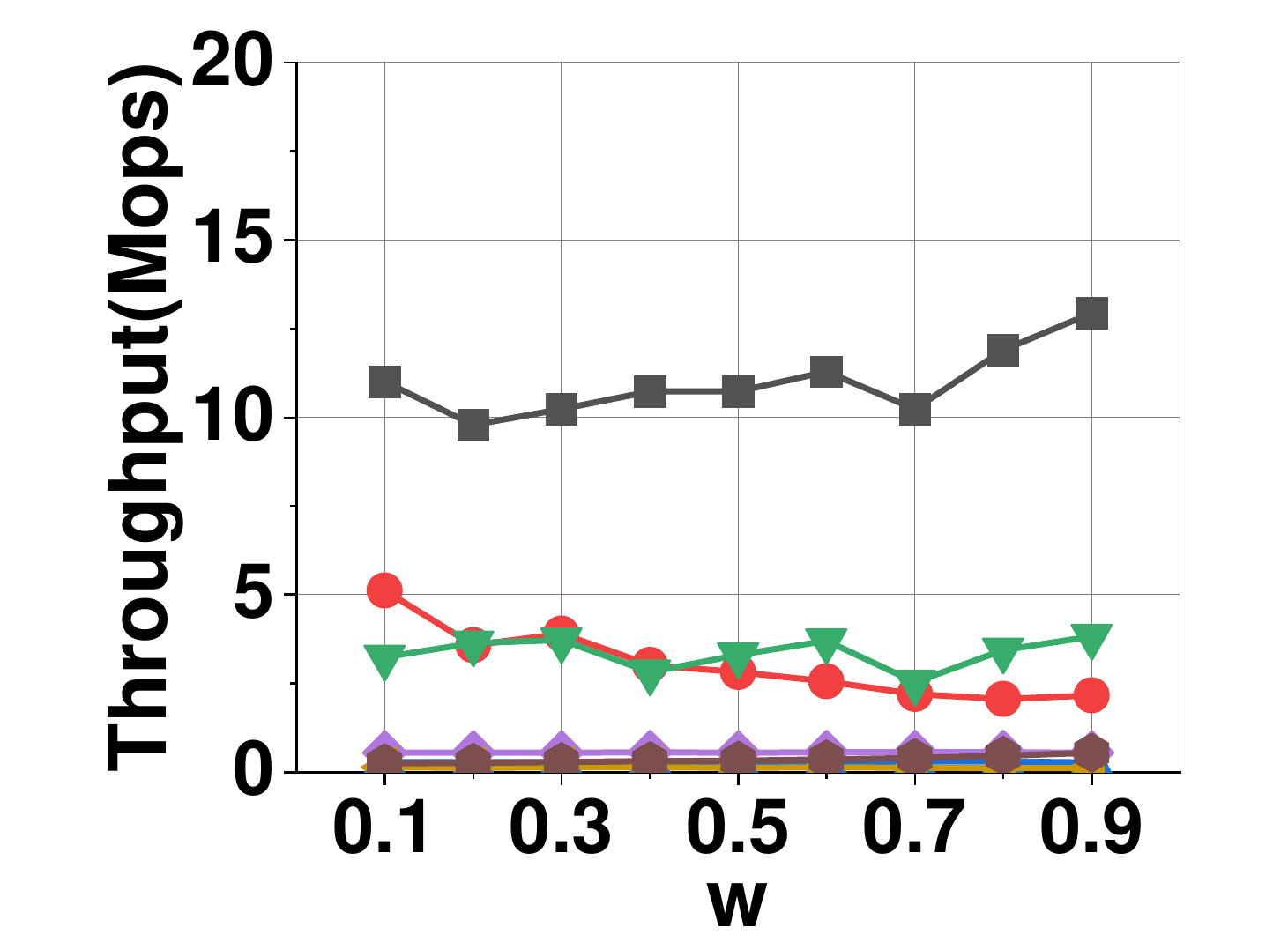}
            \vspace{-0.05in}
		}
	\end{minipage}}
    \caption{Query Throughput on Per-key Situation}
    \vspace{-0.15in}
    \label{fig::per::thp}
\end{figure*}

\subsection{Experiments in Single-key Situation}
\label{exp::single}

In this section, we compare the performance of \ourname{} with prior work (GK, KLL, DDSketch, $t$-digest, ReqSketch) on CAIDA dataset and Synthetic dataset in single-key situation. 
To support single-key insertion and query, we ignore the key of all items in the data stream and we set the total memory to 10KB for every algorithm. \textit{Experimental results \emph{(Figure \ref{fig::single::ae}-\ref{fig::single::thp})} show that \ourname{} is as good as other algorithms in terms of AE, but \ourname{} runs much faster.} 
The AE of \ourname{} is far less than that of $t$-digest and ReqSketch, but it's not much different from other algorithms. 
The insertion throughput of \ourname{} is around 20M item per second, which is around 2 times higher than $t$-digest, 7 times higher than GK, 15 times higher than DDSketch and only lower than KLL and ReqSketch. Moreover, the insertion throughput of \ourname{} peaks when $w=0.5$, as we have to insert positive (negative) infinities by probability method when $w\neq 0.5$, and the larger $|w-0.5|$ is, the more infinities \ourname{} has to insert, which results in the lower throughput.

\subsection{Experiments in Per-key Situation}
\label{exp::per}

In this section, we compare the performance of \ourname{} with baseline solution (GK, KLL, DDSketch, $t$-digest, ReqSketch), SQUAD\footnote{Experimental results show that SQUAD has to set $\varepsilon=1$ within 900KB on CAIDA dataset, and it cannot work within such small memory. As a result, we only provide the results of SQUAD with 1100, 1300 and 1500KB memory. } and SketchPolymer in per-key situation\footnote{Experimental results show that M4 cannot run within tight memory. With similar problem settings, M4 uses at least 6MB to run for CAIDA dataset in \cite{dongm4}, so we do not list the results of M4 here. }. 
\textit{Experimental results \emph{(Figure \ref{fig::per::caida}-\ref{fig::per::thp})} show that \ourname{} performs much better than state-of-the-art in per-key situation.} 
The AE of \ourname{} is all lower than 0.06 on CAIDA dataset, and lower than 0.03 on Synthetic dataset, which is approximately half of other algorithms. 
Experimental results also show that \ourname{} works especially well in small memory cases, which SQUAD and SketchPolymer fails to achieve. 
As to throughput, \ourname{} achieves balance betweeen insertion throughput and query throughput. Both insertion and query throughput of \ourname{} are close to 10 Mops on these datasets. 
On CAIDA dataset, the insertion throughput of GK and SketchPolymer is smaller than 5 Mops, so they cannot catch up with high-speed items in data streams; the query throughput of KLL, $t$-digest, ReqSketch and SQUAD is smaller than 2 Mops, so they have to spend a lot of time to answer quantile query.

\begin{figure}[!ht]
	\centering
	\vspace{-0.05in}
	\subfigure[$w=0.5$]{
	\begin{minipage}[t]{0.22\textwidth}{
			\includegraphics[width=1\textwidth]{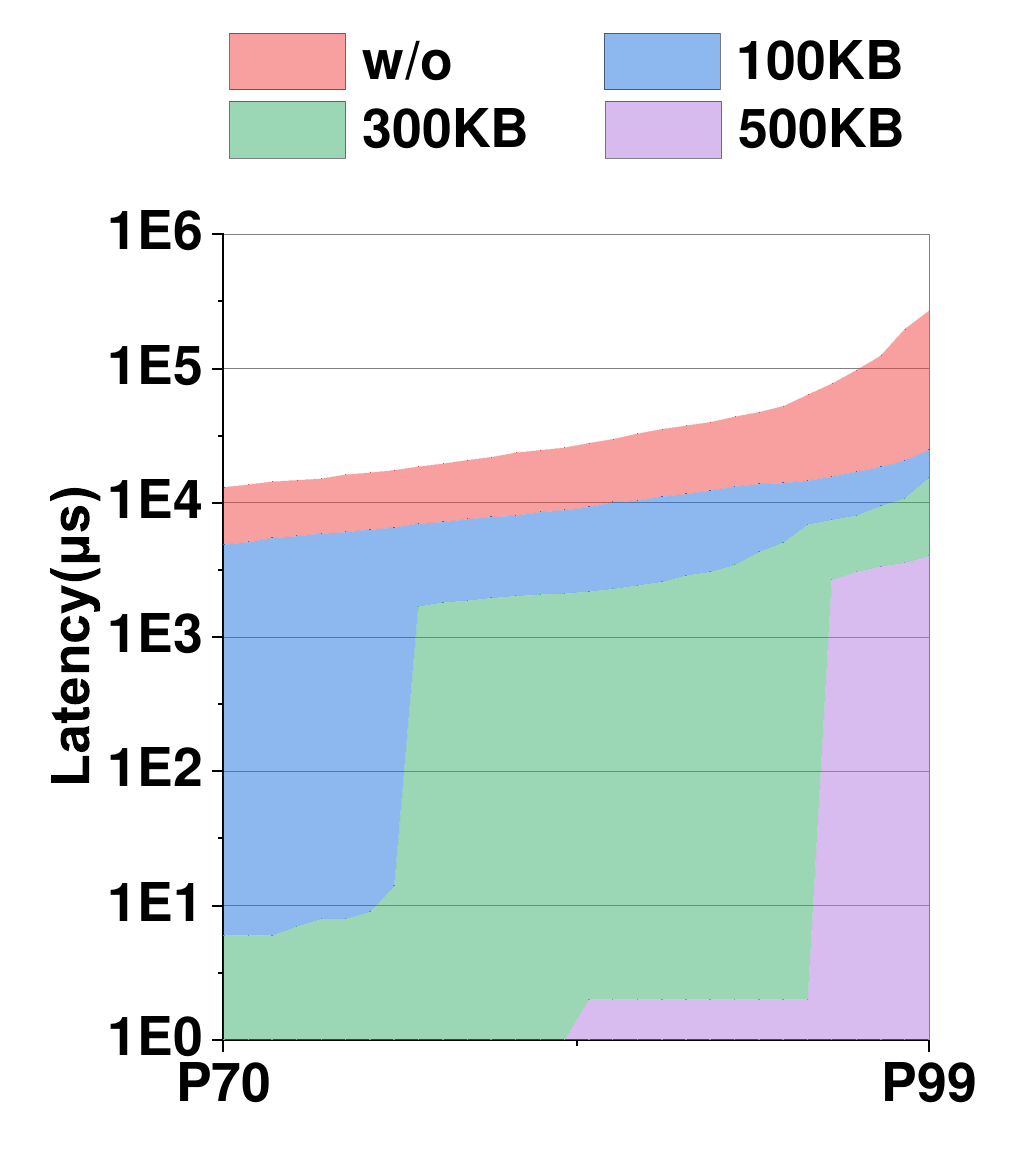}}
	\end{minipage}}
	\subfigure[$w=0.9$]{
	\begin{minipage}[t]{0.22\textwidth}{
		\includegraphics[width=1\textwidth]{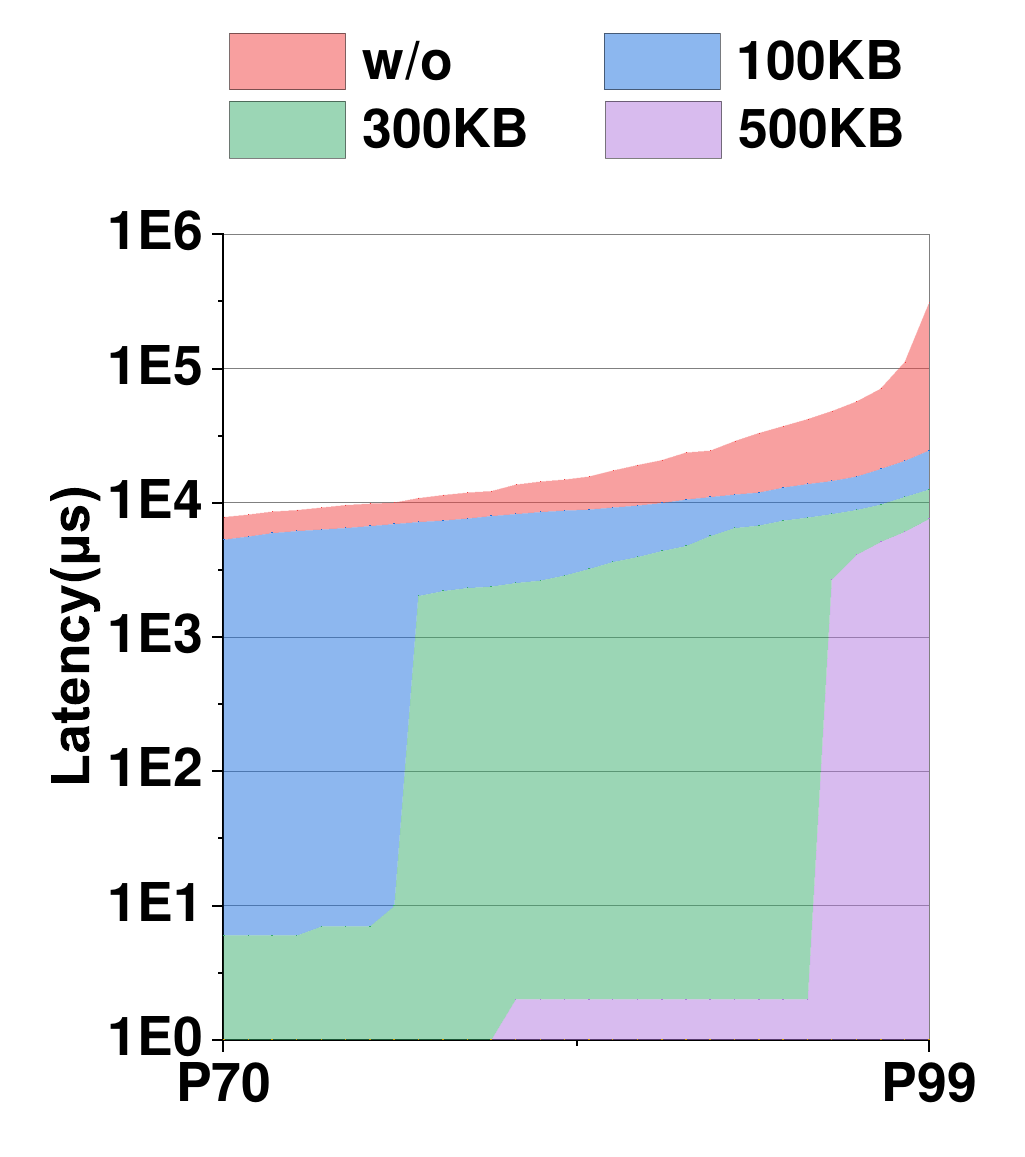}}
	\end{minipage}}
	\caption{Experiments on RocksDB database}
	\label{fig::rocksdb}
	\vspace{-0.2in}
\end{figure}

\subsection{Experiments on RocksDB Database}
\label{exp::rocksdb}

In this section, we implement \ourname{} on RocksDB database \cite{rocksdb} to accelerate quantile estimation in database scenarios. We create a synthetic RocksDB benchmark to measure the running time of quantile estimation with and without \ourname{}. We generate 1M items with 64-bit key and 64-bit value following Zipfian distribution with $\alpha=1$. For items with the same key, we add distinct 16-bit suffix to the key to distinguish them, so RocksDB database will not report two records with the same key. For simplicity, we set $w=0.5$ and $w=0.9$ and compare the following two schemes on the running time of quantile estimation:

\textbf{1) \ourname{}:} For every item to be inserted into RocksDB database, we insert it into \ourname{} as well. When query the $w$-quantile of values with key $k$, we first use $k$ to query \ourname{}. If the information of key $k$ is stored in \ourname{}, then we directly return the query result of \ourname{} as its query result; otherwise we select all values with key $k$ from the database and calculate their $w$-quantile. 

\textbf{2) Baseline Solution:} No data structure in RocksDB database is specifically designed for quantile estimation. Consequently, to answer quantile query for key $k$, the baseline solution has to select all items in the database with key $k$.

We measure the tail latency of quantile query for baseline solution and different size of \ourname{} in Figure \ref{fig::rocksdb}. \textit{Our experimental results show that allocating small memory for \ourname{} can significantly accelerate quantile query speed in RocksDB database. }  Within 100KB memory, \ourname{} can reduce the P90 latency by $3.23/2.21\times$ compared to baseline solution. As the allocated memory increases, \ourname{} can store more frequent items, and the acceleration effect is more pronounced. Within 500KB memory, the P99 latency is reduced by 67.06/35.51$\times$ and is smaller than 8ms, so \ourname{} supports fast quantile query request in real database scenarios.

\section{Discussion}

\textbf{Question:} To apply Distribution Calibration, \ourname{} assumes that all values in the data stream are i.i.d. Is this acceptable in practice? 

\textbf{Answer:} In reality, values in datasets do not strictly follow i.i.d. distribution; this assumption serves more as an approximation. Yet, this assumption is common and reasonable in both scientific research and real datasets, with which we can give simpler and more elegant results.

In theoretical computer science, most prior work on quantile estimation applies the classical comparison-based model, where the algorithm is only allowed to compare items via the total ordering on the universe \cite{greenwald2001space, karnin2016optimal, cormode2021relative, gribelyuk2024simple}. 
These algorithms do not rely on prior knowledge of data distribution, and they provide optimal space complexity for this problem. 
However, we would like to point out that, in real database experiments, programmers shall consider more than theory. A simpler and more straightforward algorithm is preferred in implementation, maintenance, reuse, and scalability. 

Our example is two popular quantile estimation algorithms, DDSketch \cite{ddsketch} and $t$-digest \cite{tdigest}. The theory of DDSketch not only assumes that values of all items are i.i.d., but also supposes values follow the Pareto distribution. $t$-digest also assumes that values follow i.i.d. distribution, and $t$-digest does not explicitly provide its space complexity. 
Both DDSketch and $t$-digest can be problematic when the distribution is not i.i.d.: under certain input, these algorithms fail to achieve accurate quantile estimation \cite{cormode2021theory}. Nonetheless, DDSketch and $t$-digest are still widely used in databases and actually works well \cite{li2022microsketch, dunning2021t}, and this is why we collect data to do empirically experiments in addition to giving theoretical guarantees. 

In this paper, we run \ourname{} on four datasets, three of which are real-world datasets, and experimental results show that values in these datasets can be approximated as i.i.d., so \ourname{} still performs well in practice. 
Finally, to tackle non-i.i.d. data stream, we can cut the data stream into several intervals, and values in each interval can be approximated as i.i.d. 
The results from different intervals can be aggregated to achieve approximate quantile estimation. 
\section{Conclusion}
\label{sec:conclusion}

In this paper, we propose \ourname{} for point-quantile estimation. In single-key situation, \ourname{} applies \textbf{Value Focus} to keep values close to median when $w=0.5$, and uses \textbf{Distribution Calibration} to transfer arbitrary quantile estimation problem into median estimation problem when $w\neq 0.5$. Moreover, \ourname{} applies \textbf{Double Filtration} to filter infrequent items in both stages to cater for per-key quantile estimation. We present rigorous mathematical analysis for \ourname{}, and experimental results show that \ourname{} outperforms existing algorithms in different scenarios. 
We have released the source codes of \ourname{} on GitHub \cite{source}.

\newpage
\clearpage
{
\bibliographystyle{IEEEtran}
\bibliography{reference}
}

\newpage
\clearpage
\appendices
\section{Pseudo-code of \ourname{}}
\label{appendix:code}

\subsection{Single-key Situation when $w=0.5$}
\label{appendix:single}

The pseudo-code of insertion and query operation in single-key situation when $w= 0.5$ is shown in Algorithm \ref{alg:single}.

\begin{algorithm}
	\caption{Single-key Situation when $w= 0.5$}
	\label{alg:single}
        \textit{// initialize the Candidate and Representative as empty} \\
        $C, R \leftarrow \varnothing$\; 
        \textit{// Insertion Procedure} \\
	\For{every item $e=(k, v)$ in the data stream}
        {
        $C \leftarrow C \cup \{v\}$\;
        \If{$|C| = r$} 
        {
        \textit{// the Candidate is full} \\
        $v', v'' \leftarrow \text{median of $C$}$\;
        $C \leftarrow \varnothing$; $R \leftarrow R \cup \{v', v''\}$\;
        \If{$|R| > s$}
        {
        \textit{// the Representative is full} \\
        $v_M \leftarrow \max \limits _{v\in R} v$; $v_m \leftarrow \min \limits _{v\in R} v$\;
        \textit{// keep the size of $R$ unchanged} \\
        $R \leftarrow R \backslash \{v_M, v_m\}$\;
        }
        }
        }
        \textit{// Query Procedure} \\
        \Return{0.5-quantile of $R$}
\end{algorithm}

\subsection{\stagetwo{} Insertion Procedure}
\label{appendix:per}
The pseudo-code of insertion operation of \stagetwo{} is shown in Algorithm \ref{alg:per:insert}. 

\begin{algorithm}
	\caption{\stagetwo{} Insertion Procedure}
	\label{alg:per:insert}
        \KwIn{an item $e=(k, v)$}
        \If{$\exists 1\leq j\leq d$, \st{} $B[h(k)][j].key = k$}
        {
        $B[h(k)][j].vote^+ \leftarrow B[h(k)][j].vote^+ + 1$\;
        $B[h(k)][j].value.\texttt{insert($v$)}$\;
        }
        \ElseIf{$\exists 1\leq j\leq d$, \st{} $B[h(k)][j]$ is empty}
        {
        $B[h(k)][j].key \leftarrow k$\;
        $B[h(k)][j].vote^+ \leftarrow 1$\;
        $B[h(k)][j].value.\texttt{insert($v$)}$\;
        }
        \Else
        {
        $B[h(k)].vote^- \leftarrow B[h(k)].vote^- + 1$\;
        \textit{// find the cell with minimum $vote^+$ in $B[h(k)]$} \\
        $j \leftarrow \mathop{\arg\min}\limits _{1\leq i\leq d} B[h(k)][i].vote^+$\;
        \If{$B[h(k)].vote^-\geq \lambda \times B[h(k)][j].vote^+$}
        {
        \textit{// clear the cell and insert $e$} \\
        $B[h(k)][j].\texttt{clear()}$\;
        $B[h(k)][j].key \leftarrow k$\;
        $B[h(k)][j].vote^+\leftarrow 1$\;
        $B[h(k)][j].value.\texttt{insert($v$)}$\;
        $B[h(k)].vote^-\leftarrow 0$\;
        }
        }
\end{algorithm}

\subsection{\ourname{} Insertion Procedure}
\label{appendix}
The pseudo-code of insertion operation of \ourname{} is shown in Algorithm \ref{alg:insert}.

\begin{algorithm}
	\caption{\ourname{} Insertion Procedure}
	\label{alg:insert}
        \KwIn{an item $e=(k, v)$}
        \If{$\texttt{TowerSketch.query}(k) < T$}
        {
        $\texttt{TowerSketch.insert}(k)$\;
        \Return \;
        }
        \Else
        {
        $\texttt{ValueSketch.insert}(k, v)$\;
        }
\end{algorithm}

\end{document}